\documentclass[reprint,onecolumn,aps,prd,10pt,nofootinbib,superscriptaddress,showkeys,notitlepage]{revtex4-1}

\usepackage[latin9]{inputenc}
\usepackage{verbatim}
\usepackage{amsthm,mathtools}
\usepackage{mathrsfs} 
\usepackage{amsmath}
\usepackage{amsfonts}
\usepackage{amssymb}
\usepackage[mathscr]{euscript}

\usepackage{bbold}
\usepackage[dvipsnames]{xcolor}
\usepackage{float}
\usepackage[section]{placeins}
\usepackage[lofdepth,lotdepth]{subfig}
\usepackage[toc,page]{appendix}
\usepackage{xtab}
\usepackage{enumitem}
\usepackage[hhmmss]{datetime}
\usepackage{prettyref}
\usepackage{booktabs}
\usepackage[unicode=true,
 bookmarks=false,
 breaklinks=false,pdfborder={0 0 1},backref=false,colorlinks=true]
 {hyperref}

\makeatletter

\usepackage{tikz}
\usetikzlibrary{shapes,arrows}

\usepackage{booktabs}
\allowdisplaybreaks

\pdfpageheight\paperheight
\pdfpagewidth\paperwidth

\definecolor{darkBlue}{rgb}{0.15,0.15,0.65}
\definecolor{lightRed}{rgb}{1,0.7,0.8}

\hypersetup{
  linktocpage=true,
  citecolor=black,
  linkcolor=black,
  urlcolor=black,
  breaklinks=true
}

\allowdisplaybreaks

\usepackage{eso-pic, rotating, graphicx}

\newcommand{\bSe}{\begin{subequations}}
\newcommand{\eSe}{\end{subequations}}

\makeatother

\theoremstyle{definition}

\theoremstyle{definition}
\newtheorem*{defnw}{Definition}

\theoremstyle{plain}
\newtheorem{prop}{Proposition}

\theoremstyle{plain}

\theoremstyle{remark}

\theoremstyle{remark}

\newcommand{\eqnref}[1]{eq.\;\eqref{#1}}  % Use for one eq.
\newcommand{\eqsref}[1]{eqs.\;\eqref{#1}}  % Use for one or more eqs.

\newcommand{\eH}{\mathrm{\scriptscriptstyle H}}
\newcommand{\eC}{\mathrm{\scriptscriptstyle C}}
\newcommand{\eS}[1]{\mathrm{\scriptscriptstyle #1}}

\global\long\def\ee{\mathrm{e}}

\global\long\def\tr{\mathsf{{\scriptscriptstyle T}}}

\global\long\def\dd{\mathrm{d}}

\global\long\def\op#1{\operatorname{#1}}
\global\long\def\Tr{\operatorname{Tr}}

\usepackage{tensind}
\tensordelimiter{?}

\global\long\def\tud#1#2#3{?{\mbox{\ensuremath{#1}}}^{#2}{}_{#3}?}

\newcommand{\colSep}[1]{\setlength{\arraycolsep}{#1}}

\newcommand{\phasePlot}[5]{
	\begin{tikzpicture}[x=1mm,y=1mm]
	\node[inner sep=1mm] at (0,0) {
		\includegraphics[width=56mm]{#1.pdf}
	};
	\node [anchor=north] at (  0,-31.0) {\small #2};
	\node [anchor=north] at (-26,-27.5) {\small $-$#3};
	\node [anchor=north] at (  0,-27.5) {\small $0$};
	\node [anchor=north] at (+26,-27.5) {\small #3};
	\node [anchor=east]  at (-32.0,    0) {\small #4};
	\node [anchor=east]  at (-27.5,+27.5) {\small #5};
	\node [anchor=east]  at (-27.5,    0) {\small $0$};
	\node [anchor=east]  at (-27.5,-27.5) {\small $-$#5};
	\end{tikzpicture}
}

\hypersetup{hidelinks}

\makeatletter

\AtBeginDocument{%
	\expandafter\renewcommand\expandafter\subsection\expandafter{%
		\expandafter\@fb@secFB\subsection
	}%
}

\makeatother

\begin{document}

\renewcommand{\theenumi}{(\roman{enumi})}% Lower-case roman letters lists

\title{Classification and asymptotic structure of black holes in bimetric theory}

\author{Francesco \surname{Torsello}}
\email{francesco.torsello@fysik.su.se}

\author{Mikica \surname{Kocic}}
\email{mikica.kocic@fysik.su.se}

\author{Edvard \surname{M\"ortsell}}
\email{edvard@fysik.su.se}

\affiliation{Department of Physics \& The Oskar Klein Centre, \\
Stockholm University, AlbaNova University Centre, SE-106 91 Stockholm, Sweden}

\begin{abstract}
We study general properties of static and spherically symmetric bidiagonal black holes in Hassan-Rosen bimetric theory by means of a new method. In particular, we explore the behaviour of the black hole solutions both at the common Killing horizon and at the large radii. The former study was never done before and leads to a new classification for black holes within the bidiagonal ansatz. The latter study shows that, among the great variety of the black hole solutions, the only solutions converging to Minkowski, Anti-de Sitter and de Sitter spacetimes at large radii are those of General Relativity, i.e., the Schwarzschild, Schwarzschild-Anti-de Sitter and Schwarzschild-de Sitter solutions. Moreover, we present a proposition, whose validity is not limited to black hole solutions, which establishes the relation between the curvature singularities of the two metrics and the invertibility of their interaction potential.
\end{abstract}

\keywords{Classical Theories of Gravity, Modified Theories of Gravity, Bimetric Theory, Black Holes}

%\arxivnumber{...}

%%%%%%%%%%%%%%%%%%%%%%%%%%%%%%%%%%%%%%%%%%%%%%%%%%%%%%%%%%%%%%%%%%%%%% End Title

\maketitle
\flushbottom

\clearpage{}

\tableofcontents

\section{Introduction}
\label{sec:introduction}

General Relativity (GR) is a remarkably successful theory of gravitational interactions. It has passed many experimental tests so far both in the weak- and the strong-field regime \cite{lrr-2014-4,PhysRevLett.116.061102}. Yet, GR cosmology cannot give a satisfactory explanation for the physical origin of the accelerated expansion of the Universe \cite{Ade:2015xua}. This issue contributes to the so-called ``cosmological constant problem'' \cite{RevModPhys.61.1}, one of the most intriguing open problems in modern physics. It provides a strong motivation for studying alternative theories of gravity, which could shed light on the physical cause of the acceleration.

Moreover, GR cosmology needs the introduction of an unknown matter component, called the ``dark matter", which is required for explaining many astrophysical and cosmological observations and has well established theoretical motivations \cite{Bergstrom:2000pn,Bertone:2004pz,Feng:2010gw}. Despite the effort that has been dedicated to revealing the dark matter, its origin is still unknown. So far, only gravitational interactions between the dark matter and the Standard Model sector have been detected; this suggests that the physical explanation for the origin of the dark matter can still be found within the realm of modified gravity. Some examples of this research front can be found in \cite{Aoki:2014cla,Blanchet:2015bia,Enander:2015kda,Babichev:2016bxi}.

It is well-known that GR is the unique, nonlinear theory describing the self-interaction of a massless spin-2 tensor field (the ``graviton'', if quantised) and its interaction with non-gravitational fields \cite{Lovelock:1971yv,Schwartz:2013pla}. Therefore, a natural way of generalising it is to add a massive spin-2 tensor field to the theory. The first attempt to describe massive spin-2 fields was in 1939 when Fierz and Pauli formulated a linear theory describing the dynamics of a massive spin-2 field \cite{Fierz:1939,Fierz211}. In 1972, Boulware and Deser claimed that any nonlinear generalisation of the Fierz-Pauli theory has to suffer from a ghost mode, the Boulware-Deser (BD) ghost \cite{Boulware1972227,Boulware:1973my}. However, more recently a ghost-free nonlinear theory of a massive spin-2 field --- de Rham-Gabadadze-Tolley (dRGT) massive gravity --- was proposed \cite{deRham:2010ik,deRham:2010kj}. The dRGT theory involves two rank-2 tensors (usually called metrics, for the sake of familiarity), one dynamical and the other a non-dynamical reference metric. This model has five propagating modes associated with the massive spin-2 field. A generalisation of the dRGT model to an arbitrary reference metric was given in \cite{Hassan:2011tf}. That the dRGT massive gravity indeed does not suffer from the BD ghost was proven in \cite{Hassan:2011hr}.

A further generalisation of the dRGT massive gravity was done in \cite{Hassan:2011zd} by giving dynamics to the reference metric. For this model, the absence of the BD ghost was proven in \cite{Hassan:2011ea}. The resulting ghost-free, nonlinear theory of two dynamical interacting spin-2 tensor fields --- Hassan-Rosen (HR) bimetric theory --- describes the interaction between a massless and a massive spin-2 field, together with their self-interactions and interactions with non-gravitational fields. This theory contains seven propagating modes, two of them associated with the massless spin-2 field and five of them with the massive spin-2 field.

Having constructed these consistent theories of gravity, their viability as alternative theories to GR should be tested by studying their phenomenology. In this respect, the effort has been dedicated to finding black hole (BH) solutions both in dRGT massive gravity and in HR bimetric theory.

The simplest solutions one can consider are static and spherically symmetric. In GR, the no-hair theorems \cite{Israel:1967wq,Israel:1967za,Carter:1971zc} guarantee that a static and spherically symmetric BH without electrical charge is described by only one constant parameter, its mass. The validity of these theorems is not proved in dRGT massive gravity and HR bimetric theory. Then, it is natural to ask whether there are asymptotically flat, static and spherically symmetric solutions, not completely determined by their mass.

Unlike in GR, in HR bimetric theory, on which we will be focusing in the present paper, the static and spherically symmetric field equations cannot, in general, be integrated analytically. Nevertheless, in some regions suitable approximations can be made, allowing for an analytical study of the behaviour of solutions.

In \cite{Comelli:2011wq}, the authors performed an analytical study of spherically symmetric BH solutions, finding that there are two principal branches of solutions, later denoted \emph{bidiagonal} and \emph{non-bidiagonal}.\footnote{ 
Note that the two branches (non-bidiagonal and bidiagonal) were introduced by \cite{Salam:1976as} in the context of strong-gravity.} For the non-bidiagonal solution, there is no coordinate system in which the two metrics can be simultaneously diagonalised, whereas, for the bidiagonal solutions, such coordinate frame exists. The authors show that, in the former case, the equations can be integrated analytically and the solutions are equivalent to those in GR: Schwarzschild, Schwarzschild-de Sitter (SdS) and Schwarzschild-Anti-de Sitter (SAdS). In the latter case, the equations of motion cannot, in general, be solved analytically. However, the authors derived linearised solutions, valid outside the Vainshtein radius \cite{Vainshtein:1972sx} of the system (inside the Vainshtein radius, by definition, non-linearities must be taken into account). Linearising the equations decouples the massless mode from the massive one, and the metric functions are sums of a Newtonian potential (massless mode) and a Yukawa potential (massive mode).

Bidiagonal, static and spherically symmetric solutions were also studied in \cite{Babichev:2013pfa}. The authors argue that approximate analytical solutions can be found by solving an algebraic equation, valid inside a region between a point reasonably close to the event horizon (the point \emph{reasonably close} to the horizon being defined as the point when the gravitational fields become large) and the Compton wavelength of the massive graviton. In this way, the analytical behaviour of the solutions can be also understood inside the Vainshtein radius.

Restricting ourselves to exact, bidiagonal solutions, in \cite{Babichev:2014tfa} it was shown for the first time that, under suitable conditions, a Kerr $g$-metric and a flat reference metric, $f$, is an exact solution in dRGT massive gravity. Also, the bi-Kerr case with both $g$ and $f$ being Kerr metrics is an exact solution in HR bimetric theory. Since Schwarzschild geometry can be obtained from Kerr, setting the  angular momentum to zero, this implies that the Schwarzschild solution with a flat reference metric $f$ is an exact solution in dRGT massive gravity and that the bi-Schwarzschild solution is an exact solution in HR bimetric theory. Charged BH solutions on the 
Reissner-Nordstr\"{o}m-de Sitter form were also found for dRGT massive gravity and HR bimetric theory in \cite{Babichev:2014fka}. To summarise, all known {\em exact} BH solutions of HR bimetric theory correspond to the GR solutions.

In addition to exact and approximate analytical solutions, numerical solutions have been found in the interval $r\in (r_\eH,+\infty)$, with the Killing horizon radius, $r_\eH$, not being included in the domain. A comprehensive analysis of static and spherically symmetric numerical solutions was performed in \cite{Volkov:2012wp}, in which bidiagonal solutions different from the GR solutions were obtained numerically. 
Small perturbations around GR solutions were added to the initial conditions at the Killing horizon after which numerical integration was performed, using the new perturbed initial conditions. 
The results indicated that only SAdS solutions admit perturbations which have
AdS-type asymptotic in the leading order.
On the other hand, asymptotically flat perturbations of Schwarzschild solutions were not found and SdS solutions were found not to admit non-compact perturbations in the sense that perturbations diverge at finite radii.
Numerical studies were also performed in \cite{Brito:2013xaa}, in which hairy, static and spherically symmetric solutions, being asymptotically flat were explicitly found. The authors proposed that they may represent the final stable state of perturbed Schwarzschild BHs, which are in general unstable \cite{Babichev:2013una,Brito:2013wya,Babichev:2014oua,Babichev:2015zub} (see below for more details about stability of BH solutions).

Astrophysical and cosmological properties of bidiagonal static and spherically symmetric BHs were studied in \cite{Enander:2013kza,Enander:2015kda}. In \cite{Enander:2013kza}, the bidiagonal solutions found in \cite{Comelli:2011wq} were rederived using isotropic coordinates to study their lensing properties. The relationship between bidiagonal static and spherically symmetric BHs and stars in HR bimetric theory was discussed in \cite{Enander:2015kda}. The authors noted that BH solutions and star solutions differ in their asymptotic properties at large radii, indicating that stars cannot gravitationally collapse into bidiagonal BHs.

Stability is of large importance when discussing the physical relevance of BH solutions. Two types of stability are of importance for BH solutions: The first is the stability of a solution against a small change in its initial conditions, the second is the stability of a solution on external perturbations.  In the former case, perturbations are made in the phase-space of the system of differential equations, whereas in the latter, perturbations are dynamical (e.g., a radial oscillation for a Schwarzschild BH). In the following, we will refer to the former case as \emph{Lyapunov stability} and to the latter as \emph{dynamical stability}.
Note that the study of Lyapunov stability is standard practice when considering differential equations in GR, see for example \cite[sections 10.2--10.5]{rendall2008partial}.

Linear, dynamical stability analyses of BH solutions were performed in \cite{Babichev:2013una,Brito:2013wya,Babichev:2014oua,Babichev:2015zub}. In \cite{Babichev:2013una} it was shown that linear perturbations around bidiagonal bi-Schwarzschild solutions, both in dRGT massive gravity and HR bimetric theory, are unstable, consistent with the results in \cite{Brito:2013wya} for linear, spherically symmetric perturbations of a Schwarzschild BH. In \cite{Brito:2013wya}, it was also showed that Kerr BHs are unstable against radial perturbations and superradiant perturbations (a recent review about superradiance can be found in \cite{Brito:2015oca}). The dynamical instability of the bidiagonal Schwarzschild solution was confirmed in \cite{Babichev:2014oua}, where it was also shown that radial perturbations around static and spherically symmetric {\em non-bidiagonal} solutions do not exhibit unstable modes. These results were again obtained in \cite{Babichev:2015zub}, where the authors perform an extensive study of the quasi-normal modes of both bidiagonal and non-bidiagonal solutions. Note that for parameter values corresponding to the partially massless case (for partially masslessness, see \cite{Hassan:2012gz} and reference therein), the bidiagonal GR solutions (Schwarzschild, SdS and SAdS) are dynamically stable \cite{Brito:2013yxa}.

General properties of the horizon structure of static and spherically symmetric BHs and their thermodynamics in modified theories of gravity can be found in, e.g., \cite{Deffayet:2011rh,Banados:2011hk}.
Recent reviews about BH solutions in dRGT massive gravity and HR bimetric theory can be found in \cite{Volkov:2013roa,Volkov:2014ooa,Babichev:2015xha}.

\paragraph{Structure of the paper.}

\autoref{fig:diagram} at page 23, apart from describing the new proposed classification for BH solutions, shows the structure of the paper until that point. The caption of the figure contains a detailed description of what we do in each section. In any case, a more general description is reported in the following.

In \autoref{section-2} we briefly review the HR bimetric theory formulation. In \autoref{section-3}, we describe the theoretical background and clarify the meaning of the bidiagonal ansatz. We also present an original proposition stating that, if one metric is regular and without curvature singularities, to have a non-singular square root matrix $S=\sqrt{g^{-1}f}$ is a necessary condition for not inducing curvature singularities into the other metric. We further introduce a new parametrisation for the two metrics, allowing for a description of the solutions both outside, inside and at the Killing horizon, as well as detecting a possible pathological behaviour of the solutions. This new parametrisation allows us, for the first time, to study the behaviour of the solution a the Killing horizon, and leads us to present two propositions regarding the behaviour of the metrics at their common Killing horizon. Taking into account these results, we propose a new classification for the solutions within the bidiagonal ansatz and discuss the fields equations and initial conditions at the Killing horizon. In \autoref{section-4}, we present our BH solutions and study their convergence properties when $r\rightarrow \infty$, both numerically and analytically through a Lyapunov stability analysis.

\paragraph{Executive summary.}
Using a combination of numerical and analytical techniques, we conclude that the only solutions converging to Minkowski, de Sitter (dS) and Anti-de Sitter (AdS) spacetimes for large radii in HR bimetric theory, are the GR solutions: Schwarzschild, Schwarz\-schild-de Sitter (SdS) and Schwarzschild-Anti-de Sitter (SAdS). For these solutions, the $g$ and $f$ metrics are conformal, i.e., proportional to each other.

Non-GR solutions exist, but they diverge from Min\-kow\-ski, dS and AdS, 
even if they differ only infinitesimally from the GR solutions at the Killing horizon. Also, all the found perturbations around 
Schwarzschild and SdS solutions display singularities in the interaction potential between the two metrics at finite radii. For this reason, we do not consider them as acceptable solutions.

These results, concerning the asymptotic behaviour of the solutions, are compatible with those in \cite{Volkov:2012wp}, where the leading order asymptotic properties of the solutions are studied for $r \rightarrow \infty$. Therein, the relative differences between non-GR solutions and GR ones are taken into account, whereas, in this paper, we focus on the absolute differences between them (which is the relevant property for Lyapunov stability). We note that the behaviour of the relative differences does not tell anything about the behaviour of the absolute difference (i.e., about the Lyapunov stability); therefore, the results of \cite{Volkov:2012wp} do not imply ours. Therefore, the main motivation of this work is to extend the results of \cite{Volkov:2012wp} by studying the Lyapunov stability of the solutions.

Moreover, we study for the first time the causal structure of BH spacetimes in HR bimetric theory at the common Killing horizon by using Eddington-Finkelstein coordinates. In the light of this analysis, we propose a new classification for the BH solutions and point out that choosing a bidiagonal ansatz outside or inside the Killing horizon, does not guarantee bidiagonality exactly \emph{at} the Killing horizon.

As already pointed out in \cite{Volkov:2012wp}, GR solutions are not unique since, given a set of values for the parameters in the action, we can have different conformal factors between the two metrics. Each conformal factor corresponds to one value of the cosmological constant for the GR solutions. If two conformal factors lead to $\Lambda=0$, then we will have two Schwarzschild solutions. Similarly for SAdS when we have different $\Lambda<0$ and for SdS when we have different $\Lambda>0$.
Note that this cosmological constant is different from the one defined in the action of GR.
In this sense, a GR solution will not be completely specified by the global parameters of the theory and its mass, since, for fixed values of these, we can have a discrete and finite number of observationally distinguishable solutions, specified by the the different conformal factors. In addition, turning to non-GR solutions, to completely specify the BH, we need to specify possible deviations from the GR solutions regarding the initial conditions at the Killing horizon, even though such deviations are diverging from the GR solutions at large radii. Whether the differences between these solutions should be denoted \emph{hairs}, depends on the exact definition of \emph{hair} employed.
We note, however, that for fixed values of the global action parameters and the BH mass, all found BH solutions always diverge from each other at large radii.
\vspace{0.8em}

\noindent Given that:\vspace{-0.5em}
\begin{enumerate}\itemsep0pt
\item the only bidiagonal BH solutions in HR bimetric theory asymptotically converging to Minkowski, dS and AdS spacetimes are GR solutions with conformal metrics,
\item bidiagonal BH solutions display dynamical linear instabilities,
\item spherically symmetric solutions including matter sources, e.g. star solutions, have non-conformal metrics,
\end{enumerate}
the outstanding remaining question is: What is the endpoint of gravitational collapse of matter in HR bimetric theory?

\section{The Hassan-Rosen bimetric theory}
\label{section-2}

The action of the HR bimetric theory in vacuum is \cite{Hassan:2011zd},
\begin{equation}
\label{eq:bimetricAction}
	\mathcal{S} =
		\intop \dd^4 x\left[ 
		\dfrac{1}{2} M_g^2\sqrt{-\det g}R^g
		+ \dfrac{1}{2} M_f^2\sqrt{-\det f}R^f 
		- m^2 M_g^2 % here it was m^4 which is wrong. TODO Edvard, can you please check this
		\sqrt{-\det g}\sum _{n=0}^4 \beta _n e_n\left(S\right) \right]\!,
\end{equation}
where $M_g,M_f$ are the reduced Planck masses and $R^g,R^f$ are the Ricci scalars respectively for $g_{\mu \nu}$ and $f_{\mu \nu}$. The kinetic terms for $g_{\mu \nu}$ and $f_{\mu \nu}$ have the standard Einstein-Hilbert form. The interaction potential is given in terms of the square root matrix $S\coloneqq\sqrt{g^{-1}f}$ (so that $\tud{S}{\mu}{\rho} \tud{S}{\rho}{\nu} = g^{\mu\rho}f_{\rho\nu}$) through the elementary symmetric polynomials $e_n(S)$. These polynomials are the scalar invariants of $S$ and can be expressed in the following way,
\begin{gather}
\label{eq:elesympol}
	e_0(S) \coloneqq 1, \qquad e_n(S) \coloneqq  \tud{S}{[\mu_1}{\mu_1} \cdots
	 \tud{S}{\mu_n]}{\mu_n}, \ n \ge 1,
\end{gather}
which can be expanded (noting that $e_n(S) \equiv 0$ for $n > 4$),
\begin{gather}
	e_1(S)=\Tr(S), \qquad e_2(S)=\dfrac{1}{2}\left( (\Tr S)^2-\Tr(S^2) \right), \nonumber \\
	e_3(S)=\dfrac{1}{6}\left( (\Tr S)^3-3\Tr(S^2)\Tr(S)+2\Tr(S^3) \right), \qquad e_4(S)=\det(S).
\end{gather}
The $\beta _n$ parameters are arbitrary real numbers, whereas $m$ is a parameter having the dimension of mass, setting the energy scale of the interaction. Since we will not consider any source of matter in this work, we do not include any matter Lagrangian in the action.

Equations of motion in vacuum are obtained by varying the action \eqref{eq:bimetricAction} with respect to $g_{\mu \nu}$ and $f_{\mu \nu}$ separately, yielding,
\begin{subequations}\label{eq:bimetricEoM}
\begin{align}
	\label{eq:bimetricEoM-1}
	R_{\mu \nu}^g-\dfrac{1}{2}g_{\mu \nu}R^g +
		m^2g_{\mu \lambda} \sum\limits _{n=0}^3\left[\beta _n (-1)^n{{Y_{(n)}}^\lambda }_\nu \left( S \right)\right]&=0, \\
	\label{eq:bimetricEoM-2}
	R_{\mu \nu}^f-\dfrac{1}{2}f_{\mu \nu}R^f + 
		\dfrac{m^2}{\kappa}f_{\mu \lambda}\sum\limits _{n=0}^3\left[\beta _{4-n} (-1)^n{{Y_{(n)}}^\lambda }_\nu \left( S^{-1} \right)\right]&=0,
\end{align}
\end{subequations}
where $R^g_{\mu \nu}$ and $R^f_{\mu \nu}$ are the Ricci tensors of $g_{\mu \nu}$ and $f_{\mu \nu}$, respectively, and,
\begin{equation}
\label{eq:Ytensors}
	\tud{Y_{(n)}}{\mu}{\nu} \left( S \right)\coloneqq \sum _{k=0}^n(-1)^k 
	\; \tud{(S^{n-k})}{\mu}{\nu} \; 
	e_k(S),\qquad
	\kappa \coloneqq \left(\dfrac{M_f}{M_g}\right)^2.
\end{equation}
Finally, the Bianchi constraints read,
\begin{subequations}\label{eq:bianchi}
	\begin{align}
	\label{eq:bianchi-1}
	\nabla^{(g)}_\mu \left\lbrace \sum\limits _{n=0}^3\left[\beta _n (-1)^n {{Y_{(n)}}^\mu }_\nu \left( S \right)\right] \right\rbrace &=0, \\
	\label{eq:bianchi-2}
	\nabla^{(f)}_\mu \left\lbrace \sum\limits _{n=0}^3\left[\beta _{4-n} (-1)^n{{Y_{(n)}}^\mu }_\nu \left( S^{-1} \right)\right] \right\rbrace&=0.
	\end{align}
\end{subequations}

\section{Bidiagonal ansatz and field equations}
\label{section-3}

This section is organised as follows. 
In \autoref{subsec:setup}, we define the systems of our interest and clarify the assumptions. We then present a general result, stated as a proposition, about the relation between the curvature singularities and the regularity of the square root. We remark that this result is valid generically in HR bimetric theory, not only for BH solutions. In \autoref{subsec:choice}, we study for the first time the physics of bidiagonal BHs at their common Killing horizon and propose a new classification of the solutions. In \autoref{subsec:equations}, we expand the field equations as nonlinear ODEs and discuss their properties. Finally, in \autoref{subsec:initialconditions}, we consider the \emph{exact} initial conditions at the Killing horizon (not determined in previous works).  

\subsection{Static and spherically symmetric bidiagonal BHs}
\label{subsec:setup}

Our goal is to find all possible bidiagonal static and spherically symmetric BH solutions with a \emph{common} Killing horizon for the two metrics.
The reason to investigate only the bidiagonal solutions is motivated by the fact that all non-bidiagonal solutions are known to correspond to GR solutions (Schwarzschild, SdS and SAdS) as shown, for example, in \cite{Comelli:2011wq}. Whether this is the case also for bidiagonal solutions is the main question we want to address in this paper.

As proved in \cite{Deffayet:2011rh,Banados:2011hk}, if two static and spherically symmetric metrics are diagonal in a common coordinate system and they describe smooth geometry, they must share their Killing horizons.
The statement is very general as it uses the invariant scalars composed from different contractions of two metrics, which are independent from the chosen coordinate system.
This result will be later used to construct the ansatz for the metric $f$.
We first consider the $g$-sector.

The $g$-sector is endowed with the static timelike Killing vector field $\mathcal{K}$. Since the geometry of $g$ is static and spherically symmetric, assuming that the vector field $\mathcal{K}$ is unique, we can adopt coordinates as prescribed in \cite{wald1984}: 
We begin by selecting a spacelike surface $\mathcal{S}$ orthogonal to the orbits of the one-parameter isometry generated by $\mathcal{K}$. We use $t$ to denote this parameter and refer to it as the `time coordinate'. The image of this isometry are the surfaces $\mathcal{S}_t$ also orthogonal to $\mathcal{K}$, which we label by $t$ so that all the points at $\mathcal{S}_t$ have the same $t$. Now, we select a two-sphere on the spacelike surface $\mathcal{S}_t$ and choose angular coordinates $(\theta,\phi)$ where the metric on each 2-sphere takes the form $r^2 \left( \dd \theta^2 + \sin^2 \theta \, \dd \phi^2 \right)$. Although the parameter $r$ is proportional to the square root of the total area of the sphere, it will be referred to as the `radial coordinate' of the sphere. We carry these spherical coordinates to the other spheres of $\mathcal{S}_t$ along the radial geodesics orthogonal to the selected sphere. Assuming a non-vanishing $\nabla_\mu r \ne 0$, we select $(r,\theta,\phi)$ as the coordinates on $\mathcal{S}_t$, resulting in the spacetime coordinates $(t,r,\theta,\phi)$.
Allowing further for the reparametrisation of $r$ in terms of a new coordinate $\xi$ (with the unitary gauge $r = \xi$), we arrive at the coordinate chart $x^\mu=(t,\xi,\theta,\phi)$. The introduction of $\xi$ will become obvious later, when we obtain a system of non-linear ordinary differential equations (ODEs) as the field equations. In short, the system will become autonomous with $\xi$ as the independent variable.
In these coordinates, the metric $g$ takes the form,
\begin{equation}\label{eq:def-g}
\dd s_g^2 \coloneqq -g_{tt}(\xi) \, \dd t^2 + g_{\xi\xi}(\xi) \, \dd \xi^2
+ r(\xi)^2 \left( \dd \theta^2 + \sin^2 \theta \, \dd \phi^2 \right),
\end{equation}
with $g_{tt}(\xi)=-\mathcal{K}^2$. Note that $r(\xi)$ is necessarily monotonic \cite{choquet2008general}. Observe also that this chart breaks down at points where $\mathcal{K}$ or $\nabla_\mu r$ vanish, or when they become collinear.

Turning our attention to the $f$-sector, we begin by noting that the quadratic spaces of $g$ and $f$ (which are locally built on a common tangent space), are mapped to each other by the square root $S$ (a linear transformation).
For any two non-vanishing tangent space vectors $U$ and $V$, due to fact that $S$ is a square root $f(U,V) = g\left(U,S(S(V)\right)$, we have $g\left(S(U),V\right)=g\left(U,S(V)\right)$ and $f(U,V) = g\left(S(U),S(V)\right)$, that is, $S$ is symmetric (self-adjoint) relative to $g$ and it acts as a local pull-back of $g$ to $f$. 
The self-adjointness of $S$ follows from Corollary 1.34 in \cite{higham2008}.

The equivalent expressions of the above statements in matrix notation are: $f = gS^2$ ($S$ is the square root), $\left(g S\right)^\tr = S^\tr g = g S$ ($S$ is self-adjoint), and $f = S^\tr g S$ ($S$ is a congruence). 
If $S$ and $g$ are both diagonal, then $f$ is diagonal as well, and the dynamics of $f$ can be easily transferred to $S$ by defining,
\begin{equation}\label{eq:def-S}
	S \coloneqq \op{diag} \left(\; 
	\pm_1\, |\tau(\xi)|,\; \pm_2\, |\Sigma(\xi)|,\; \pm_3\, |R(\xi)|,\; \pm_4\, |R(\xi)| \;\right),	
\end{equation}
where $\tau(\xi)$, $\Sigma(\xi)$ and $R(\xi)$ are real fields.
The set of independent $\pm_n$ in \eqref{eq:def-S} indicates the square root branch (giving 16 branches in total for $n=1,...,4$). The branch will be selected later in \autoref{subsec:equations}.

Consequently, we have $f_{\xi\xi}(\xi)=\Sigma(\xi)^2\,g_{\xi\xi}$ and $f_{tt}(\xi)= \tau(\xi)^2\, g_{tt}(\xi)=-\mathcal{K}^{\prime2}$, where $\mathcal{K}^{\prime2}$ is the norm of the translational Killing vector $\mathcal{K}$ taken with respect to $f$; hence,
\begin{equation}\label{eq:def-f}
\dd s_f^2 = -\tau(\xi)^2 g_{tt}(\xi) \, \dd t^2 + \Sigma(\xi)^2 g_{\xi\xi}(\xi) \, \dd \xi^2
+ R(\xi)^2 r(\xi)^2 \left( \dd \theta^2 + \sin^2 \theta \, \dd \phi^2 \right).
\end{equation}
The radial coordinate for $f$ is effectively $R(\xi) r(\xi)$, which must be monotonic. The re\-para\-me\-tri\-sa\-tion $U(\xi) = R(\xi) r(\xi)$ will cast \eqref{eq:def-f} into a standard form, as in \eqref{eq:def-g}.
For a non-vanishing $\tau(\xi)$, $\mathcal{K}^{\prime2}=0$ if and only if $\mathcal{K}^2=0$. Hence, the ansatz \eqref{eq:def-g}--\eqref{eq:def-S} ensures a common Killing horizon.

In the above, we required a non-vanishing $\tau(\xi)$ component of $S$. As we shall see, a reasonable requirement is that all the fields of $S$ in \eqref{eq:def-S} are non-vanishing as they contribute to the determinant of $S$.
There are two reasons for that.

First, the square root is a multi-valued function and the branch of the square root is defined by the `sign' variables in \eqref{eq:def-S}. Since different branches are \emph{different functions}, changing the selected branch requires manual `gluing' of the solutions. Note that the selection of the branch is a discrete variable which cannot come out as the solution of a differential equation.

Second, to assert if the obtained solutions are healthy, we can track the elementary symmetric polynomials of $S$, all of which are scalar invariants and independent of the coordinate system used. In particular, the determinant of $S$ vanishes if and only if $S$ is singular, and a non-invertible $S$ will introduce curvature singularities.
This is demonstrated by the following proposition.

\begin{prop}
	\label{th:finiteness}
	Let $g$ be non-singular with finite principal scalar invariants at some point in a region of spacetime, i.e., without curvature singularities. Let further $S$ be a linear transformation which defines $f(\cdot,\cdot) = g(S(\cdot),S(\cdot))$. Then a necessary condition for the principal scalar invariants of $f$ to be finite is that $S$ is invertible.
\end{prop}
\begin{proof}
	The proof is coordinate independent. A non-invertible $S$ induces a curvature singularity in $f$ since the index of $f$ is reduced as the Lorentzian signature of $g$ is spoiled by $S$ (i.e., the singular $S$ will create a non-empty kernel of $f$ so that $f$ cannot be inverted).
	The principal invariants of $f$ (e.g., the Kretschmann scalar) cannot be constructed, since an invertible $f$ is required for contractions. More precise, the connection is not unique for such $f$ and we get a degenerate manifold with singular semi-Riemannian geometry \cite{kupeli1987,kupeli1996}. Note that $f$ and $g$ are interchangeable. 
\end{proof}

\noindent As a corollary, this proposition obviously holds for the square root in the HR action \eqref{eq:bimetricAction}.

One could think that in vielbein formulation, a continuous transition of the vielbein at the point where $\det S=0$ would not be a problem. However, the vielbeins are symmetrised by a fixed Lorentz transformation, and the square root branches correspond to reflections of such transformation (where the principal branch of the square root corresponds to the proper orthochronous Lorentz transformation). This can be seen, e.g., from eqs.~(2.4), (2.7) and (2.16) in \cite{Volkov:2012wp}, after supplementing (2.7) with the symmetrising Lorentz transformation.

Note that the existence of a healthy connection for $f$ and an invertible $S$ is mandatory for writing the Bianchi constraints in the $f$-sector \eqref{eq:bianchi-2}.
When the field equations are studied as a Cauchy problem, the Bianchi constraints will become the energy and the momentum constraint equations on a spacelike hypersurface, which will be propagated by the evolution equations \cite{straumann2013general}. Then during development, without the Bianchi constraint, one cannot cross the hypersurface where $\det S$ vanishes.

Therefore, we postulate an invertible $S$ by definition. However, even if $\det S$ is non-vanishing, some other elementary symmetric polynomial of $S$ could be infinite. For instance, $\tau(\xi)\Sigma(\xi)$ could be finite with  $\tau(\xi) \to \infty$ and $\Sigma(\xi)\to 0$. In this case, the trace of $S$ will diverge. Hence, we require that all the elementary symmetric polynomials of $S$, which are in the action \eqref{eq:bimetricAction} and the equations of motion \eqref{eq:bimetricEoM}, are finite. This imposes, $\tau (\xi)\ne0$, $\Sigma (\xi)\ne0$ and $R(\xi) \neq 0$ for all $\xi$.

\subsection{Choice of the metric fields and behaviour at the Killing horizon}
\label{subsec:choice}

The fields $g_{tt}(\xi)$ and $g_{\xi\xi}(\xi)$ in \eqref{eq:def-g} can be denoted as $q(\xi)$ and $F(\xi)$, so that $F(\xi)=0$ corresponds to the Killing horizon, retaining the finiteness of the metric determinant at the Killing horizon (this will be suitable when going to Eddington-Finkelstein coordinates later). 
The ansatz becomes,
\begin{equation}\label{eq:def-g2}
	\dd s_g^2 \coloneqq -\ee^{q(\xi)} F(\xi) \, \dd t^2 
		+ F(\xi)^{-1} \, \dd \xi^2 
		+ r(\xi)^2 \left( \dd \theta^2 + \sin^2 \theta \, \dd \phi^2 \right),
\end{equation}
\begin{equation}\label{eq:def-S2}
	S \coloneqq \op{diag} \left(\; 
		\pm_1\, |\tau(\xi)|,\; \pm_2\, |\Sigma(\xi)|,\; \pm_3\, |R(\xi)|,\; \pm_4\, |R(\xi)| \;\right),
\end{equation}
with the induced,
\begin{equation}\label{eq:def-f2}
	\dd s_f^2 = -\tau(\xi)^2 \ee^{q(\xi)} F(\xi) \, \dd t^2 
		+ \Sigma(\xi)^2 F(\xi)^{-1} \, \dd \xi^2 
		+ R(\xi)^2 r(\xi)^2 \left( \dd \theta^2 + \sin^2 \theta \, \dd \phi^2 \right).
\end{equation}
From now on, we employ the following naming conventions:
\begin{itemize}
	\itemsep0pt
	\item The fields at an arbitrary point $\xi$ are denoted just by their name.
	\\Example: $r \coloneqq r(\xi)$, $q \coloneqq q(\xi)$, $R \coloneqq R(\xi)$ etc. 
	\item The derivatives of the fields (denoted with primes) are always with respect to $\xi$.
	\item The fields at some designated point, $\xi=\xi_\mathrm{label}$, inherit the same \normalfont{`label'}.
	\\Example: At $\xi=\xi_\eH$, we have
	$r_\eH \coloneqq r(\xi_\eH)$, $q_\eH \coloneqq q(\xi_\eH)$, $R_\eH \coloneqq R(\xi_\eH)$ etc.
\end{itemize}
As earlier pointed out, the coordinate system $(t,\xi,\theta,\phi)$ is not suitable to go through the Killing horizon where $\mathcal{K}^2$ vanishes. For this purpose we introduce the ingoing Eddington-Finkelstein coordinates $x^\mu = (v,\xi,\theta,\phi)$ adapted for $g_{\mu \nu}$. In these coordinates, the metrics $g$ and $f$ are (see \autoref{appendix-B} for more details),
\begingroup\colSep{3pt}
\begin{equation}
	\label{eq:EFmetrics}
		g_{\mu \nu} =
			\begin{pmatrix}
			-\ee^{q} \, F & \ee^{q/2} & 0 & 0 \\
			\ee^{q/2} & 0 & 0 & 0 \\
			0 & 0 & r^2 & 0 \\
			0 & 0 & 0 & r^2 \sin ^2\theta
			\end{pmatrix} \!, \quad
		f_{\mu \nu} = 
			% \tud{S}{\rho}{\mu} g_{\rho\sigma} \tud{S}{\sigma}{\nu} =
			\begin{pmatrix}
			-\ee^{q} \, F \, \tau^2 & \ee^{q/2}\tau^2 & 0 & 0 \\
			\ee^{q/2} \, \tau^2 & \dfrac{\Sigma^2-\tau^2}{F} & 0 & 0 \\
			0 & 0 & R^2 r^2 & 0 \\
			0 & 0 & 0 &  R^2r^2\,\sin ^2\theta
			\end{pmatrix} \!,\vspace{-0.3em}
\end{equation}
with the square root,
\begin{equation}
		\tud{S}{\mu}{\nu} =
		\begin{pmatrix}	
		\pm_1\,|\tau| & ~~ \ee^{-q/2}\,\dfrac{\pm_2\,|\Sigma| - (\pm_1)\,|\tau|}{F} & 0 & 0 \\
		0 & \pm_2\,|\Sigma| & 0 & 0 \\
		0 & 0 & \pm_3\,|R| & 0 \\
		0 & 0 & 0 & \pm_4\,|R|
		\end{pmatrix} \!.
		\label{eq:EFmetrics-2}
\end{equation}
\endgroup
The common Killing horizon is at $F=0$. Clearly, $g$ is regular at the Killing horizon in this coordinate system (which then becomes a null frame for $g$). 
We have, however, a possible problem with $f$ for which the component $f_{\xi\xi}$ has $F$ in denominator. We thus make the following proposition:

\begin{prop}[Crossing condition]
\label{prop:crossing}
~\\\indent
Consider the metrics $g$, $f$ and the square root $S$ given by \eqsref{eq:EFmetrics} and \eqref{eq:EFmetrics-2}, wherein $g$ is regular. The necessary and sufficient condition for $f$ to smoothly cross the common Killing horizon $F(\xi_\eH)=0$ is that the following limit is finite,
\begin{equation}
	\label{eq:crossing}
	f_{\xi\xi}(\xi_\eH) \coloneqq \lim _{\xi\rightarrow \xi_{\mbox{\emph{\tiny{H}}}}}
	\dfrac{\Sigma(\xi)^2-\tau(\xi)^2}{F(\xi)} = \op{const}.
\end{equation}
\end{prop}
\begin{proof}
The proof simply follows from \eqref{eq:EFmetrics}.\footnote{The form of $S$ does not influence the result. Consider the $2\times2$ blocks of $g$ and $S$ in the form,
	\begin{equation}
	g = \begin{pmatrix} -F & 1 \\ 1 & 0 \end{pmatrix} \!, \qquad
	S = \begin{pmatrix} \tau & a \\ b & \Sigma \end{pmatrix} \!, \nonumber
	\end{equation}
so that $f = g S^2$. The symmetrisation $f = f^\tr$ implies $a = (\Sigma - \tau) / F$ without any restriction on $b$. Then (with $b=0$ in our case),
	\begin{equation}
	f = \begin{pmatrix} 
	  - F \tau^2 + 2 b \tau & ~\tau ^2 + b\frac{\Sigma -\tau}{F} \\
	\tau ^2 + b\frac{ \Sigma -\tau }{F} & ~\frac{\Sigma^2 -\tau^2 }{F} 
	\end{pmatrix} \!. \nonumber
	\end{equation}
}
\end{proof}
\noindent We refer to the condition \eqref{eq:crossing} as the \emph{crossing condition}.\footnote{This crossing condition is similar to the GR one for $g = F/(UV) \, \dd U \dd V$ in Kruskal-Szekeres coordinates, where $ F/(UV)$ should remain finite at $U = 0$ or $V = 0$, which correspond to the Killing horizons.}
The meaning of `\emph{smoothly cross}' in the proposition will become clear later when discussing the relationship between the null cones of the two metrics in \autoref{subsec:nullcones}.

From now on, we assume that the crossing condition is satisfied.
In the limit $\xi \rightarrow \xi_\eH$ when $F(\xi)\rightarrow 0$, we have,
\begingroup\colSep{2pt}
\begin{equation}
	\label{eq:EFmetricsHor}
	g_{\mu \nu} =
		\begin{pmatrix}
			0 & \ee^{q_\eH/2} & 0 & 0 \\
			\ee^{q_\eH/2} & 0 & 0 & 0 \\
			0 & 0 & r_\eH^2 & 0 \\
			0 & 0 & 0 &  \hspace{-1pt} r_\eH^2 \sin ^2\theta \\
		\end{pmatrix} \!, \quad 
	f_{\mu \nu} =
		\begin{pmatrix}
			0 & \ee^{q_\eH/2}\tau_\eH^2 & 0 & 0 \\
			\ee^{q_\eH/2}\tau_\eH^2 &  \boxed{f_{\xi\xi}(\xi_\eH)} & 0 & 0 \\
			0 & 0 & R_\eH^2r_\eH^2 & 0 \\
			0 & 0 & 0 &  \hspace{-3pt} R_\eH^2r_\eH^2\sin ^2\theta \\
		\end{pmatrix} \!.
\end{equation}
\endgroup
Now, denoting non-zero components by $\bullet$, two \emph{real symmetric} matrices $A$ and $B$ cannot be simultaneously diagonalised if they have the following structure,
\begingroup\colSep{5pt}
\begin{equation}
	A = \begin{pmatrix}
		0 & \bullet & 0 & 0 \\
		\bullet & 0 & 0 & 0 \\
		0 & 0 & \bullet & 0 \\
		0 & 0 & 0 & \bullet \\
	\end{pmatrix} \!, \qquad
	B = 
	\begin{pmatrix}
		0 & \bullet & 0 & 0 \\
		\bullet & \bullet & 0 & 0 \\
		0 & 0 & \bullet & 0 \\
		0 & 0 & 0 & \bullet \\
	\end{pmatrix} \!.
\end{equation}
\endgroup
This follows from
the theorem on canonical pair forms (see \cite{Uhlig:1973,Hassan:2017ugh} and the references therein). Hence, as a corollary, the following proposition holds:
\begin{prop}[Proper bidiagonality condition]
\label{prop:propdiag}
~\\\indent
Let $g$ and $f$ be the metrics of \autoref{prop:crossing}. Then $g$ and $f$ are simultaneously diagonalisable at the common Killing horizon $F(\xi_\eH)=0$, if and only if,
\begin{equation}
	f_{\xi\xi}(\xi_\eH) = 0. \label{eq:propdiag}
\end{equation}
\end{prop}
\noindent We refer to the condition \eqref{eq:propdiag} as the \emph{proper bidiagonality condition}.
A similar statement can be made by extending Proposition 1 in \cite{Deffayet:2011rh} (see \autoref{appendix-C} for more details).

The other way to state \autoref{prop:propdiag} would be that the two metrics are simultaneously diagonalisable at the Killing horizon, if and only if, the Eddington-Finkelstein coordinates for $g_{\mu \nu}$ are the Eddington-Finkelstein coordinates for $f_{\mu \nu}$ too, and they are a null frame for both metrics. This is obvious from \eqref{eq:EFmetricsHor} after setting $f_{\xi\xi}=0$. In such case, having also $g_{vv} = f_{vv} = 0$, the $(v,\xi)$ coordinates become null for both metrics.

\autoref{prop:propdiag} contains two pieces of information: The first is that, if \eqref{eq:propdiag} and consequently $\Sigma_\eH^2=\tau_\eH^2$ holds, then the two metrics are conformally related separately in their time-radial and their angular part at the Killing horizon. The conformal factor for the time-radial part is $\tau_\eH^2$ and for the angular part, $R_\eH^2$. We refer to the time-radial 2-block of the metric with $g_\eS{TR}$ and to the angular 2-block with $g_\eS{ANG}$. At the Killing horizon,
\begin{equation}\label{eq:block-conf}
	g_{\mu \nu}=
		\begin{pmatrix}
			g_\eS{TR} & 0 \\
			0 & g_\eS{ANG}
		\end{pmatrix} \!,
	\qquad
	f_{\mu \nu}=
		\begin{pmatrix}
			\tau_\eH^2g_\eS{TR} & 0 \\
			0 & R_\eH^2g_\eS{ANG}
		\end{pmatrix} \!,
\end{equation}
hence, the metrics are \emph{block-conformal} with conformal factors $\tau_\eH^2$ and $R_\eH^2$.
The second piece is that, even if we impose bidiagonality outside and/or inside the horizon, this is not enough to guarantee bidiagonality \emph{at} the horizon. If the proper bidiagonality condition \eqref{eq:propdiag} is not satisfied whereas the crossing condition \eqref{eq:crossing} is satisfied, the two metrics are bidiagonal outside and inside the horizon, but not exactly at the horizon. This is due to the fact that, at the horizon, the metrics written in Schwarzschild coordinates are not well-defined, and one cannot apply the inverse coordinate transformation from Eddington-Finkelstein to Schwarzschild to recover bidiagonality.\footnote{The Jacobian \eqref{eq:EFjacobian} is ill-defined at the Killing horizon independently of the proper bidiagonality condition.}

An example of coordinates that diagonalises the metrics, if the proper bidiagonality condition is true, is,
$T \coloneqq (v + r)/\sqrt2$, $X \coloneqq (-v + r)/\sqrt2$ (which is the $\pi /4$ clockwise rotation of the two basis vectors $\partial_v$ and $\partial_{\xi}$).
In these coordinates, we obtain the following two block-conformal metrics at horizon,
\begin{subequations}
\begin{align}
	\dd s_g^2 &= -\ee^{q_\eH/2} \, \dd T^2 
		+ \ee^{q_\eH/2} \, \dd X^2 
		+ r_\eH^2 \left( \dd \theta^2 + \sin^2 \theta \, \dd \phi^2 \right),\\[0.3em]
	\dd s_f^2 &= -\ee^{q_\eH/2}\tau_\eH^2 \, \dd T^2 
		+ \ee^{q_\eH/2}\tau_\eH^2 \, \dd X^2 
		+ R_\eH^2r_\eH^2 \left( \dd \theta^2 + \sin^2 \theta \, \dd \phi^2 \right).
\end{align}
\end{subequations}
If we use the same coordinates without imposing the proper bidiagonality condition, $f$ will be manifestly non-diagonal with a non-zero $\dd T \dd X$ component.

This different behaviour at the horizon suggests a new classification for the solutions:
\vspace{-0.2em}
\begin{enumerate}[label=(\roman*)]\itemsep0pt
	\item \emph{proper bidiagonal} solutions, for which the proper bidiagonality condition holds, and
	\item \emph{improper bidiagonal} solutions, for which only the crossing condition holds. The improper bidiagonal solutions are bidiagonal everywhere except at the Killing horizon.
\end{enumerate}

\subsection{Field equations}
\label{subsec:equations}

To expand the field equations \eqref{eq:bimetricEoM}, we need to specify the branch of the square root. We select the principal branch where $\pm_n$ in \eqref{eq:EFmetrics} are all positive, and drop the absolute values in \eqref{eq:bimetricEoM} assuming all the fields of $S$ to be positive. The other branches will be treated elsewhere.
With this choice of the branch, we further define an auxiliary field called the \emph{horizon crossing function},
\begin{equation}
	  \Phi(\xi)\coloneqq \frac{\Sigma(\xi)-\tau(\xi)}{F(\xi)}.
\end{equation}
According to \eqnref{eq:crossing}, $\Phi(\xi)$ must be finite at the horizon.
For increased readability, we also introduce the auxiliary field $\sigma(\xi)$ by,
\begin{equation}
	\sigma(\xi) \coloneqq \frac{\Sigma(\xi)}{\left(R(\xi)r(\xi)\right)^\prime} = \frac{\Sigma(\xi)}{R'(\xi)r(\xi)+R(\xi)r'(\xi)}.
\end{equation} 
Note that $\sigma$ is not to be considered an independent field since its dynamics is governed by $\Sigma$. Also, given that $rR$ must be monotonic \cite{choquet2008general} and $\Sigma$ must be non-zero and non-diverging due to the regularity of the elementary symmetric polynomials of $S$, $\sigma$ must also be non-zero and finite.

With the fields given in \eqref{eq:def-g2}--\eqref{eq:def-f2}, gauge-fixing $r(\xi)=\xi$, and restricting ourselves to the principal square root branch, the field equations are,
\begin{subequations}\label{eq:EoM}
\begin{align}
		q'r &= m^2r^2\Phi\alpha_1, \label{eq:EoM-1} \\
		F'r &= 1-F-m^2r^2\left[\alpha _0+\Sigma\alpha_1 \right], \label{eq:EoM-2} \\
		\sigma'r &= \dfrac{1}{2\tau}\left[\dfrac{2(\sigma-1)\Sigma}{\alpha _1}\left( b_1 +\tau b_2 \right)+\Phi\sigma\left(-1+F+m^2r^2\left(\alpha _0+\dfrac{\sigma\alpha_1}{\kappa}\right)\right) \right], \label{eq:EoM-3} \\
		\tau'r &= \dfrac{\tau}{2F}\left[-1+F+m^2r^2\left(\alpha _0+\tau\alpha_1\right)\right] \;- \nonumber \\
			& \qquad \dfrac{\Sigma}{2R\sigma F}\left[ \dfrac{m^2r^2}{\kappa}\sigma^2\alpha_1+\tau \left( F+\sigma^2\left( \dfrac{m^2r^2}{\kappa}\alpha_2-1 \right) \right) \right], \label{eq:EoM-4} \\
		\tau'r &= \dfrac{(\sigma-1)\Sigma}{\sigma\alpha_1}\left( b_1 +\tau b_2 \right)-\dfrac{1}{2}\Phi\left[-1+F+m^2r^2\left(\alpha _0+\tau\alpha_1\right)\right],  \label{eq:EoM-5} \\
		r' &= 1.\label{eq:EoM-6}
\end{align}
\end{subequations}
Equation \eqref{eq:EoM-5} is the Bianchi constraint for $g$. The Bianchi constraint for $f$ has the same structure as of $g$, but with the $S$-fields in the denominator (i.e., divided by $\det S$).
Here, we have introduced the \emph{utility functions} $b_n[R]$ and $\alpha_n[R]$ defined as,
\begin{equation}
	b_n[R] \coloneqq \beta_n+\beta_{n+1}R, \qquad \alpha_n[R] \coloneqq b_n[R] +b_{n+1}[R]\,R = \beta_n+2\beta_{n+1}R+\beta_{n+2}R^2,
\end{equation}
which come out from the interaction between the two metrics; they contain the $\beta$-para\-meters and are always multiplied by the energy scale $m$ of the interaction potential.\footnote{Note that our $\alpha$-utility functions are the analogous of the same-named functions defined in \cite{Volkov:2012wp}, but shifted for index 1 and written with our choice of the metric fields.} Since all the fields in \eqref{eq:EoM} are only functions of $\xi$ and there is no explicit dependence on $\xi$ in the equations, the system is an autonomous system of nonlinear ODEs with $\xi$ being the independent variable. 

Note that the field derivatives in \eqref{eq:EoM} are grouped with $r$ and that $r$ is always given together with $m$. This shows that the system exhibits a scaling invariance with respect to $r$, provided that $mr$ is kept unchanged.\footnote{This scaling property can be used to make the radial coordinate dimensionless by setting the Compton wavelength to 1.} By simply replacing $r \to 1$ and $m \to m\, \ee^{\tilde r}$, the system will be given in terms of $\tilde r = \log r$. This scaling invariance provides an intriguing possibility to trade $r$ for $m$ and consider $m$ as the dynamical variable.

The system \eqref{eq:EoM} is linear in the field derivatives and in an almost \emph{normal form}, i.e., the form where first derivatives of each field depend only on the fields and not on their derivatives, with the dependence given as explicit functions.
However, solving \eqref{eq:EoM} for derivatives, we get $R'=0$ which implies that the system is dynamically overdetermined. In other words, there is an algebraic constraint rendering one of the fields non-dynamical. In principle, any of the fields can be rendered algebraic from any two equations, but the inspection of \eqref{eq:EoM} shows that only $\tau$ and $F$ have a linear dependence. From \eqsref{eq:EoM-4} and \eqref{eq:EoM-5}, we get,
\begin{subequations}\label{eq:tauF45}
\begin{align}
	\tau&=\frac{-4b_{1}FR(\sigma-1)+\alpha_{1}\sigma R\left(-1+F+\alpha_{0}m^{2}r^{2}\right)-\sigma^{2}\alpha_{1}^{2}m^{2}r^{2}\kappa^{-1}}{4b_{2}FR(\sigma-1)+\alpha_{1}F+\alpha_{1}\sigma^{2}\left(\alpha_{2}m^{2}r^{2}\kappa^{-1}-1\right)-\sigma R\alpha_{1}^{2}m^{2}r^{2}}, \label{eq:tau45-eq}\\[0.5em]
	F&=\frac{\sigma\alpha_{1}\left(R\left(-1+m^{2}r^{2}\left(\alpha_{0}+\tau\alpha_{1}\right)\right)+\sigma\left(\tau-\left(\alpha_{1}+\tau\alpha_{2}\right)m^{2}r^{2}\kappa^{-1}\right)\right)}{4\left(b_{1}+b_{2}\tau\right)R(\sigma-1)+\alpha_{1}\left(\tau-R\sigma\right)}. \label{eq:F45-eq}
\end{align}
\end{subequations}
Taking the derivative of either $\tau$ or $F$ from \eqref{eq:tauF45}, and plugging it back into \eqref{eq:EoM}, enables us to solve for $R'$ (or $\Sigma$).\footnote{Following a procedure similar to that described in \cite{Volkov:2012wp}.}
Although $\tau $ and $F$ could be used on equal footing, the remnant equations are simpler for $\tau$.
Thus, we select $F$ to be dynamical and $\tau$ to be determined from the other fields by \eqref{eq:tau45-eq}.
Note that if we selected $F$ to be determined algebraically, all dynamics would be determined by the square root, $S$.
Once we know the algebraic expression for $R'$, we also know $\Sigma$ and $\Phi$ \emph{algebraically} (again noting that $\sigma$ took over dynamics from $\Sigma$). Hence, we got the normal form of \eqref{eq:EoM}.

We also note that equations \eqref{eq:EoM-2}--\eqref{eq:EoM-6} do not depend on $q$. This reflects an invariance under the rescaling of $t$ coordinate, $\ee^{q} \dd t \to \dd t$. As a consequence, we can integrate \eqnref{eq:EoM-1} and determine $q$ after the other fields are known (up to an integration constant).

To summarise: $\tau$, $\Sigma$ and $q$ can be determined once the fields $r$, $F$, $\sigma$ and $R$ are known, and for the latter we have an \emph{autonomous system} of four nonlinear ODEs, where three of them are coupled,
\begin{equation}
	\label{eq:dynsystem}
	\mathscr{V} '(\xi)=\mathscr{F} \left[\mathscr{V}(\xi)\right],
	\quad
	\mathscr{V}(\xi) \coloneqq 
		\begin{pmatrix}
			F(\xi) \\
			\sigma (\xi) \\
			R(\xi) \\
			r(\xi)
		\end{pmatrix},
	\quad
	\mathscr{F} \left[\mathscr{V}(\xi)\right] \coloneqq
		\begin{pmatrix}
			\mathscr{F}_F(r,F,\sigma,R; \, \{\beta_n\},m,\kappa) \\
			\mathscr{F}_\sigma(r,F,\sigma,R; \, \{\beta_n\},m,\kappa) \\
			\mathscr{F}_R(r,F,\sigma,R; \, \{\beta_n\},m,\kappa) \\
			1
		\end{pmatrix} \!.
\end{equation}
Here, we introduced two vector fields, $\mathscr{V}$ and $\mathscr{F}$, to simplify expressions (these vector fields will be extensively used later on in the paper). The fields $\mathscr{F}_F$ and $\mathscr{F}_\sigma$ are given by \eqref{eq:EoM-2} and \eqref{eq:EoM-3}, respectively.  The expression for $R'=\mathscr{F}_R$ will not be quoted here as its total number of indivisible sub-expressions is around 2600.

\subsection{Exact initial conditions at the Killing horizon}
\label{subsec:initialconditions}

The system \eqref{eq:dynsystem} is an initial value problem. Thus, to solve it, either analytically or numerically, we must specify four initial values (conditions).\footnote{Since we are dealing with rational functions, the existence and the uniqueness of the solutions to \eqref{eq:dynsystem} are guaranteed by the continuous differentiability of $\mathscr{F}$ in all its variables and parameters, unless at poles.} With the choice of fields \eqref{eq:def-g2}--\eqref{eq:def-f2}, we can impose the exact initial conditions at some common Killing horizon and integrate both outwards and inwards. Therefore, we can obtain the complete solutions in the radial interval $r\in (0,+\infty)$.

Denoting one of the Killing horizons $\xi_\eH$, we can start integrating from some $r_\eH=r(\xi_\eH)$.\footnote{Here we assume the existence of at least one Killing horizon. Observe that a solution to \eqref{eq:EoM} can have none, one or many Killing horizons satisfying $F(\xi)$=0. In particular, some solutions can have no Killing horizons, if we start integration from $F(\xi)\ne0$.} This is our first of four initial conditions. The second exact initial condition at the horizon is, of course, $F(\xi_\eH)=0$.

We use the crossing condition in \autoref{prop:crossing} to determine one of the two remaining initial conditions. 
In order to do this, we study the behaviour of the crossing function,
\begin{equation}
\label{eq:phi}
\Phi =\dfrac{\Sigma-\tau}{F}=\dfrac{\sigma(R+r R')-\tau}{F},
\end{equation}
at the horizon. The algebraic expression for $\tau$ is known and, as pointed out earlier, it is possible to derive an equation for the combination $R+r R'$, which is inside $\Sigma $. It turns out that both $\tau$ and $\Sigma$ depend on $F$ as rational functions,
\begin{equation}
	\Sigma = \dfrac{P_\Sigma(F)}{Q_\Sigma(F)}, \qquad \tau = \dfrac{P_\tau(F)}{Q_\tau(F)},
\end{equation}
where $P$ indicates the polynomial in $F$ at the numerator and $Q$ the polynomial in $F$ at the denominator. Therefore, we can write,
\begin{align}
%\label{eq:phiF}
\Phi &=\dfrac{1}{F}\left[ \dfrac{P_\Sigma}{Q_\Sigma}-\dfrac{P_\tau}{Q_\tau} \right]=\dfrac{1}{F}\left[ \dfrac{P_\Sigma       Q_\tau-Q_\Sigma P_\tau}{Q_\Sigma Q_\tau} \right] =\dfrac{1}{F}\left[ \dfrac{p_0+p_1F+p_2F^2}{q_0+q_1F+q_2F^2+q_3F^3} \right] \nonumber \\
&= \dfrac{p_0}{F\left(q_0+q_1F+q_2F^2+q_3F^3\right)} +\dfrac{p_1+p_2F}{q_0+q_1F+q_2F^2+q_3F^3}.
\label{eq:phipol}
\end{align}
The coefficients $p_i$ and $q_j$ are functions of $R$, $\sigma $ and $r$. In order to have a finite limit when $F\rightarrow 0$, the term $p_0$ must vanish at the horizon,
\begin{equation}
\label{eq:BC}
	p_0\left( r_\eH, R_\eH, \sigma _\eH \right)=0.
\end{equation}
After fixing $r_\eH$, \eqnref{eq:BC} is fourth-order in $R_\eH$ and third-order in $\sigma _\eH$ and we can solve it for either. We consider $R_\eH$ as a free parameter, and solve \eqnref{eq:BC} for $\sigma_\eH$ given a value of $R_\eH$.
It turns out that there are always only two independent non-zero solutions for $\sigma _\eH$, corresponding to two possible disconnected branches of solutions.

The last thing to consider is the fact that the last term in \eqnref{eq:phipol}, should have a finite limit,
\begin{equation}
	\lim_{F\rightarrow 0}\left[\dfrac{p_1+p_2F}{q_0+q_1F+q_2F^2+q_3F^3}\right]=
	\left. \dfrac{p_1}{q_0} \right|_{\xi=\xi_\eH}.
\end{equation}
Therefore, after having solved for $\sigma_\eH$, we must check that this quantity is finite. If it is, then the crossing condition is satisfied and we can have a BH solution. Otherwise, we exclude these initial conditions. 
Note also that, since $\Phi _\eH=\left.(p_1/q_0)\right|_{\xi=\xi_\eH}$, then $p_1(r_\eH,R_\eH,\sigma_\eH)=0$ if and only if we have a proper bidiagonal solution. Solving the last equation is equivalent to finding proper bidiagonal solutions.

To summarise, besides $r_\eH$, we only have one free parameter, $R_\eH$. Once specified, and after choosing the branch for $\sigma _\eH$, the system is completely determined and can, in principle, be integrated. However, the equations cannot be integrated analytically, in general. Nonetheless, in the next sections, we will describe a strategy allowing us to do analytical considerations.

\section{Black hole solutions}
\label{section-4}

In this section, we discuss the BH solutions. Subsection \ref{subsec:globalparameters} is about the choice of global parameters, i.e., parameters appearing in the action \eqref{eq:bimetricAction}. In \autoref{subsec:phasespace}, we describe the detailed analysis of found BH solutions. In particular, in subsection \ref{subsubsec:results}, we apply the analysis to the model defined by the chosen global parameters and show that the only solutions converging to Minkowski, AdS and dS at large radii are the GR solutions. In \autoref{subsec:nullcones}, we study the causal structure of improper bidiagonal solutions both inside and outside the Killing horizon.

\subsection{Global parameters}
\label{subsec:globalparameters}

From now on, we use geometrical units, $G=c=1$, together with the Planck constant, $h=1$. We measure lengths in units for which the Fierz-Pauli mass \cite{Fierz:1939,Fierz211} is $m_\mathrm{g}=10^{-1}$. To facilitate comparison with the results of  \cite{Volkov:2012wp}, we use the default set of global parameters given in \autoref{tab:parameters}. We have considered also many other parameter values, but the qualitative behaviours of the found solutions do not depend on their specific values. 
\begin{table}[t]
	\centering
	\begin{tabular}{ccccccc}
	\toprule
		$m_\mathrm{g}$ & $\kappa =\left(\dfrac{M_f}{M_g}\right)^2$ &  $ \beta _0$&$\beta _1$&$\beta _2$&$\beta _3$&$\beta _4$ \\
	\midrule
		$10^{-1}$ & $\tan (1)^{-2}$ &  $-3+3\beta _2+2\beta _3$&$1-2\beta _2 -\beta _3$&$-\dfrac{1}{2}$&$-\dfrac{2}{5}$&$-1-\beta _2 -2\beta _3$  \\
	\bottomrule
	\end{tabular}
	\caption{The default global parameter values used in this paper. The asymptotic flatness condition and normalisation condition (see text below) are used to constrain three $\beta$ parameters, namely $\beta _0$, $\beta _1$ and $\beta _4$.}
	\label{tab:parameters}
\end{table}

Considering the coordinate chart $x^\mu = (t,r,\theta,\phi)$, we demand that the bi-Minkowski solution,
\begin{equation}
\label{eq:biMinkowski}
	g_{\mu \nu}=f_{\mu \nu}=\eta_{\mu \nu} = \op{diag}(-1,1,1,1),
\end{equation}
is a solution of \eqnref{eq:dynsystem}, and obtain the so-called \emph{asymptotic flatness conditions},\footnote{These conditions are actually the `flatness' conditions, because the bi-Minkowski solution is valid everywhere and not only asymptotically.}
\begin{equation}
	\label{eq:AFconditions}
	\beta_0=-3 \beta_1 - 3 \beta _2 - \beta _3 \qquad \beta _4=-\beta _1 - 3\beta _2 - 3\beta_3.
\end{equation}
These conditions can also be deduced from cosmological solutions \cite{vonStrauss:2011mq}. In addition, we use the \emph{normalisation condition},
\begin{equation}
	\label{eq:normalizationCondition}
	\delta = 1, \qquad \text{where} ~\; \delta \coloneqq \beta _1+2\beta _2+\beta _3.
\end{equation}
The normalisation \eqref{eq:normalizationCondition} can always be assumed without loss of generality, since the action \eqref{eq:bimetricAction} is invariant under the imposing the normalisation condition together with the rescaling $m^2\rightarrow m^2 \delta$. 
By using the normalisation condition, we effectively trade one of the free parameters $\beta_1,\beta_2,\beta_3$ in favour of $m_\mathrm{g}$.

Usually, the asymptotic flatness conditions and the normalisation condition contain the conformal factor, often denoted by $c$, between the two metrics for the bi-Minkowski solution. However, any such $c$ can be absorbed into the $\beta $ parameters by rescaling $\beta_n \to \beta_n c^n$. 

\subsection{Numerical solutions and phase space analysis}
\label{subsec:phasespace}

\subsubsection{Choice of initial conditions at the Killing horizon}
\label{subsubsec:ic-at-kh}

\begin{figure}[t]
\centering
	\subfloat[The plane $r_\eH = 1$ intersects the surface $\sigma_\eH(r_\eH,R_\eH)$, thus defining the allowed initial conditions indicated by the black curve, shown in blue and red in (b).]{\raisebox{7mm}{\includegraphics[width=0.45\textwidth]{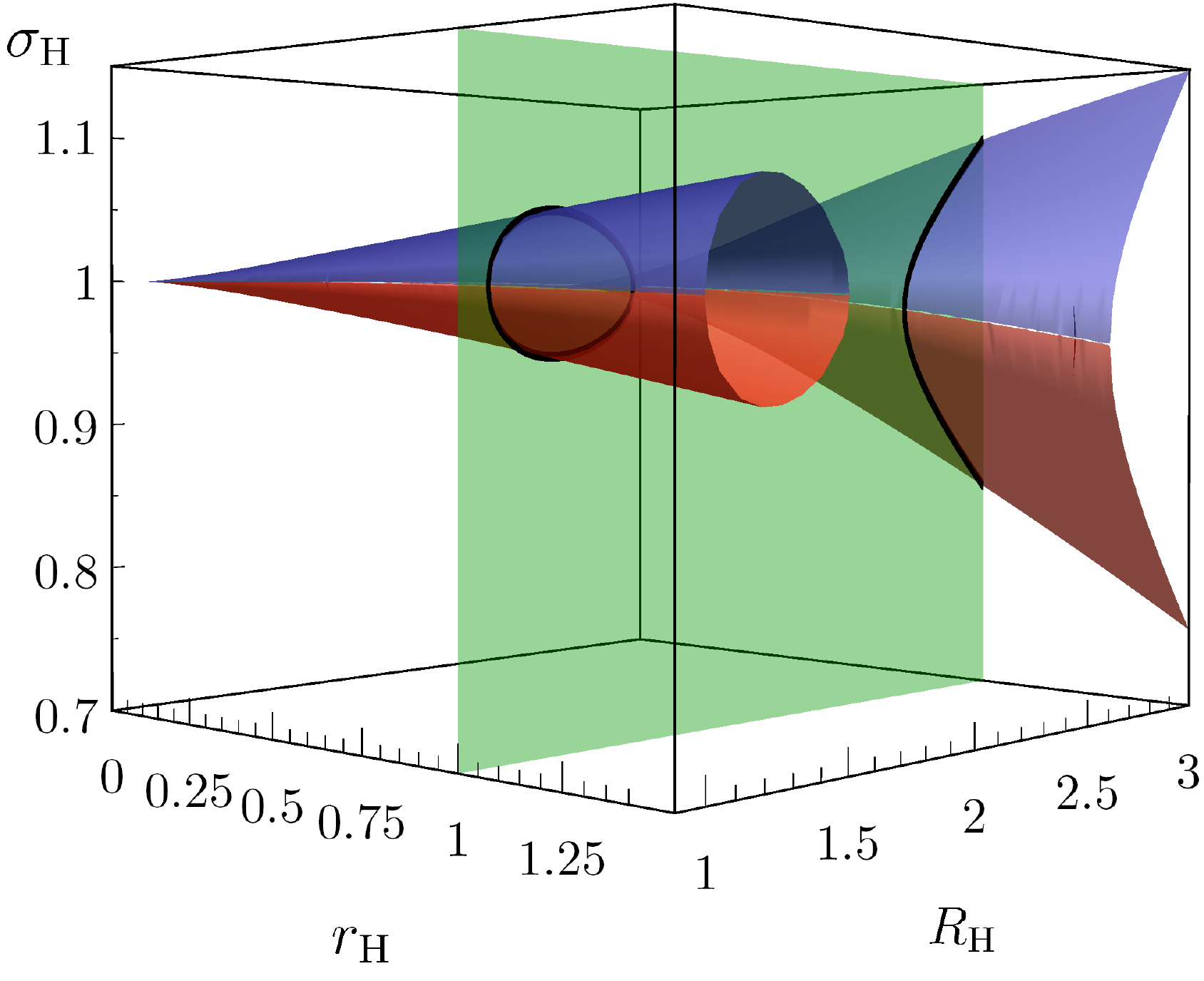}\label{fig:BCa}}}
	\hspace{0.5cm}
	\subfloat[The top panel shows the black curve in (a) with the two branches emphasised in blue and red. The bottom panel shows $\Phi_\eH$ as a function of $R_\eH$. See the text below and \autoref{subsec:lyapconj} for explanations.]{\includegraphics[width=0.45\textwidth]{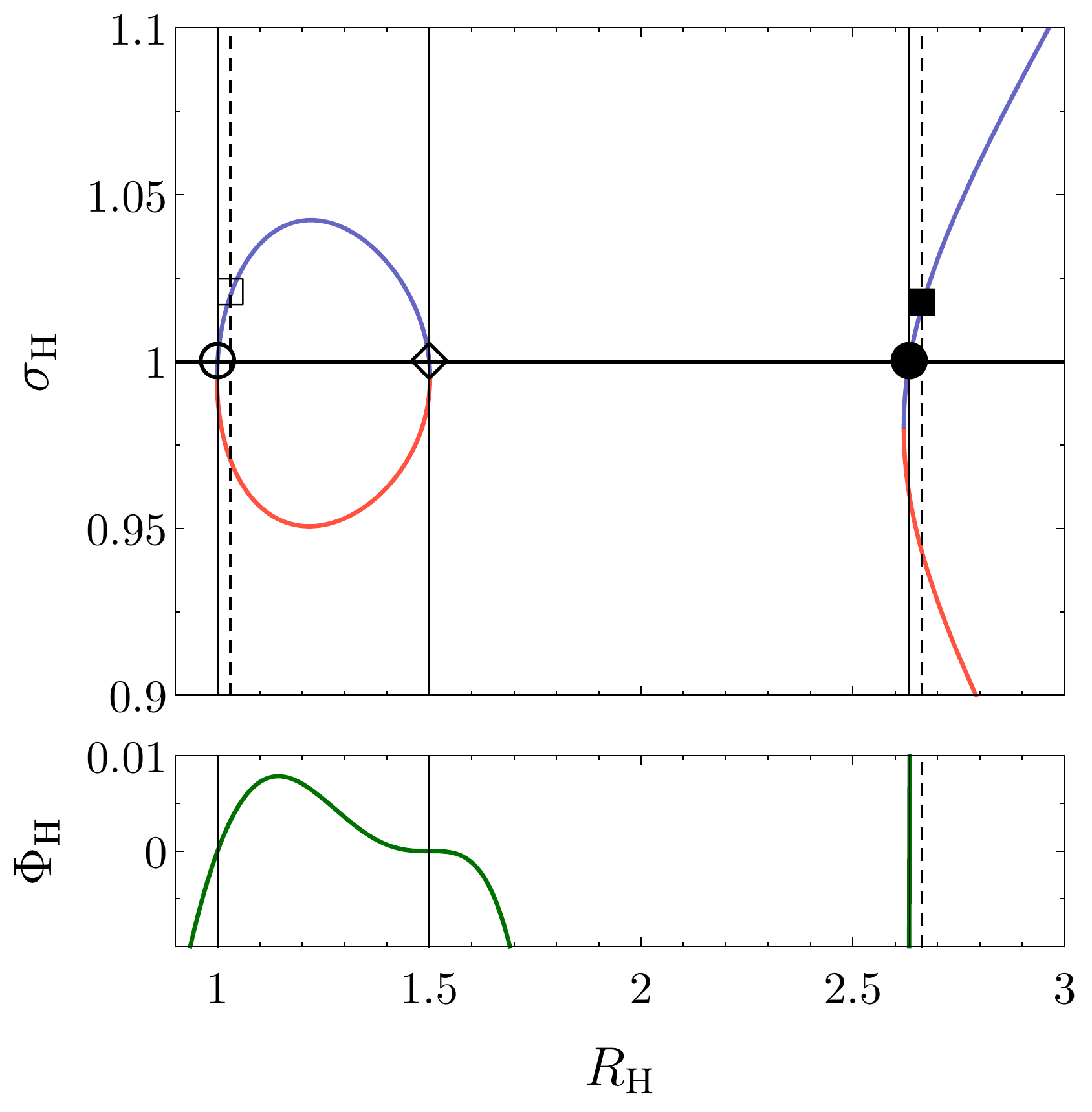}\label{fig:BCb}}
	\caption{
		Choice of initial conditions at the Killing horizon. For the given values of the radius of the Killing horizon $r_\eH$ and of $R_\eH$, we have two possible values of $\sigma_\eH$, belonging to different branches (blue and red).
	}
	\label{fig:BC}
\end{figure}

By definition, in the 4-dimensional phase space $\left( r,F,\sigma,R \right)$ of \eqnref{eq:dynsystem}, exact initial conditions at a Killing horizon must belong to the 3-dimensional hyperplane $F=0$. Moreover, \eqnref{eq:BC} defines a 2-dimensional surface $\mathcal{A}$ embedded in the hyperplane $F=0$. To be more precise, \eqnref{eq:BC} defines an \emph{algebraic variety} \cite{hartshorne1977alg}, i.e., the set of solutions of a system of polynomial equations, in our case \eqnref{eq:BC}.
The surface $\mathcal{A}$ is shown in Figure\autoref{fig:BCa}. The red and blue regions represent the two independent branches for $\sigma _\eH$ (which are the solutions of \eqnref{eq:BC}), as functions of $R_\eH$ and $r_\eH$. Setting the Killing horizon to some $r_\eH=r_*$, defines a 1-dimensional curve $\mathcal{C} \subset \mathcal{A}$, as the intersection of the plane $r_\eH=r_*$ and $\mathcal{A}$ (see Figures\autoref{fig:BCa} and\autoref{fig:BCb}, where $r_*=1$).\footnote{
Note that the distinction between the branches is artificial. We can always introduce a change of coordinates $(r,\sigma,R) \to (\tilde r,\tilde\sigma,\tilde R)$ so that any point becomes centred in the new coordinates with the tangent at the algebraic variety being horizontal.
}

In Figure\autoref{fig:BCb}, the vertical solid lines indicate the values of $R_\eH$ for which $\Phi_\eH=0$, i.e., the corresponding solutions are proper bidiagonal. The vertical dashed lines represent improper bidiagonal solutions close to the proper bidiagonal ones. Once we have $R_\eH$, we must choose one branch for $\sigma _\eH$ (red or blue in the top panel). In this example, we choose the blue branch for all the five solutions considered. {\Large$\circ $} indicates the Schwarzschild solution (being a proper bidiagonal solution), $\square $ indicates an improper bidiagonal solution close to the Schwarzschild one, {\Large$\bullet $} indicates the SAdS solution (also a proper bidiagonal solution) and $\blacksquare $ an improper bidiagonal solution close to SAdS. The SdS solution is not considered in \autoref{fig:BC} since it has $R_\eH<0$ for this parameter set.\footnote{The meaning of $R_\eH<0$ is clarified in subsection \ref{subsubsec:results}.} In \autoref{subsec:lyapconj}, we will explain in more detail why GR solutions correspond to these initial conditions. {\Large $\diamond$}~indicates a proper bidiagonal solution which is non-GR (we come back to this solution in \autoref{subsubsec:relation}).

Summarising, for a BH to have a Killing horizon at some $r_*$, the metric fields must assume values lying on $\mathcal{C}$ at $r=r_*$. Starting from initial conditions belonging to $\mathcal{A}$, we can obtain \emph{all} possible static and spherically symmetric proper and improper bidiagonal BH solutions having a common Killing horizon. Note that only values of $R_\eH$ belonging to the domain of $\mathcal{C}$ are allowed (that is, for which \eqnref{eq:BC} has real solutions).

We choose the Killing horizon to be at $r_\eH=1$. Given that $m_\mathrm{g} = 10^{-1}$, the Compton wavelength of the massive mode is $\lambda =10$, and $r_\eH/\lambda =10^{-1}$. The Fierz-Pauli mass is observationally constrained to be very small, $m_\mathrm{g} \lessapprox 10^{-30}-10^{-33}$eV \cite{deRham:2014zqa}. Therefore, the ratio between the Killing horizon and the Compton wavelength is expected to be much smaller than $10^{-1}$. Such small numbers are demanding to handle numerically. However, since the algebraic variety is only shrunk when $r_\eH$ becomes smaller (see Figure\autoref{fig:BCa}), there are no qualitative changes in the structure of initial conditions. For this reason, we can reasonably extend the following results to more realistic BHs.

\subsubsection{Lyapunov stability and the conjugate system of ODEs}
\label{subsec:lyapconj}

Lyapunov stability is an important property of a solution to an autonomous system of ODEs. It is defined as follows \cite{Lyapunov1992General}: 
\begin{defnw}[Lyapunov stability]
\label{def:lyapunov} 
~\\\indent Consider a generic autonomous system of ODEs with independent variable $t$ and let $\mathscr{S}$ denote the set of all its solutions. A solution $x(t)\in \mathscr{S}$ is called \emph{stable} if, for every $\epsilon >0$, there exist a $\delta _{\epsilon}>0$ and a $t_0$ such that the following holds,
\begin{equation}
\forall y(t)\in \mathscr{S}\quad :\quad \left| y(t_0)-x(t_0) \right|\leq\delta_{\epsilon}\quad ,\quad \left| y(t)-x(t) \right|\leq\epsilon\quad \forall t\geq t_0.
\end{equation}
If, in addition, there exist a $\delta >0$ and a $\bar{t}$ such that the following holds,
\begin{equation}
\forall y(t)\in \mathscr{S}\quad :\quad \left| y(\bar{t})-x(\bar{t}) \right|\leq\delta \quad ,\quad \lim _{t\rightarrow \infty}\left| y(t)-x(t) \right|=0,
\end{equation}
then $x(t)$ is called \emph{asymptotically stable}.
\end{defnw}

One can consider both linear and nonlinear Lyapunov stabilities. In the former case, the equations of motion are linearised and stability properties refer to the linearised system, whereas in the latter case, equations are left in the original nonlinear form and stability properties refer to the full system. If not stated explicitly, we will be referring to nonlinear Lyapunov stability of generic solutions of \eqnref{eq:dynsystem} simply as \emph{stability}. Note that the independent variable $t$ does not need to be time. Thus, this concept of stability does not necessarily concern the time evolution of the system.

When dealing with ODEs and partial differential equations in GR, the study of the stability of the solutions usually refers to Lyapunov stability \cite[sections 10.2--10.5]{rendall2008partial}. Indeed, this concept is very useful physically; for example, when considering linear perturbation around a fixed background, one typically expands the metric as the \emph{sum} of the background metric and the linear perturbation. This is equivalent to study the absolute difference between the unperturbed and the (linearly) perturbed solution, as one has to do in studying (linear) Lyapunov stability. Our approach is quite similar, although it concerns \emph{nonlinear} Lyapunov stability.

In general, when studying a system of ODEs, it is possible to change variables in order to simplify the equations or to get additional insights. In doing this, one applies a $C^k$-diffeomorphism to the phase-space, with $k$ at least 1. The resulting system is said to be \emph{conjugate} to the original one and shares with the original system most of its qualitative properties, such as the number of fixed points \cite{arrowsmith1990}. We remind that a fixed point of a system of ODEs is a solution with all derivatives vanishing identically; in our case, this is when $\mathcal{V}'(\xi)\equiv0$ in \eqnref{eq:dynsystem}.
On the other hand, in order to study the stability of a generic solution $x_0(t)$, one can use the following $C^1$-diffeomorphism,
$
	x(t)\rightarrow x(t)-x_0(t),
$
where $x(t)$ is another generic solution of the system of ODEs \cite{Lyapunov1992General}. Since the new system is conjugate to the old one, one can study the former in order to obtain information about the latter without loss of information. Clearly, if $|x(t)-x_0(t)|$ satisfies the requirements in the definition of Lyapunov stability, then $x_0(t)$ is a stable or asymptotically stable solution.

We now apply these concepts to our case.
Consider the $C^\infty$-diffeomorphism given by,
\begin{subequations}
\label{eq:shifting}
\begin{align}
	\bar{F}&\coloneqq F-\left[ 1-\dfrac{r_\eH}{r}-\dfrac{\Lambda}{3}r^2\left( 1-\dfrac{r_\eH^3}{r^3} \right) \right], \\
	\bar{\sigma}&\coloneqq \sigma  -1, \\
	\bar{R}&\coloneqq R-R_0, \\
	\rho &\coloneqq  1/r,
\end{align}
\end{subequations}
defined in $r\in (0,+\infty)$ and where $\Lambda \in \mathbb{R}$ is constant. Here, $\lbrace \rho , \bar{F},\bar{\sigma}, \bar{R} \rbrace$ are the variables of the conjugate system of ODEs. In doing this, we are: (i) shifting the fields with respect to GR solutions, and (ii) introducing the inverse of the radial coordinate in order to map the radial infinity to $\rho=0$.\footnote{Note that $\rho=0$ is outside of the original domain, but it can be included by compactification.}

Plugging the shifted functions into \eqsref{eq:EoM-2}, \eqref{eq:EoM-3}, \eqref{eq:EoM-6} and into the equation for $R'$ in \eqnref{eq:dynsystem}, we obtain the conjugate ODEs system,
\begin{subequations}
	\label{eq:shiftedEoM}
\begin{align}
	\bar{F}'&=-\bar{F}\rho+\dfrac{ \Lambda -m^2\left(\alpha _0+\Sigma\alpha_1 \right)}{\rho}, \\
	\bar{\sigma} '&=\dfrac{\rho}{2\tau}
	\left[ %\Biggl[
	\dfrac{2\Sigma}{\alpha _1}\left( b_1+\tau b_2 \right)\bar{\sigma } \, + %\nonumber \\[-0.4em] &\qquad 
	(\bar{\sigma }+1)\Phi\left(\bar{F}+\rho r_\eH+\dfrac{\Lambda}{3\rho ^2}+\dfrac{m^2}{\rho^2}\left(\alpha _0+\dfrac{(\bar{\sigma}+1)\alpha_1}{\kappa ^2}\right)\right) 
	\right] %\Biggr] 
	\!, \\
	\bar{R}'&= \bar{\mathscr{F}}_{\bar{R}} (\rho,\bar{F},\bar{\sigma},\bar{R}), \\[0.3em]
	\rho '&=-\rho ^2,
\end{align}
\end{subequations}
where all the fields are functions of $\xi$. Here, $\Sigma $ is a function of $\bar{\mathscr{F}}_{\bar{R}} (\rho,\bar{F},\bar{\sigma},\bar{R})$ rather than $\bar{R}'$. Also, $\alpha_n$ and $b_n$ now stand for $\alpha_n[\bar{R}+R_0]$ and $b_n[\bar{R}+R_0]$, respectively. To find the necessary conditions that must be imposed in order to have fixed points of the conjugate system, we assume that $\left( \bar{F},\bar{\sigma},\bar{R},\rho \right)\equiv \left( \bar{F}_0,\bar{\sigma}_0,\bar{R}_0, \rho _0 \right)$, with the subscript 0 indicating a constant value. If $\bar{F}_0 \neq 0$, then $F$ is not a solution to the original system. Given that solutions of conjugate systems are uniquely mapped into each other through the diffeomorphism \cite{arrowsmith1990}, we conclude that we can have a fixed point of the conjugate system only if $\bar{F}_0 =0$. Besides, $\rho '\equiv 0$ if and only if $\rho \equiv 0$, corresponding to radial infinity. Our fixed points are thus located at radial infinity and their study give information about the asymptotic structure of the solutions.

After inserting $\left(\bar{F},\bar{\sigma},\bar{R},\rho \right)\equiv \left( 0,\bar{\sigma}_0,\bar{R}_0,0 \right)$ into \eqnref{eq:shiftedEoM}, the derivatives of the fields are identically zero only if the following conditions hold (otherwise the right-hand sides of the equations diverge, having $\rho $ in the denominators),
\begin{subequations}
\label{eq:necessaryFP}
\begin{align}
	\Lambda &\equiv m^2\left( \alpha _0[R_0+\bar{R}_0]+\Sigma \alpha _1[R_0+\bar{R}_0]\right), \label{eq:necessaryFP-1}\\
	\bar{\sigma}_0 &\equiv 0, \\
	\Phi &\equiv 0 ~~ \Longrightarrow ~~ \Sigma \equiv \tau , \\
	\bar{R} &\equiv \mbox{const} ~~\Longrightarrow~~
	\left\lbrace\begin{array}{l}
	\alpha _n[R]\equiv \mbox{const},~ \forall n, \\[0.5em]
	\Sigma \equiv R\sigma \equiv R\left(\bar{\sigma}+1 \right)\equiv R\equiv \tau, \\
	\end{array}\right.
\end{align}
\end{subequations}
where $R_0+\bar{R}_0=R_\eH$ from \eqref{eq:shifting}. In particular, condition (a) makes $\bar{F}'$ vanish identically, whereas conditions (b) and (c) combined make $\bar{\sigma}'$ vanish identically, and condition (d) is equivalent to $\bar{R}' = 0$ identically. 
By imposing these necessary conditions, we are considering all the fixed points of the conjugate system of the form,
\begin{equation}
	\left( \rho,\bar{F},\bar{\sigma},\bar{R} \right)\equiv \left( 0,0,0,\bar{R}_0 \right).
\end{equation}
Keeping $\rho$ arbitrary and substituting $\left(\rho,\bar{F},\bar{\sigma},\bar{R} \right)= \left( \rho,0,0,\bar{R}_0 \right)$ in \eqnref{eq:shifting}, together with the conditions in \eqnref{eq:necessaryFP}, we obtain (coming back to the original system),
\begin{subequations}
\label{eq:GRsolutions}
\begin{align}
	F &\equiv  1-\dfrac{r_\eH}{r}-\dfrac{\Lambda (R_\eH)}{3}r^2\left( 1-\dfrac{r_\eH^3}{r^3} \right) \!, \\
	\sigma  &\equiv 1, \\[0.4em]
	\Sigma  &\equiv R \equiv \tau  \equiv \bar{R}_0+R_0= R_\eH.
\end{align}
\end{subequations}
Finally, using \eqsref{eq:tau45-eq} and \eqref{eq:GRsolutions}, we obtain the following quartic equation in $R$ which yields all possible $R_\eH$,
\begin{equation}
	\label{eq:asy-equation}
	R_\eH \in \left\lbrace \; R \in \mathbb{R} \;\vert\;
	m^2 R \left( \alpha_0[R] + R \; \alpha_1[R] \right)- 
	m^2 \kappa^{-1} \left( \alpha_1[R] +  R \; \alpha_2[R] \right)=0
	\;\right\rbrace,
\end{equation}
with the resulting (see \eqnref{eq:necessaryFP-1}),
\begin{equation}
\label{eq:cosmconst}
\Lambda(R_\eH) = m^2 \left( \alpha_0[R\eH] + R\eH \; \alpha_1[R\eH] \right).
\end{equation}
If we require that $R_\eH=1$ is a \emph{flat} solution, i.e., that 
$
 0 = \Lambda(1) = \alpha_0[1] + \alpha_1[1] = \alpha_1[1] +  \alpha_2[1] 
$, we get the asymptotic flatness conditions \eqref{eq:AFconditions}.
Conversely, imposing \eqref{eq:AFconditions} ensures that $R_\eH=1$ is always an element in the set \eqref{eq:asy-equation}.

Now, the equations in \eqref{eq:GRsolutions} tells us that the fixed points correspond to the points at infinity reached by GR solutions when $r\rightarrow \infty$. Second, this shows that GR solutions are exact solutions of \eqnref{eq:dynsystem}. Note that the cosmological constant of the GR solutions depend on $R_\eH$, the free parameter in our initial conditions. Therefore, depending on $R_\eH$, we will have Schwarzschild, SdS or SAdS solutions (see Figure\autoref{fig:BCb}). The asymptotic structure is thus determined by the initial conditions at the Killing horizon. For example, for $R_\eH = 1$, $\Phi \equiv 0$ and $\Lambda(R_\eH)=0$, i.e., the Schwarzschild solution. Note again that the asymptotic flatness conditions \eqref{eq:AFconditions} guarantee that $R_\eH=1$ is always a possible initial value for $R(\xi)$.
Substituting the fields in \eqnref{eq:GRsolutions} into the metrics, the metrics become conformal, with conformal factor $R_\eH^2$. These results are all compatible with previous results (see, for example, \cite{Volkov:2012wp}).
After the considerations made in this subsection, we can better understand \autoref{fig:BC}, in which GR solutions have $\sigma_\eH = 1$ and $\Phi_\eH= 0$.

We note that, $\Lambda$ being determined by $R_\eH$ through \eqnref{eq:cosmconst}, all values of $R_\eH$ for which $\Lambda (R_\eH)=0$ give us Schwarzschild solutions. Therefore, it is possible to find specific sets of global parameters for which more than one value of $R_\eH$ has real solutions for $\sigma _\eH$ and $\Lambda (R_\eH)=0$. In such cases, the Schwarzschild solutions will differ in the value of the conformal factor between the two metrics. Of course, this can also be the case for SAdS and SdS solutions, with each solution having a different value of the cosmological constant. 

Solutions are thus not totally determined by the size of their horizon (i.e., by their mass), but also by $R_\eH$. In \cite{Volkov:2012wp} such cases were already found and denoted ``special'' solutions. We show that 
the values of the global parameters, determine whether multiple Schwarzschild, SAdS or SdS solutions exist or not. In all cases, every possible solution can be found by integrating from the algebraic variety at the Killing horizon.

\subsubsection{Proper bidiagonality and GR solutions}
\label{subsubsec:relation}

Equation \eqref{eq:GRsolutions} shows that GR solutions have $\Phi= 0$ identically.
\begin{equation}
	\text{GR solutions} ~~ \Longrightarrow ~~  \Phi \equiv 0 
	~~ \Longrightarrow ~~ \text{Proper bidiagonality}.
\end{equation} 
The constraint $\Phi_\eH = 0$ defines a finite number of points on $\mathcal{C}$, hence a finite number of proper bidiagonal solutions (we saw three of them in Figure\autoref{fig:BCb}). Proper bidiagonality does not imply GR solutions, as we can see in Figure\autoref{fig:BCb}. The solution indicated by {\Large $\diamond$} has $\Phi_\eH = 0$, $\sigma_\eH = 1, R_\eH = 3/2$ but $R'_\eH \ne0$ at the Killing horizon; therefore, it is not a fixed point (i.e., a GR solution).

However, we can prove the following,
\begin{equation}
	\Phi \equiv 0 ~~ \Longrightarrow ~~ \text{GR solutions}.
\end{equation}
Assuming $\Phi \equiv 0$ implies $\Sigma \equiv \tau $ and
\begin{align}
	0\equiv \Phi'&=\dfrac{\left(r R'+R\right) \sigma '+\sigma  \left(r R''+2 R'\right)-\tau '}{F}-\dfrac{F'\left( \Sigma-\tau \right)}{F^2} \nonumber \\
	&=\dfrac{\left(r R'+R\right) \sigma '+\sigma  \left(r R''+2 R'\right)-\tau '}{F}.
	\label{eq:derPhi1}
\end{align}
Substituting $\sigma '$ and $\tau '$ into $\Phi '$ by using the equations of motion, and substituting $\tau \equiv \Sigma =\sigma \left(r R'+R\right)$, we obtain,
\begin{equation}
\label{eq:derPhi}
\Phi'=\dfrac{\sigma  \left(r R''+2 R'\right)}{F}-\frac{4 (\sigma -1) \left(r R'+R\right) \left(b_1+b_2 \sigma  \left(r R'+R\right)\right)}{r \alpha_1 F},
\end{equation}
where all the metric fields are functions of $\xi $. By hypothesis, $\Phi '$ is identically zero. The second term depends on $\beta $ parameters which are inside the utility functions $\alpha _1$, $b_1$ and $b _2$, whereas the first term does not depend on them. Therefore, in order for $\Phi '$ to be identically zero, independently of the $\beta $ parameters, $R$ must be constant (i.e., $R\equiv R_\eH$) and $\sigma \equiv 1$. This implies $\Sigma \equiv \tau \equiv R \equiv R_\eH$, i.e., $g_{\mu \nu}$ and $f_{\mu \nu}$ are conformal (as we can see in \eqref{eq:EFmetrics}), corresponding to GR solutions. The classification of bidiagonal static and spherically symmetric BH solutions is shown in \autoref{fig:diagram}.

\begin{figure}[t]
\centering
\includegraphics[width=0.75\textwidth]{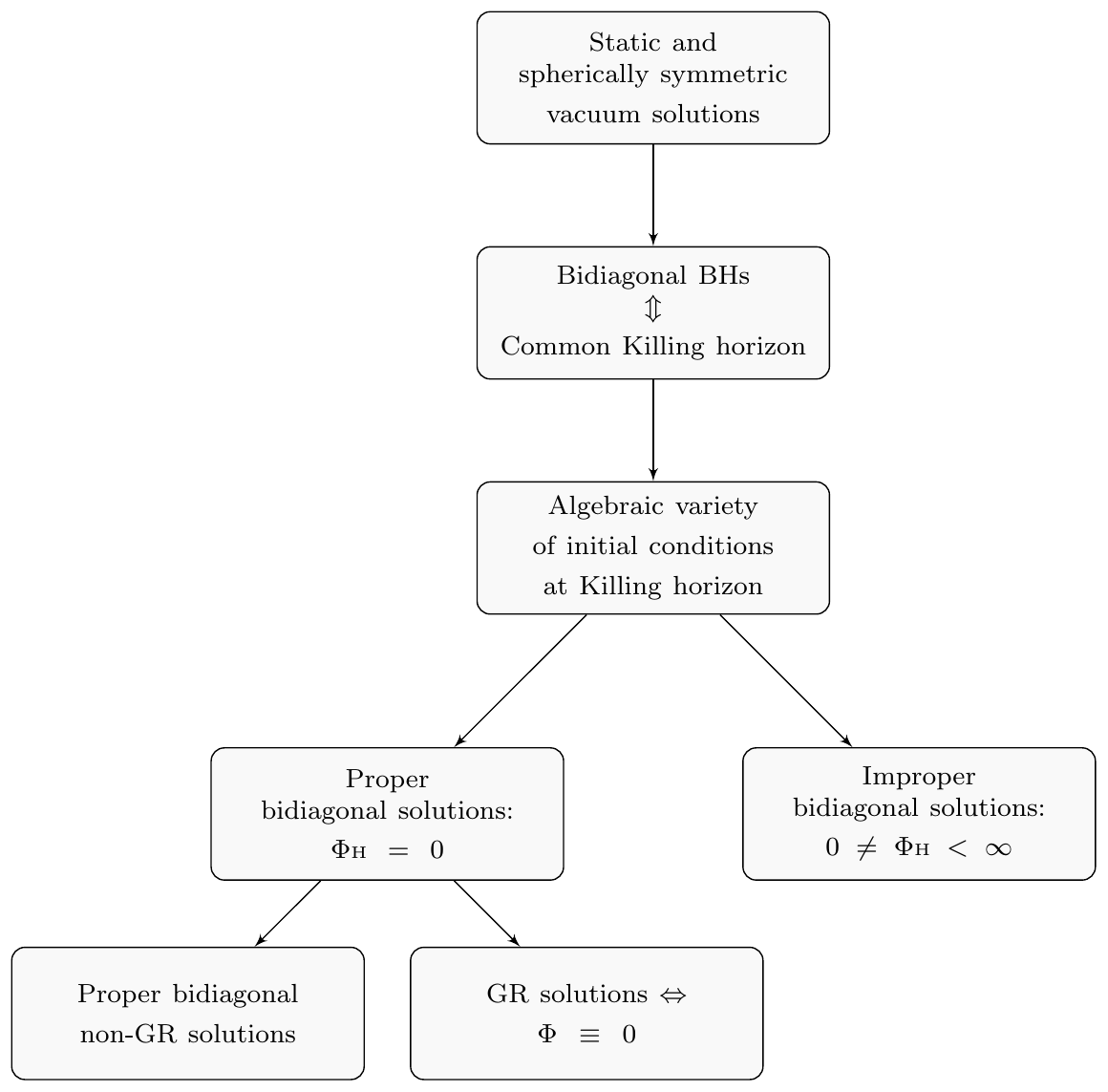}
\caption{The classification of bidiagonal static and spherically symmetric BH solutions. We start by imposing the staticity and spherical symmetry to vacuum solutions. Then, we restrict ourselves to bidiagonal BHs having the common Killing horizon (subsections \ref{subsec:setup} and \ref{subsec:choice}). Going further, these BHs exist only if their metric functions cross the algebraic variety defined by \eqnref{eq:BC} at the Killing horizon. Depending on the value of $\Phi_\eH$ we can have proper and improper diagonal solutions (subsections \ref{subsec:initialconditions} and \ref{subsubsec:ic-at-kh}). Still, the proper bidiagonal solutions contains GR and non-GR solutions, as shown in this subsection.}
\label{fig:diagram}
\end{figure}

\subsubsection{Lyapunov instability of GR solutions}
\label{subsubsec:results}

In this section, we show that GR solutions are \emph{nonlinearly} Lyapunov unstable. We will do this using both numerical investigations and analytical considerations.

We start with the numerical solutions (see \autoref{appendix-D} for more details about the accuracy of these solutions). Using the parameter values given in \autoref{tab:parameters}, the set of the fixed points \eqref{eq:asy-equation} gets the values,
\begin{equation}
\label{eq:RHvalues}
	R_\eH \in \{ ~
	R_\eH^\mathrm{SdS} \approx -8.557,~ R_{\eH,2}^\mathrm{SAdS} \approx -0.6459,~ 
	R_\eH^\mathrm{Schw} = 1, ~ R_\eH^\mathrm{SAdS} \approx 2.633 ~ \}.
\end{equation}
This is consistent with \cite{Volkov:2012wp}. When used as the initial values, these correspond to GR solutions.
First, note that  the value $R_\eH=1$ is exact, guaranteed by the asymptotic flatness conditions in \eqnref{eq:AFconditions}. 
Second, there are two negative values for $R_\eH$. For GR solutions, $R\equiv R_\eH$; For non-GR solutions, instead, given that $R$ cannot change sign (since this implies $\det S=0$), if $R_\eH$ is negative, then $R$ will be negative for all $\xi$ (i.e., all $r$). The same holds for $\sigma $ which is positive (because $\sigma _\eH=1$). Then, from \eqnref{eq:EFmetrics-2}, we can deduce that different signs of $R_\eH$ mean different \emph{disconnected} square root branches, belonging to different model and different phase-spaces. Note that these square root branches are different from the two branches of solutions obtained by choosing the value of $\sigma _\eH$.
A negative $R$ is equivalent to retaining the absolute value of $R$ in \eqnref{eq:EFmetrics-2} together with $\pm_3=\pm_4=-1$.
We emphasize that the equations of motions \eqref{eq:EoM} are valid only for $R>0$. However, the transformation,
\begin{equation}
\label{eq:substitutions}
R\rightarrow -R \qquad \beta _1 \rightarrow -\beta _1 \qquad \beta _3 \rightarrow -\beta _3,
\end{equation}
is equivalent to changing the square root branch, but for different $\beta$ parameters and different asymptotic flatness conditions (i.e., in a different phase space).

Choosing $R_\eH=R_\eH^\mathrm{Schw}\equiv 1$, we obtain the Schwarzschild solution. Perturbations around it can be obtained by considering $R_\eH=R_\eH^\mathrm{Schw}(1+\epsilon)$ and selecting the same branch for $\sigma _\eH$ as the Schwarzschild solution (Figure\autoref{fig:BCb}). Hence, these perturbations are improper bidiagonal (like the perturbations around SAdS and SdS, treated later). All found perturbations around Schwarzschild solutions have a diverging $\sigma (r)$ at some finite radial coordinate $r=\hat{r}$; the bigger $\epsilon $, the smaller $\hat{r}$. A diverging $\sigma(r)$ means that $\det(S)$ is diverging, i.e., $\det(S^{-1})=0$. Therefore, these perturbations are not physical solutions. Nonetheless, we can introduce the change of variable $\sigma (r)\rightarrow s(r)\coloneqq \sigma (r)^{-1}$ and solve for $s(r)$, which is 0 when $\sigma (r)$ diverges. This strategy allows us to study the solution asymptotically using 
numerical integration, despite the fact that perturbations are not physical.

\autoref{fig:Schwinstability} shows the difference between the analytical Schwarzschild solution,
\begin{figure}[t]
	\centering
	\includegraphics[width=0.95\textwidth]{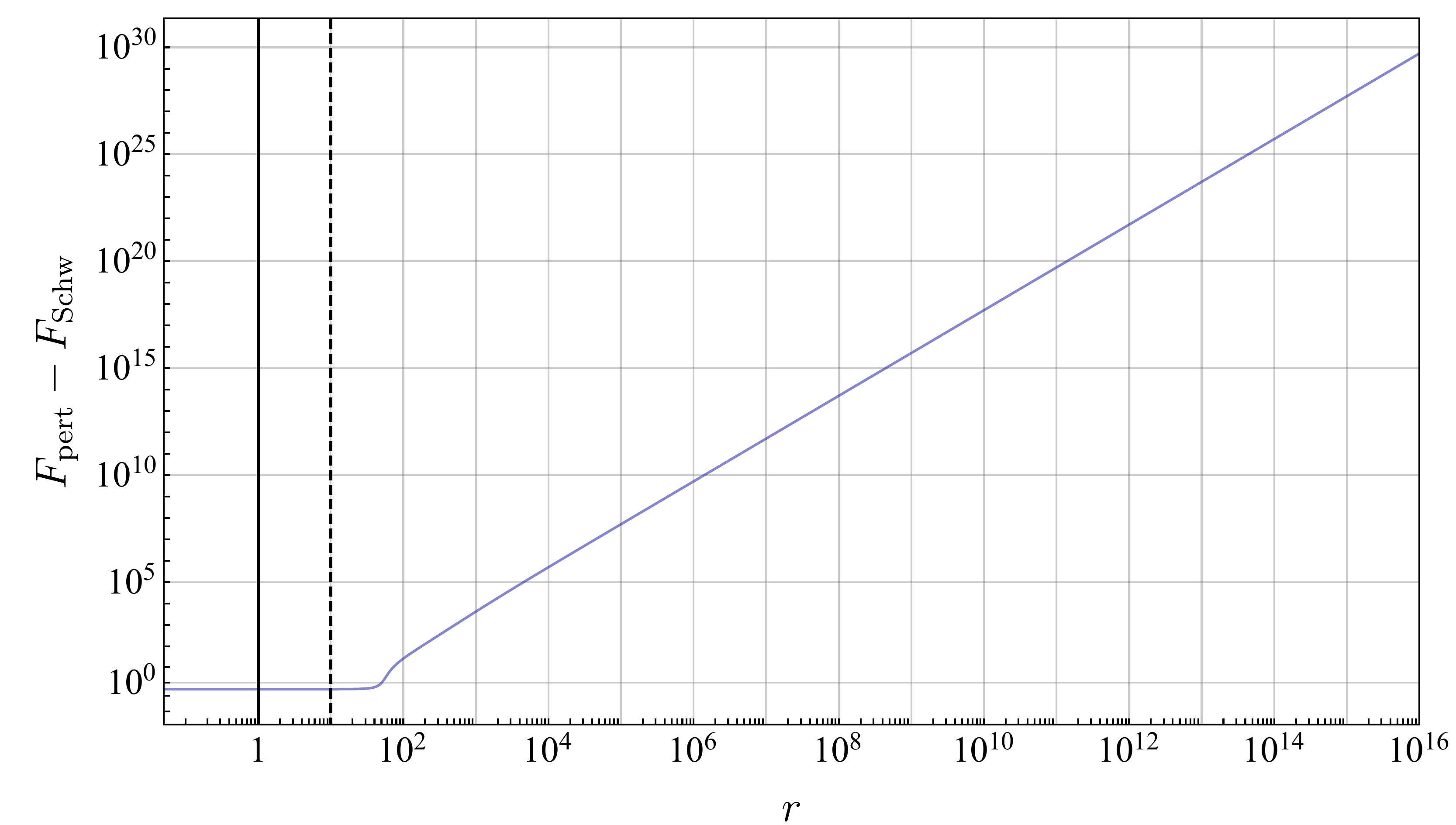}
	\caption{Difference between the analytical Schwarzschild solution, with $R_\eH=R_\eH^\mathrm{Schw}=1$, and the improper bidiagonal perturbation around it, with $R_\eH=1.03$. $ F_\mathrm{pert}-F_\mathrm{Schw}$ grows indefinitely, showing that the Schwarzschild solution is Lyapunov unstable. 
	Although the solution shows a behaviour $\propto r^2$, it is different from SAdS, which is analysed below in the main text. Note that $F_\mathrm{pert}$ exhibits a singularity in the $\sigma$ field (not shown here) at finite $r$, therefore this is not a physical solution.
	The vertical solid black line indicates the Killing horizon; the vertical dashed black line indicates the Compton wavelength of the massive mode. The $y$ axis is in $\op{sign}(y)\log_{10}(1+|y|)$ scale, which is $\approx y$ for small $y$ and $\approx \op{sign}(y)\log_{10}(|y|)$ for large $y$.
	}
	\label{fig:Schwinstability}
\end{figure}
\begin{figure}[t]
	\centering
	\includegraphics[width=0.95\textwidth]{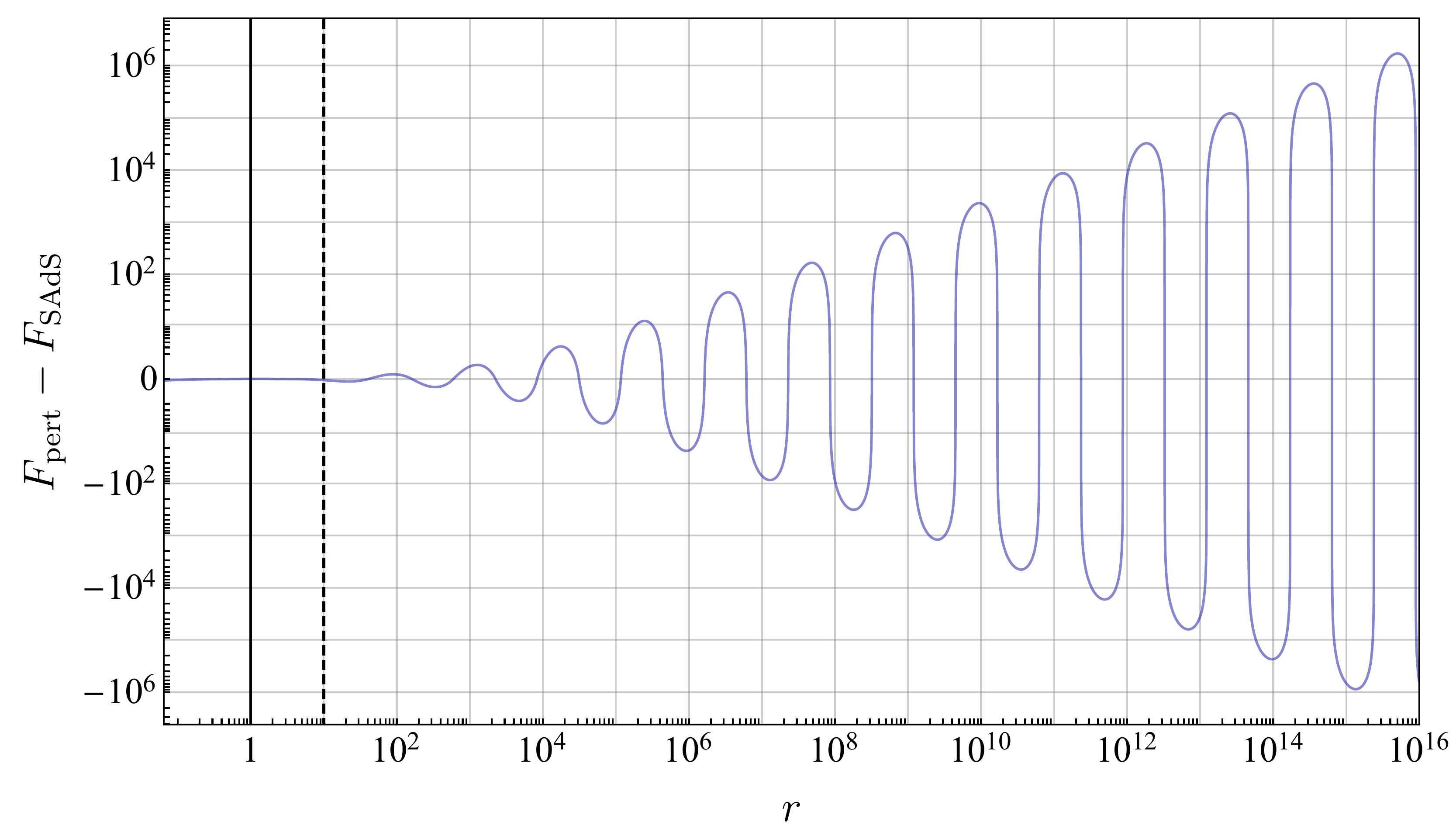}
	\caption{Difference between the analytical SAdS solution, with $R_\eH=R_\eH^\mathrm{SAdS}\approx 2.633$, and the improper bidiagonal perturbation around it, with $R_\eH=3$. $ F_\mathrm{pert}-F_\mathrm{SAdS}$ oscillates with a growing amplitude, showing that the SAdS solution is Lyapunov unstable. 
	The envelope of the oscillations is proportional to $r^{1/2}$.
	The vertical solid black line indicates the Killing horizon; the vertical dashed black line indicates the Compton wavelength of the massive mode. The $y$ axis is in $\op{sign}(y)\log_{10}(1+|y|)$ scale, which is $\approx y$ for small $y$ and $\approx \op{sign}(y)\log_{10}(|y|)$ for large $y$.
	}
	\label{fig:SAdSinstability}
\end{figure}
\begin{figure}[t]
	\centering
	\includegraphics[width=0.95\textwidth]{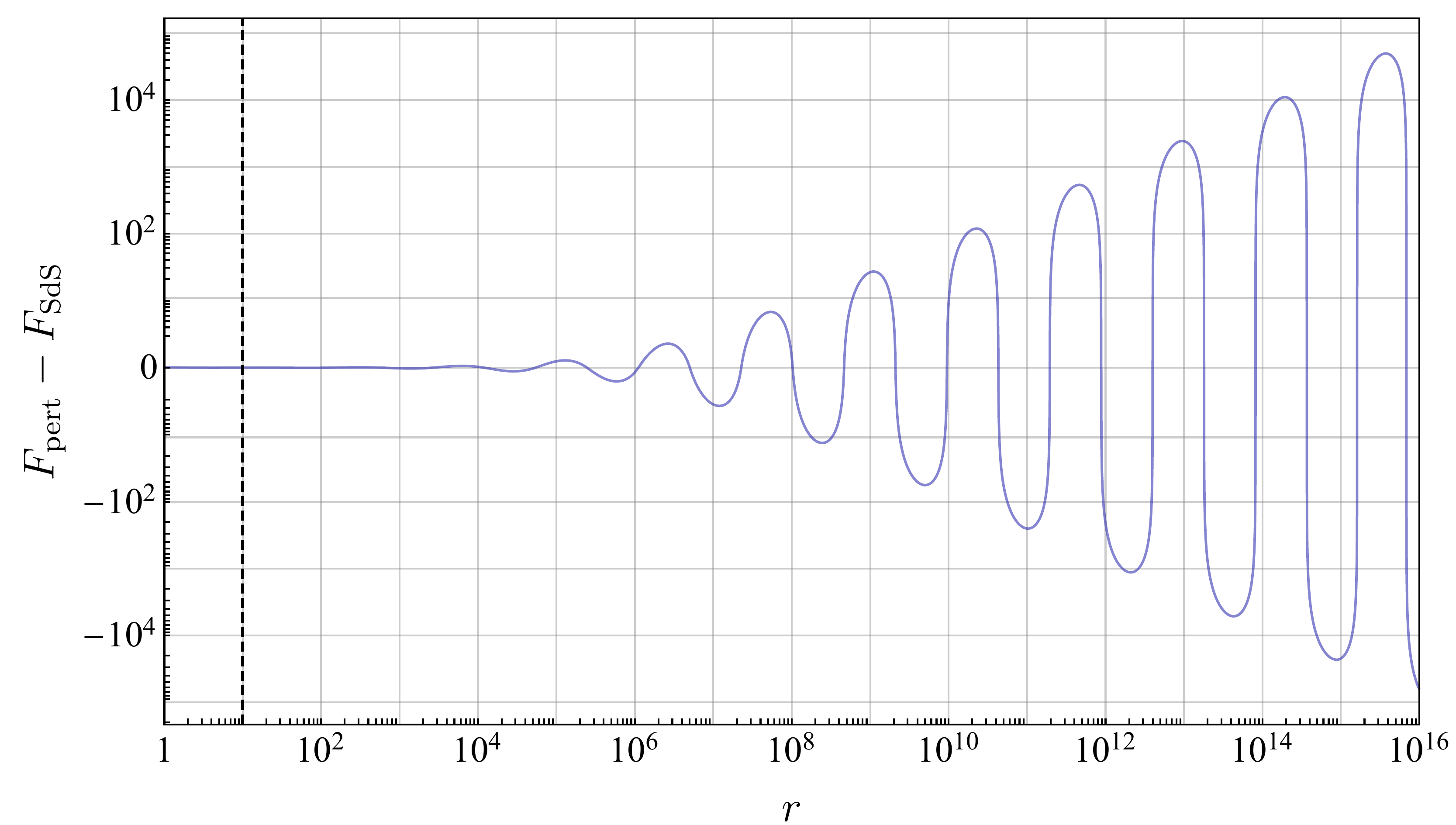}
	\caption{Difference between the analytical SdS solution in \eqnref{eq:SAdS}, with $R_\eH=R_\eH^\mathrm{SdS}\approx -8.557$, and the improper bidiagonal perturbation around it, with $R_\eH=R_\eH^\mathrm{SdS}(1+\epsilon)$, with $\epsilon =10^{-3}$. $ F_\mathrm{pert}-F_\mathrm{SdS}$ oscillates with a growing amplitude, showing that the SdS solution is Lyapunov unstable. 
	The envelope of the oscillations is proportional to $r^{1/2}$. Note that $F_\mathrm{pert}$ exhibits a singularity in the $\sigma$ field (not shown here) at finite $r$. Therefore, this is not a physical solution.
	The vertical dashed black line indicates the Compton wavelength of the massive mode. The  $y$ axis is in $\op{sign}(y)\log_{10}(1+|y|)$ scale, which is $\approx y$ for small $y$ and $\approx \op{sign}(y)\log_{10}(|y|)$ for large $y$.
	}
	\label{fig:SdSinstability}
\end{figure}
\begin{equation}
\label{eq:Schw}
F_\mathrm{Schw}(r)=1-\dfrac{r_\eH}{r}
\end{equation}
with $R_\eH=1$ and the numerical perturbation around it, $F_\mathrm{pert}$, with $R_\eH=1.03$. The fact that the difference is increasing with growing $r$, is true for all found perturbations around Schwarzschild solutions, regardless of the values of the global parameters. This shows that the Schwarzschild solution is Lyapunov unstable. 
Contrary to \cite{Brito:2013xaa} where asymptotically flat solutions other than the Schwarzschild were found, our result shows that Schwarzschild solution is the unique asymptotically flat solution (see also \autoref{appendix-A}).

The numerical integration in \autoref{fig:Schwinstability} extends to $r=10^{16}$. However, we will show later that there is strong evidence that the instability of Schwarzschild (and, more generally, of GR solutions) persists in the limit $r\rightarrow \infty$, independent of the global parameter values. 

The solutions with $R_\eH=R_\eH^\mathrm{SAdS}\approx 2.633$ and $R_\eH=R_{\eH,2}^\mathrm{SAdS}\approx -0.6459$ correspond to SAdS solutions. Perturbations around these two solutions behave very similarly, so we describe only $R_\eH=R_\eH^\mathrm{SAdS}$, belonging to the same phase space as the Schwarzschild solution. Plotting the metric functions, the SAdS solution na\"ively appears asymptotically stable. However, the study of the difference $ F_\mathrm{pert}-F_\mathrm{SAdS}$ at large radii, which matters for Lyapunov stability, shows that it is unstable. Here, $F_\mathrm{SAdS}$ indicates the analytic SAdS solution,
\begin{equation}
\label{eq:SAdS}
F_\mathrm{SAdS}(r)=1-\dfrac{r_\eH}{r}-\dfrac{\Lambda (R^\mathrm{SAdS}_\eH)}{3}\left( r^2-\dfrac{r_\eH^3}{r} \right), \quad
\Lambda(R^\mathrm{SAdS}_\eH) < 0,
\end{equation}
which is 0 when $r=r_\eH$. 
The difference between the improper bidiagonal perturbation around SAdS with $R_\eH=3$ and the analytic SAdS solution with $R_\eH=R_\eH^\mathrm{SAdS}$ is plotted in \autoref{fig:SAdSinstability}. $ F_\mathrm{pert}-F_\mathrm{SAdS}$ oscillates with increasing amplitude at radii larger than the Compton wavelength of the massive mode ($\lambda =10$ in our case). A similar asymptotic behaviour around SAdS was found in \cite{Sushkov:2015fma}.

Such a result shows that SAdS is an unstable solution.
We emphasise that all the found perturbations around SAdS have a non-diverging $\sigma (r)$ and a monotonically increasing $R(r)r$ (which is a necessary condition to have a spherically symmetric spacetime \cite{choquet2008general}). Therefore, they do not suffer from any pathological behaviour and are physically acceptable solutions. 

We now turn to the analysis of SdS solution, with $R_\eH=R_0^\mathrm{SdS}\approx -8.557$ (as pointed out, belonging to another model). Since the SdS solution has both a Killing and a cosmological horizon, we can perturb it at either horizon. When perturbing the solution at the Killing horizon, one could expect it to have a cosmological horizon close to the corresponding SdS cosmological horizon. However, all the found perturbations around the Killing horizon displayed a diverging $\sigma(r)$, before reaching a cosmological horizon. This means that they are not physically acceptable. These perturbations, defined on a compact domain, were found already in \cite{Volkov:2012wp}, where it was stated that there are no other perturbations to SdS solutions. However, solutions obtained when perturbing SdS at the cosmological horizon do not have a compact domain, since they are defined when $r\rightarrow \infty$. On the other hand, integrating inwards from the cosmological horizon, $\sigma (r)$ diverges before reaching a Killing horizon for all the found perturbations, showing that also these solutions are non-physical.  Nevertheless, in order to study the asymptotic behaviour of perturbations around SdS solution, we integrate numerically from the cosmological horizon, imposing initial conditions lying on the algebraic variety. Note that $R_\eC = R_\eH$ for the SdS solution (because $R$ is constant), of course corresponding to different $r_\eH$ and $r_\eC$.

The SdS solution with Killing horizon $r_\eH=1$ has the cosmological horizon at $r = r_\eC\approx 3.128$. 
In \autoref{fig:SdSinstability}, we show the difference between the analytical SdS solution in \eqnref{eq:SAdS} with $R_\eH=R_0^\mathrm{SdS}\approx -8.557$ and $\Lambda(R^\mathrm{SdS}_\eH) > 0$, and a perturbation around it, with $R_\eH=R_\eH^\mathrm{SdS}(1+\epsilon)$, with $\epsilon =10^{-3}$. Having the same qualitative features as \autoref{fig:SAdSinstability}, we conclude that also the SdS solution is Lyapunov unstable.
Apart from the fact that amplitudes always \emph{grow} with increasing radii,
numerical integration shows that the amplitude of the oscillations is proportional to the size of the perturbation around $R_\eH$.

We will now show how the numerically found Lyapunov instabilities of GR solutions can be motivated by looking at the equations of motion in \eqnref{eq:shiftedEoM}. Thanks to the autonomy of the system, we can study the 3-dimensional phase space section defined by $\rho =0$, corresponding to infinite radii.
When $\rho =0$, the right-hand sides of the first two equations in \eqnref{eq:shiftedEoM} are finite if the conditions in \eqnref{eq:necessaryFP} hold. That is, if we have a fixed point, corresponding to a GR solution.

For differentiability of the solutions, a non-GR solution will thus necessarily have a diverging gradient for infinitesimally small $\rho$, i.e., infinitely large $r$. Therefore it will not tend to a fixed point, since for this to be possible, the modulus of its gradient in the phase space should tend to 0. Hence, the non-GR solution will depart more and more from the GR solutions when $r\rightarrow \infty$, strongly indicating that GR solutions are Lyapunov unstable, independently of the global parameter values.

Since we have not found a Lyapunov instability function, we do not have a formal proof of the instability of GR solutions. Nevertheless, the reasoning above combined with the empirical fact that all numerical solutions we have found show the instability, and the phase space analysis that we are going to pursue, are strong hints that GR solutions are Lyapunov unstable.

\begin{figure}[t]
\centering
	\subfloat[$\Lambda=0$. All the trajectories are Schwarzschild solutions with different horizons. The black trajectory is the Minkowski solution.]{\includegraphics[width=0.31\textwidth]{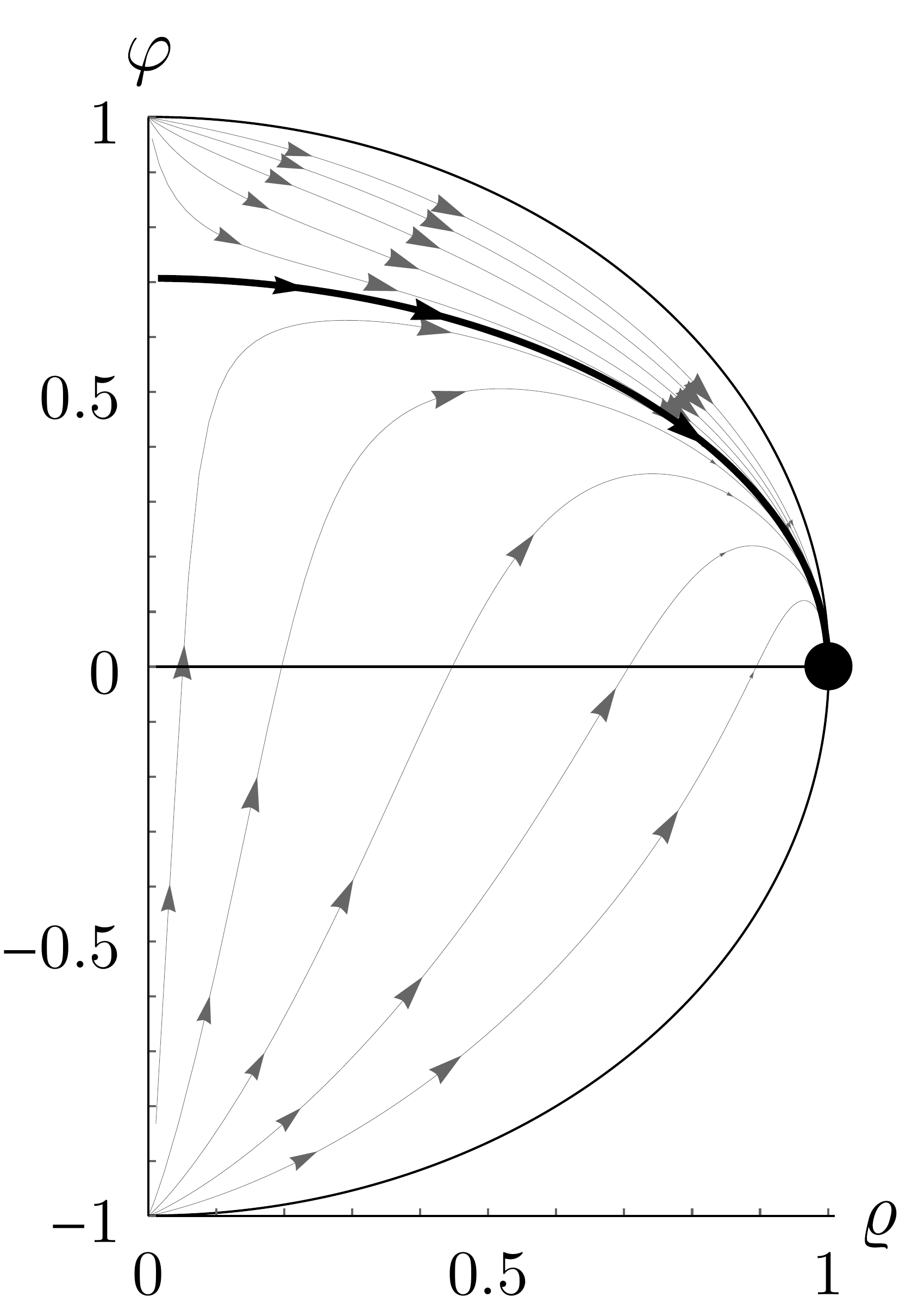}\label{fig:Schw-phase}}
	\hspace{0.3cm}
	\subfloat[$\Lambda=-1$. All the trajectories are SAdS solutions with different horizons. The black trajectory is the AdS solution.]{\includegraphics[width=0.31\textwidth]{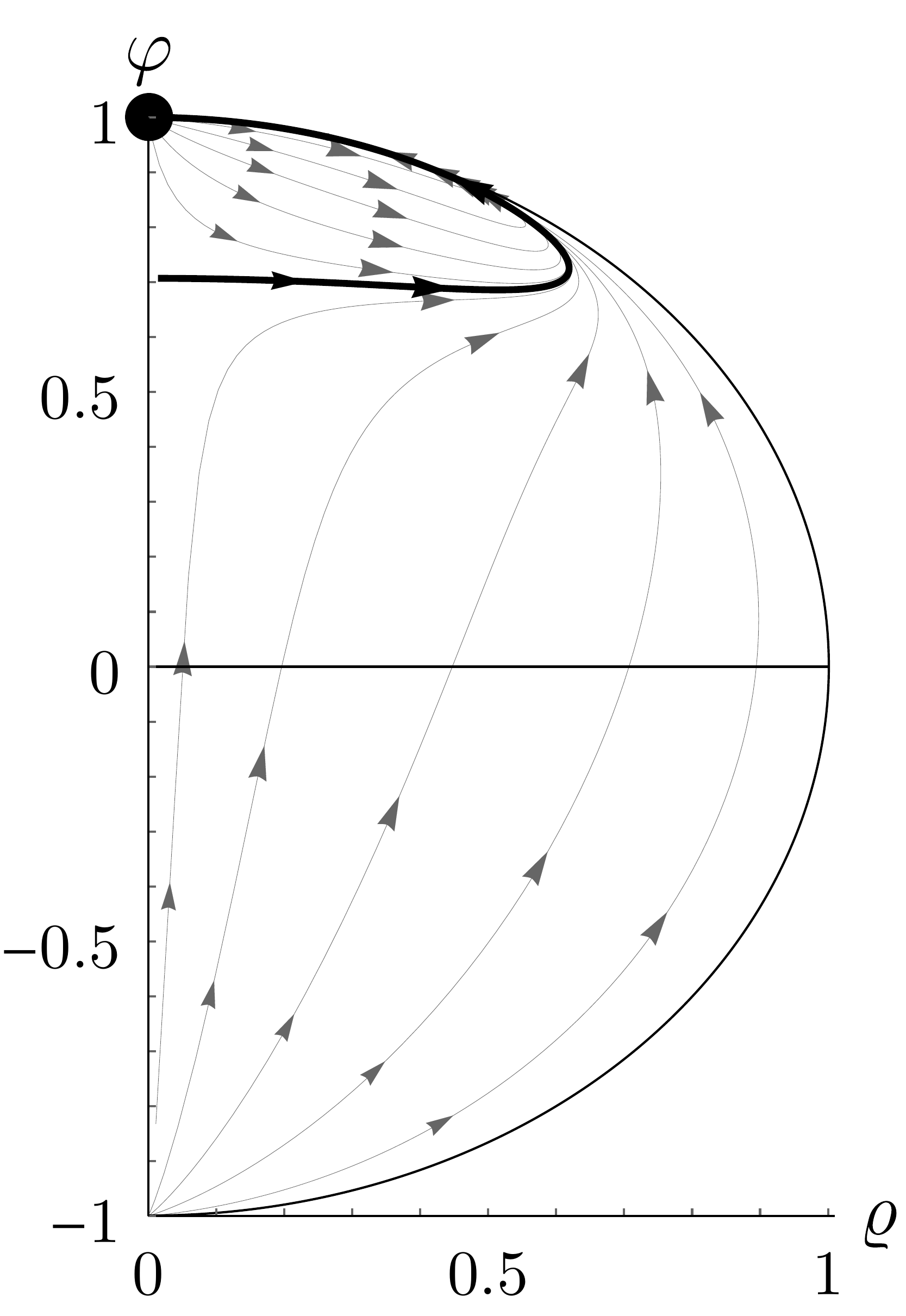}\label{fig:SAdS-phase}}
	\hspace{0.3cm}
	\subfloat[$\Lambda=1$. All the trajectories are SdS solutions with different horizons. The black trajectory is the dS solution.]{\includegraphics[width=0.31\textwidth]{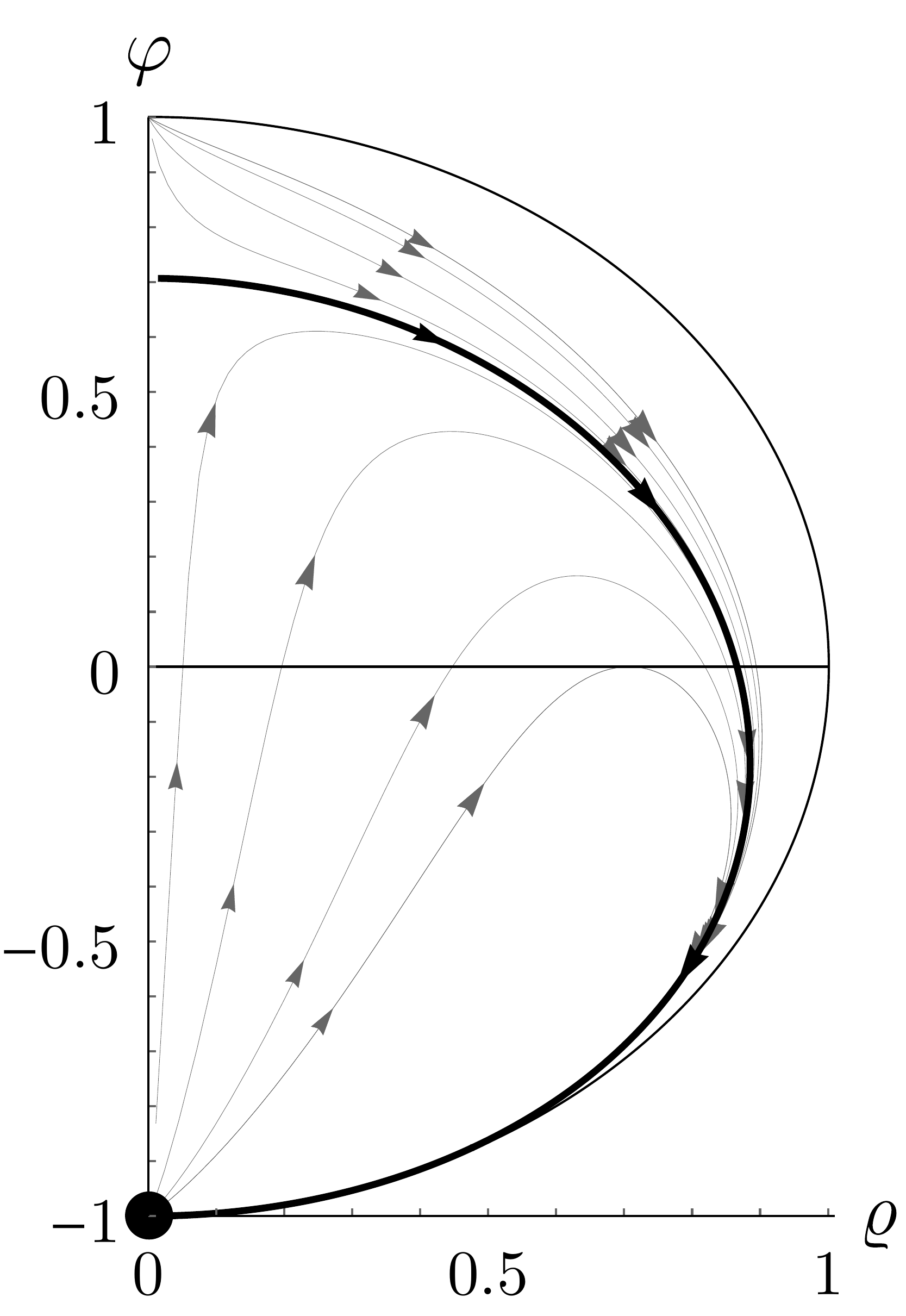}\label{fig:SdS-phase}}
	\caption{
	Compactified phase spaces in GR for different values of the cosmological constant: $\Lambda =0,-1,1$ respectively. The black dots are the points reached asymptotically by \emph{all} the trajectories for $r\rightarrow \infty $. The event and cosmological horizons are located at the $\varrho $ where trajectories cross the line $\varphi =0$. The black thicker trajectories do not have a horizon ($r_\eH=0$), neither a curvature singularity at $r=0$, whereas the trajectories above the black ones do not have any horizon, but a naked singularity at $r=0$.
	}
	\label{fig:GR-phase-spaces}
\end{figure}

In the following, we show this instability by directly plotting their phase space portraits. Before doing that, we introduce our method by showing GR phase space portraits. In GR, for a static and spherically symmetric metric,
\begin{equation}
\label{eq:GRmetric}
	g_{\mu\nu}=\op{diag}\left[\;
		-F(r),\, F(r)^{-1},\, r^2,\, r^2\sin ^2(\theta) \; 
		\right],
\end{equation}
Einstein equations reduce to,
\begin{equation}
\label{eq:GR}
\frac{\dd F}{\dd r}=\dfrac{1-F}{r}-\Lambda \; r,
\end{equation} which can be formally rewritten as an autonomous system of two ODEs,
\begin{equation}
\label{eq:GRauto}
		r'(\xi) = 1, \qquad 
		F'(\xi) = \dfrac{1-F(\xi)}{r(\xi)}-\Lambda \, r(\xi).
\end{equation}
Following the procedure described for example in \cite{jordan2007nonlinear}, we introduce variables which compactify the phase space,
\begin{equation}
	\varrho = \dfrac{r}{\sqrt{1+r^2+F^2}}, 
	\qquad \varphi = \dfrac{F}{\sqrt{1+r^2+F^2}}. \\
\end{equation}
Since $\varrho \geq 0$ and $\varrho ^2+\varphi^2 \leq 1 $, with the last equality valid only in each of the limits $r\rightarrow \infty$ and $F\rightarrow \infty$, the whole phase space is compactified onto the unitary right half disk by stereographic projection. 

Depending on the value of $\Lambda $, we will have different phase spaces, plotted in \autoref{fig:GR-phase-spaces}.
In each of the three phase spaces, the family of trajectories is parametrised by $r_\eH$ (or, equivalently, by the mass of the BH), which is the only parameter we can vary to obtain a new solution in GR.
\emph{All} trajectories will eventually reach asymptotically stable fixed points, indicated by black dots in \autoref{fig:GR-phase-spaces}, displaying the Lyapunov stability of the solutions.
\begin{figure}[t]
	\vspace{-3mm}
	\subfloat[Top view around Schwarzschild.]{%
		\phasePlot{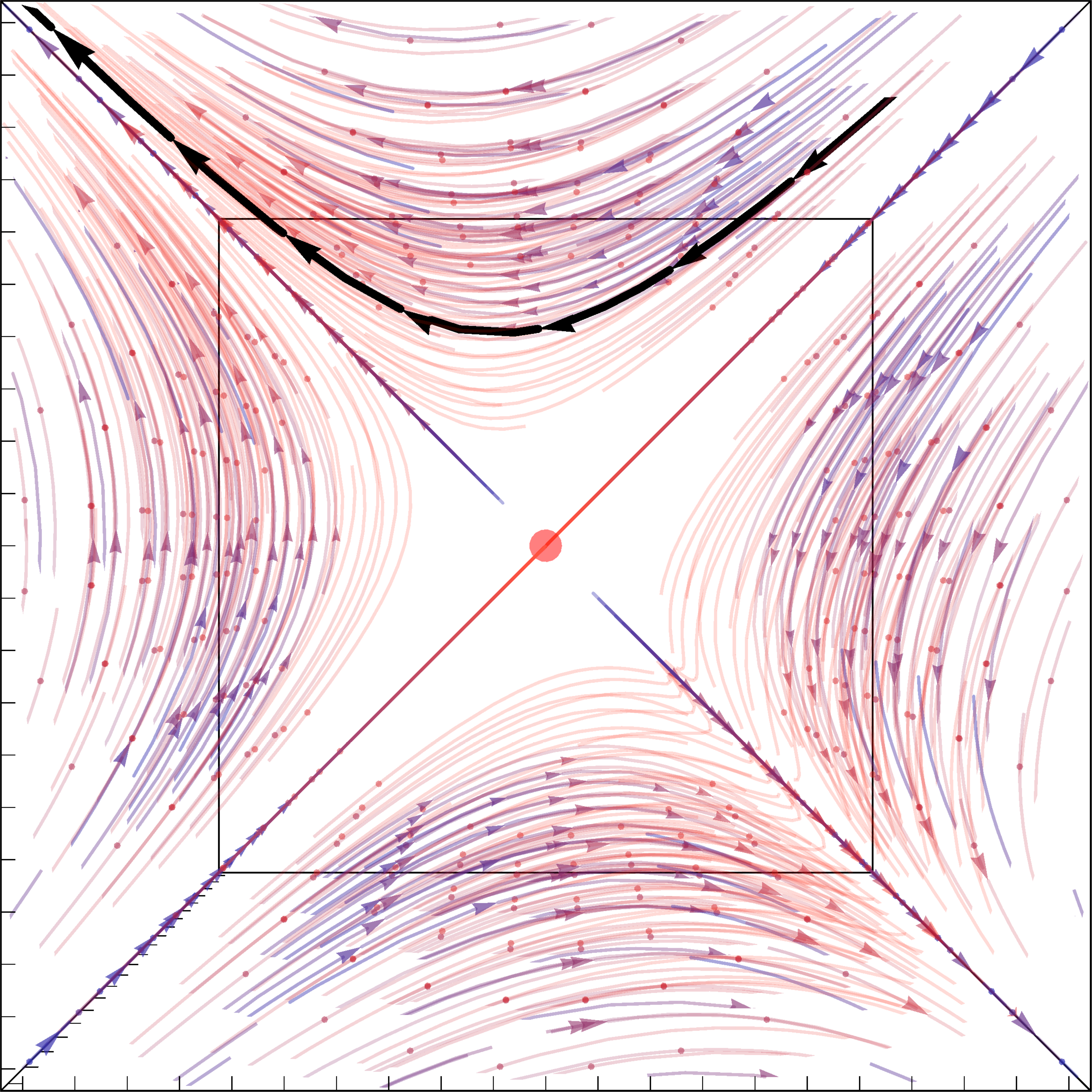}{$\bar{R}$}{$10^{-11}$}{$\bar{\sigma}$}{$10^{-4}$}
		\label{fig:BMSchwPhaseTop}
	}
	\subfloat[Front view around Schwarzschild.]{%
		\phasePlot{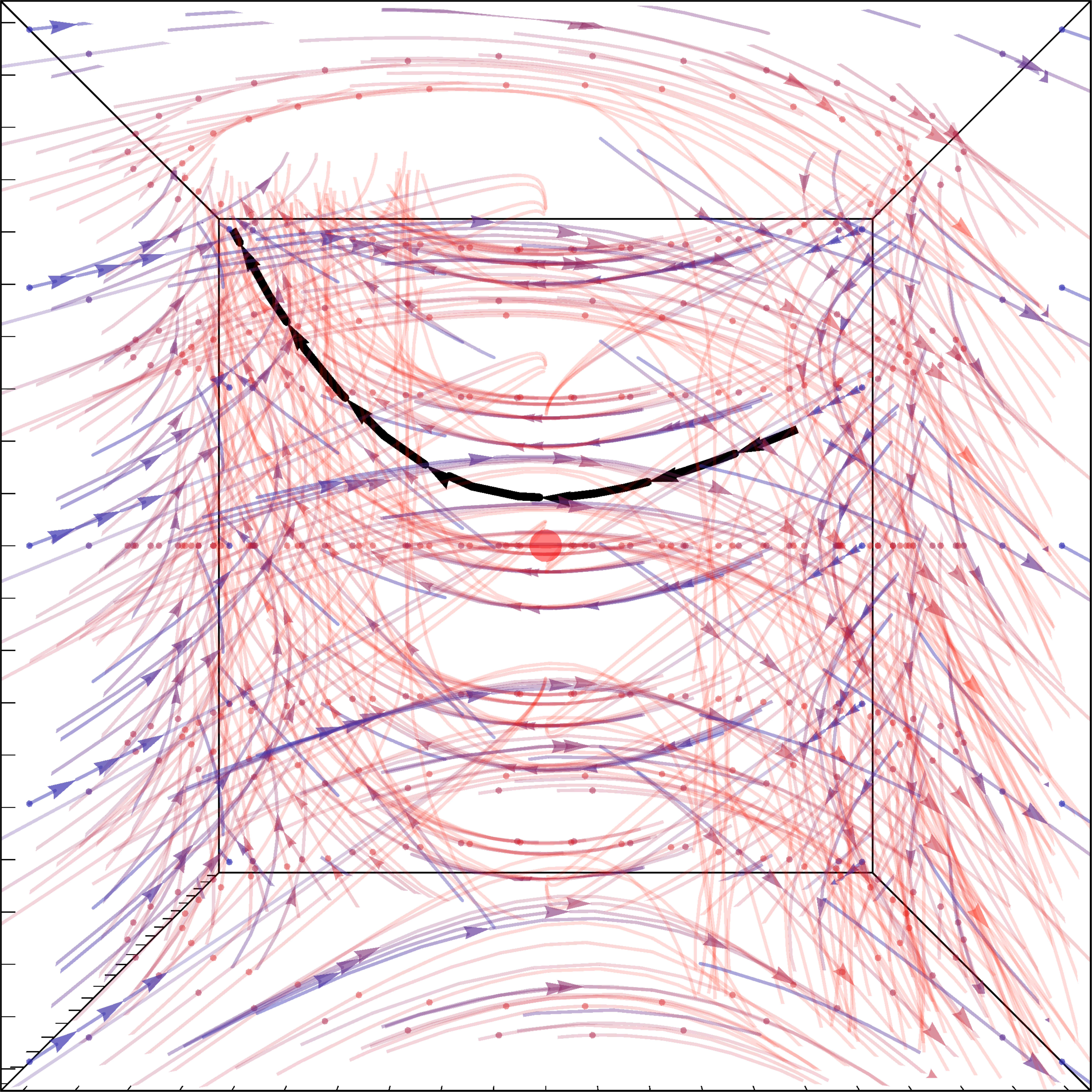}{$\bar{R}$}{$10^{-11}$}{$\bar{F}$}{$10^{-4}$}
		\label{fig:BMSchwPhaseFront}
	}
	\\
	\subfloat[Top view around SAdS.]{%
		\phasePlot{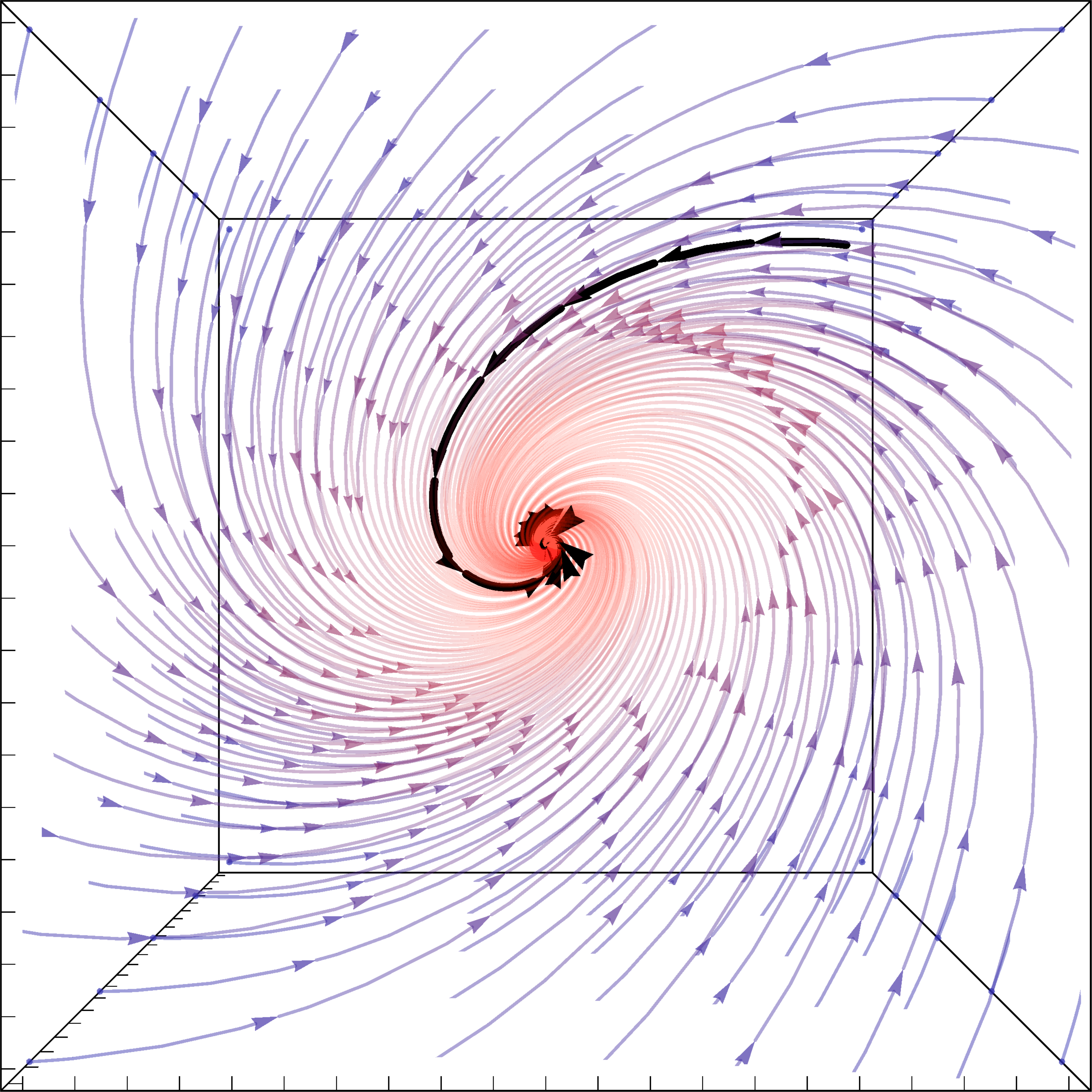}{$\bar{R}$}{$10^{-11}$}{$\bar{\sigma}$}{$10^{-8}$}
		\label{fig:BMSAdSPhaseTop}
	}
	\subfloat[Front view around SAdS.]{%
		\phasePlot{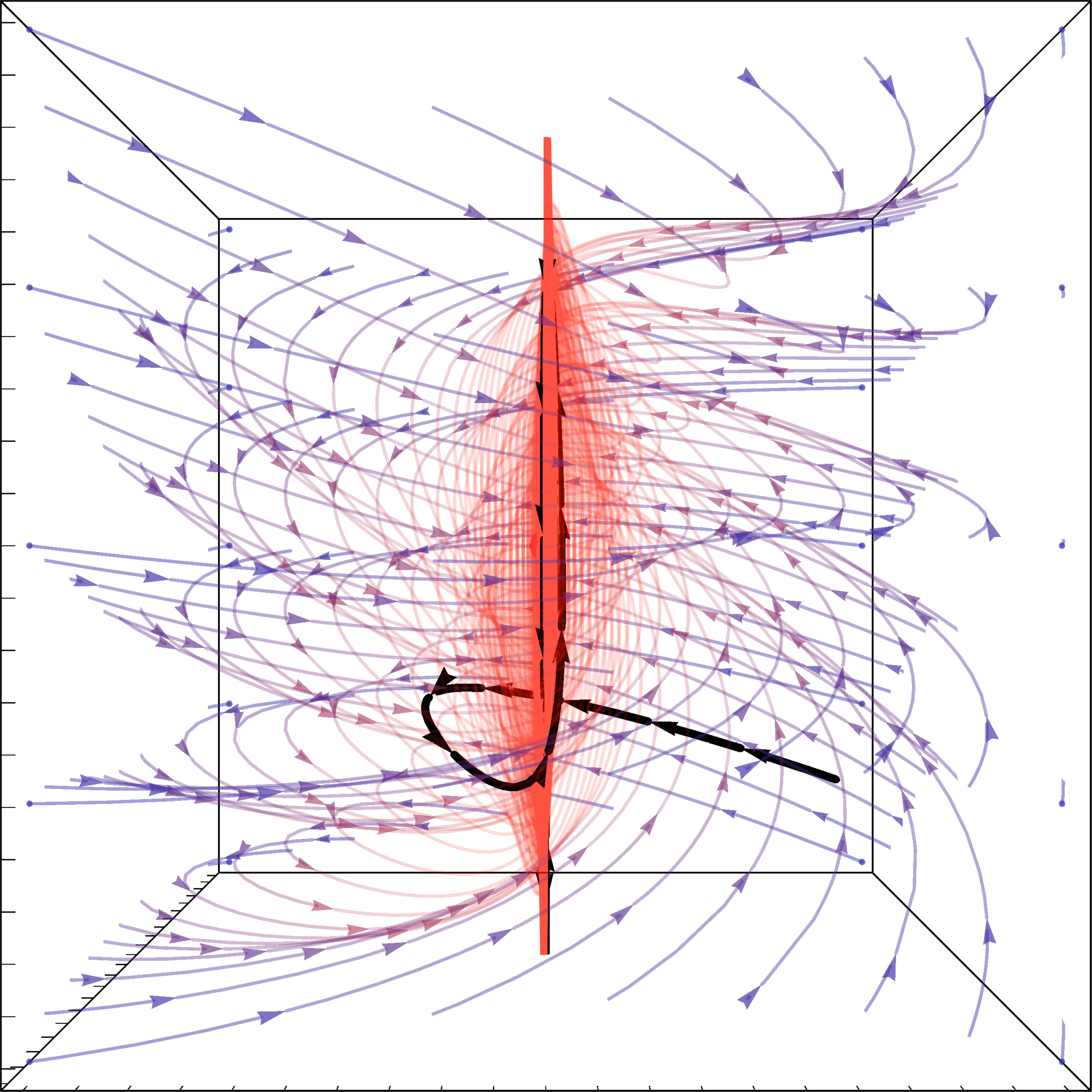}{$\bar{R}$}{$10^{-11}$}{$\bar{F}$}{$10^{2}$}
		\label{fig:BMSAdSPhaseFront}
	}
	\\
	\subfloat[Top view around SdS.]{%
		\phasePlot{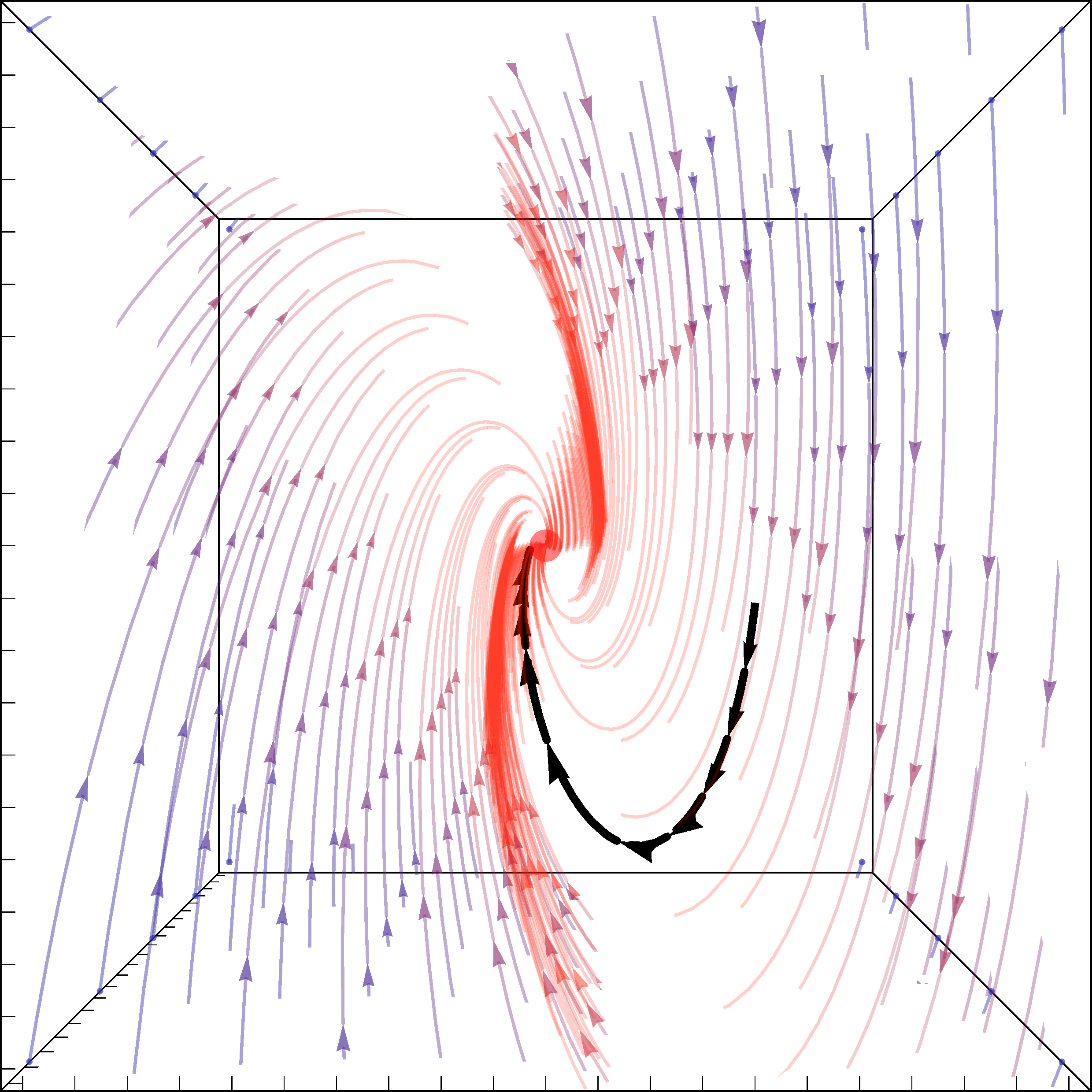}{$\bar{R}$}{$10^{-8}$}{$\bar{\sigma}$}{$10^{-9}$}
		\label{fig:BMSdSPhaseTop}
	}
	\subfloat[Front view around SdS.]{%
		\phasePlot{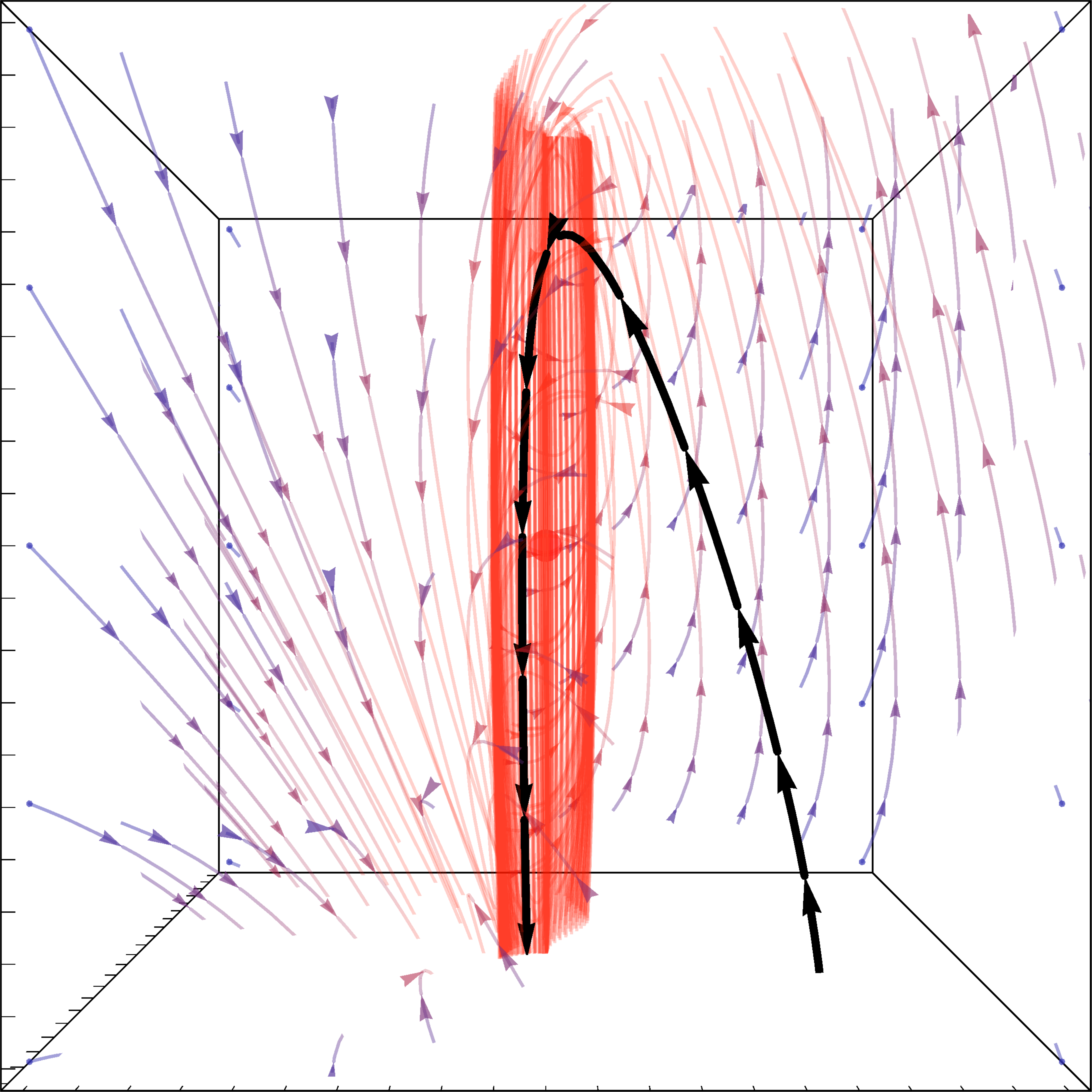}{$\bar{R}$}{$10^{-8}$}{$\bar{F}$}{$10^{-2}$}
		\label{fig:BMSdSPhaseFront}
	}
	\caption{Phase space portraits around three GR solutions. See subsection \ref{subsubsec:results} for details.}
	\label{fig:BM-phase-spaces}
\end{figure}\clearpage{}
In HR bimetric theory, varying only $r_\eH$ keeping $R_\eH$ fixed and equal to one of the values in \eqnref{eq:RHvalues}, we will obtain a family of Schwarzschild solutions, SAdS solutions or SdS solutions (depending on the constant value of $R_\eH$). The solutions belonging to the same family will asymptotically approach each other, exactly as in GR.

On the other hand, varying the ratio between the lengths of the common Killing horizon with respect to $f_{\mu \nu}$ and $g_{\mu \nu}$, i.e., $R_\eH$, \emph{does} change the asymptotic behaviour of solutions belonging to the same phase space (i.e., to the same model), as we saw explicitly in numerical solutions. This fact can be appreciated from the bimetric phase spaces shown in \autoref{fig:BM-phase-spaces}, where the vector field in \eqnref{eq:shiftedEoM} is plotted. That is, we can infer the behaviour of the solutions without the need for numerical integration. Note that the considered Schwarzschild and SAdS solutions belong to the same phase space, whereas the SdS solutions belongs to a separate phase space, as we already highlighted.

The ODE phase space is 4-dimensional, with variables $\left( \rho,\bar{F},\bar{\sigma },\bar{R}\right) $. Using a color code for one dimension, we can make 3-dimensional plots. Here, the $\rho$ axis is represented by the color code. In \autoref{fig:BM-phase-spaces}, we take 2-dimensional sections of the 3-dimensional plots in order to emphasise the features we are interested in. Each plot shows the phase space around one GR solution, indicated by a red sphere located at $\left(0,0,0\right)$.

In the first row, we show the phase space around the Schwarzschild solution, in the second row we show the phase space around the SAdS solution, in the third row we show the phase space around the SdS solution. The first column contains the $\left( \bar{R},\bar{\sigma} \right)$ projections in which we are looking at the phase space ``from the top'', the second column contains the $\left( \bar{R},\bar{F} \right)$ projections in which we are looking the phase space ``from the front''.

The value of $\rho =1/r$ is indicated by the colour of the trajectories: the bluer the trajectory, the larger the $\rho$ (smaller $r$), the redder the trajectory, the smaller the $\rho$ (larger $r$). The asymptotic behaviour of the solutions is thus described by the red regions. In each plot, one arbitrary reference solution is highlighted in black.

From Figures\autoref{fig:BMSchwPhaseTop} and\autoref{fig:BMSchwPhaseFront}, we see that the Schwarzschild solution is Lyapunov unstable. In Figure\autoref{fig:BMSchwPhaseTop} we clearly see the saddle point structure of the phase space. There is one stable direction (left-bottom to right-top diagonal) and one unstable direction (left-top to right-bottom diagonal). Looking at the black reference solution in Figure\autoref{fig:BMSchwPhaseFront}, we see that the phase space flow is departing from the Schwarzschild solution also along the $\bar{F}$ axis when $r\rightarrow \infty$. Note also that the plots represent a very small region around the Schwarzschild solution (as evident from the frame scales). Very small perturbations in $R_\eH$ thus lead to completely different asymptotic behaviours. 

We comment on the last four plots together, because they share the same qualitative features in pairs. In Figures\autoref{fig:BMSAdSPhaseTop} and\autoref{fig:BMSdSPhaseTop}, we see the top views of the phase spaces around SAdS and SdS solutions, respectively. They both look asymptotically stable, since the flow seems to be approaching the point $\left( 0,0,0 \right)$ when $r\rightarrow \infty $ (i.e., when the flow becomes redder). This is radically different from the Schwarzschild case, where the flow was showing a saddle point structure.
However, by looking at the black reference solutions in the front views of Figures\autoref{fig:BMSAdSPhaseFront} and\autoref{fig:BMSdSPhaseFront}, it is evident that SAdS and SdS solutions are unstable. Although trajectories get closer to $\left( 0,0,0 \right)$ in the $\left( \bar{R},\bar{\sigma} \right)$ plane, they are oscillating around the SAdS and SdS solutions with a \emph{growing} amplitude along the $\bar{F}$ axis. These results are consistent with what we found numerically, i.e., perturbations around Schwarzschild solutions diverge from it and perturbations around SAdS and SdS solutions diverge from them with growing amplitude oscillations.

As we already pointed out, perturbations around the Schwarzschild and SdS solutions are not physically acceptable because $\det(S^{-1})=0$ for some finite $r$. Phase space portraits tell us that, even if well-behaved perturbations were found, they would not approach GR solutions; they would follow the unstable flow in the phase space.

Of course, plots in \autoref{fig:BM-phase-spaces} describe trajectories only up to some finite $r$. However, the system is autonomous and therefore the qualitative behaviour of the flow is not changed when $r$ grows. Although not a formal proof, this is a strong indication of instability.

Finally, we have also studied Lyapunov stability of GR solutions in the dRGT massive gravity limit, $\kappa \to \infty$, and we found analogous results. However, for some values of the global parameters, the Schwarzschild solution happens to be Lyapunov stable, but not asymptotically stable, which is compatible with the results of \cite{Volkov:2012wp}. This means that perturbations oscillate around it with constant amplitude, not being asymptotically flat.

\subsection{The causal structure of the improper bidiagonal solutions}

\label{subsec:nullcones}

In the following, we investigate the causal structure of the black
hole solutions by looking at the relationship between the null cones
of the two metrics. In particular, we will construct a set of null
frame vielbeins and plot the resulting null radial geodesics aligned
with the associated generating null cones. This analysis will
highlight the difference between the proper and the improper bidiagonal
solutions in more detail.

Consider the metric fields $g$ and $f$ at some point covered by
the ingoing Eddington-Finkelstein coordinates $x^{\mu}=(v,r,\theta,\phi)$ adapted for $g$.
We can construct a complex null frame field (vielbein) of the metric
$g$ which suits the ingoing null radial geodesics as,
\begin{equation}
\ell_{g}^{\mu}=\left(\mathrm{e}^{-q/2},\frac{1}{2}F,0,0\right),\quad n_{g}^{\mu}=\left(0,-1,0,0\right),\quad m_{g}^{\mu}=\frac{1}{\sqrt{2}r}\left(0,0,1,\frac{\mathrm{i}}{\sin\theta}\right).\label{eq:g-tetrad}
\end{equation}
It is easy to verify that $g^{\mu\nu}=-2\ell_{g}^{(\mu}n_{g}^{\nu)}+2m_{g}^{(\mu}\bar{m}_{g}^{\nu)}$
is the inverse of the metric $g$ from eq.~\eqref{eq:EFmetrics}.
In general, the geodesics of $g$ are obtained by the variational
method from the action,
\begin{equation}
	\frac{1}{2}\intop\mathrm{d}\lambda\left(e^{-1}(\lambda)\,g_{\mu\nu}\frac{\mathrm{d}x^{\mu}}{\mathrm{d}\lambda}\frac{\mathrm{d}x^{\nu}}{\mathrm{d}\lambda}-\mu^{2}e(\lambda)\right),\label{eq:g-geo-lagr}
\end{equation}
where the re-parametrisation invariance is ensured by the einbein
$e(\lambda)$. Moreover, $e(\lambda)$ acts as a Lagrangian multiplier
imposing the mass-shell constraint so that $\mu^{2}=1,0,-1$ corresponds
to time-, null- or space-like geodesics, respectively. For the null
radial geodesics, besides $\mu^{2}=0$, we have $\mathrm{d}\theta/\mathrm{d}\lambda=\mathrm{d}\phi/\mathrm{d}\lambda=0$.
Note also that the vielbein \eqref{eq:g-tetrad} solely depends on
the radial coordinate $x^{1}\equiv r$. Under these conditions, the
variation of \eqref{eq:g-geo-lagr} with respect to $x^\mu(\lambda)$ and $e(\lambda)$ gives the following two geodetic
equations, up to a suitably gauged $e(\lambda)$,
\begin{equation}
\left\{ \;\begin{aligned}\frac{\mathrm{d}v}{\mathrm{d}\lambda} & =\ell_{g}^{v}=\mathrm{e}^{-q/2},\\
\frac{\mathrm{d}r}{\mathrm{d}\lambda} & =\ell_{g}^{r}=\frac{F}{2},
\end{aligned}
\right.\qquad\text{and}\qquad\left\{ \;\begin{aligned}\frac{\mathrm{d}v}{\mathrm{d}\lambda} & =n_{g}^{v}=0,\\
\frac{\mathrm{d}r}{\mathrm{d}\lambda} & =n_{g}^{r}=-1.
\end{aligned}
\right.\label{eq:g-geodesics}
\end{equation}
Hence, the above null vector fields are aligned with the tangent vector
fields of null radial congruence matching the null radial rays. In
other words, the outgoing and the ingoing null radial geodesics of
$g$ are the integral curves of the null vectors $\ell_{g}$ and $n_{g}$,
correspondingly.

In similar fashion, a complex null frame vielbein of the metric $f$
that is adapted to the ingoing null radial geodesics can be constructed
as,
\begin{equation}
\ell_{f}^{\mu}=\left(\mathrm{e}^{-q/2}\frac{\Sigma+\tau}{2\Sigma\tau},\frac{F}{2\Sigma},0,0\right),\quad n_{f}^{\mu}=\left(\mathrm{e}^{-q/2}\frac{\Phi}{2\Sigma\tau},\frac{-1}{\Sigma},0,0\right),\quad m_{f}^{\mu}=\frac{1}{R}m_{g}^{\mu},\label{eq:f-tetrad}
\end{equation}
for which $f^{\mu\nu}=-2\ell_{f}^{(\mu}n_{f}^{\nu)}+2m_{f}^{(\mu}\bar{m}_{f}^{\nu)}$
yields the inverse of the metric $f$ from eq.~\eqref{eq:EFmetrics}.
Then, the resulting outgoing and the ingoing null radial geodesics
of $f$ are the integral curves of the null vectors $\ell_{f}$ and
$n_{f}$, respectively, with the corresponding geodetic equations,
\begin{equation}
\left\{ \;\begin{aligned}\frac{\mathrm{d}v}{\mathrm{d}\lambda} & =\ell_{f}^{v}=\mathrm{e}^{-q/2}\frac{\Sigma+\tau}{2\Sigma\tau},\\
\frac{\mathrm{d}r}{\mathrm{d}\lambda} & =\ell_{f}^{r}=\frac{F}{2\Sigma},
\end{aligned}
\right.\qquad\text{and}\qquad\left\{ \;\begin{aligned}\frac{\mathrm{d}v}{\mathrm{d}\lambda} & =n_{f}^{v}=\mathrm{e}^{-q/2}\frac{\Phi}{2\Sigma\tau},\\
\frac{\mathrm{d}r}{\mathrm{d}\lambda} & =n_{f}^{r}=-\frac{1}{\Sigma}.
\end{aligned}
\right.\label{eq:f-geodesics}
\end{equation}

Given the solution to the equations of motion \eqref{eq:EoM}, the
null radial geodesics of $g$ and $f$ can be determined by integrating
eqs.~\eqref{eq:g-geodesics} and \eqref{eq:f-geodesics}, displaying
the causal structure of the spacetime. Note that the $g$ and $f$
null cone fields are simply plotted as given by eqs.~\eqref{eq:g-tetrad}
and \eqref{eq:f-tetrad}. In order for the null cones to be Minkowski-like
in the weak field limit, however, it is customary to use the coordinates
$(t_{\star},r)=(v-r,r)$. In these coordinates, the null radial geodesics
and the null cone fields for the improper bidiagonal solutions described
in \autoref{subsec:phasespace}, are plotted in Figures \ref{fig:Snullcones}\textendash \ref{fig:SdSnullcones}.

Several observations can be made from these plots. First, the
static symmetry is obvious; as the fields are dependent only on the
radial coordinate $r$, the geodetic plots are invariant under 
translations along the $t_{\star}$ axis.

Second, far from the Killing horizon, but inside the Compton
wavelength of the massive mode, the null cones are almost
coinciding. While approaching the horizon, they differ more and more.
At the horizon, the null cones share at least one null direction (Figures\autoref{fig:Spert-hor},
\ref{fig:SAdSpert-hor}, \ref{fig:SdSpert-hor}, \ref{fig:SdSpert-cosmhor}),
since the translational Killing vector becomes null for both the metrics.
Inside the horizon, both null cones shrink while $r\rightarrow0$,
showing that the Killing horizon is actually an \emph{event} horizon.
The ingoing null radial geodesics entering the Killing horizon
do not come back; they approach the curvature singularity as in GR.
However, starting from the same point, the ingoing null radial geodesics
of $g$ and $f$ reach the curvature singularity at different
finite $t_{\star}$, that is, at different advanced times $v$.

Finally, the causal structure from the plots displays the difference
between the proper and the improper bidiagonal solutions. As pointed
out in \cite{Hassan:2017ugh}, the algebraic type of the improper bidiagonal
solutions correspond to the null cones having an odd number of common
null directions with a common tangent plane (Figures\autoref{fig:Spert-hor},
\ref{fig:SAdSpert-hor}, \ref{fig:SdSpert-hor}, \ref{fig:SdSpert-cosmhor}).
On the other hand, the algebraic type of the proper bidiagonal solutions
correspond to the null cones having an even number of common null
directions with a common tangent plane (the GR solutions have
coinciding null cones, therefore they share all their null directions). 

While crossing the Killing horizon, the null cones of the improper bidiagonal
solutions deform continuously (as can be seen directly from \eqsref{eq:g-tetrad}--\eqref{eq:f-geodesics}), therefore the fact that they cannot be simultaneously diagonalised only at the horizon is not a singular
behaviour. This clarifies the meaning of the expression  ``smoothly cross'' used in \autoref{prop:crossing} (\autoref{subsec:choice}).

%%%%%%%%%%%%%%%%%%%%%%%%%%%%%%%%%%%%%%%%%%%%%%%%%%%%%%%%%%

\begin{figure}[t]
\centering
	\subfloat[$g$-geodesics and null cones (blue) and $f$-geodesics and null cones (red) for the improper bidiagonal solution around the Schwarzschild one, with $R_\eH=1.03$ ($R_\eH^\mathrm{Schw}= 1$). The common event horizon is at $r_\eH=1$ (dashed black line) and the curvature singularity is at $r=0$ (solid black line). Null cones follow one $g$-geodesic and they are almost overlapping, because the perturbation in the initial conditions is small.]{\includegraphics[width=0.75\textwidth]{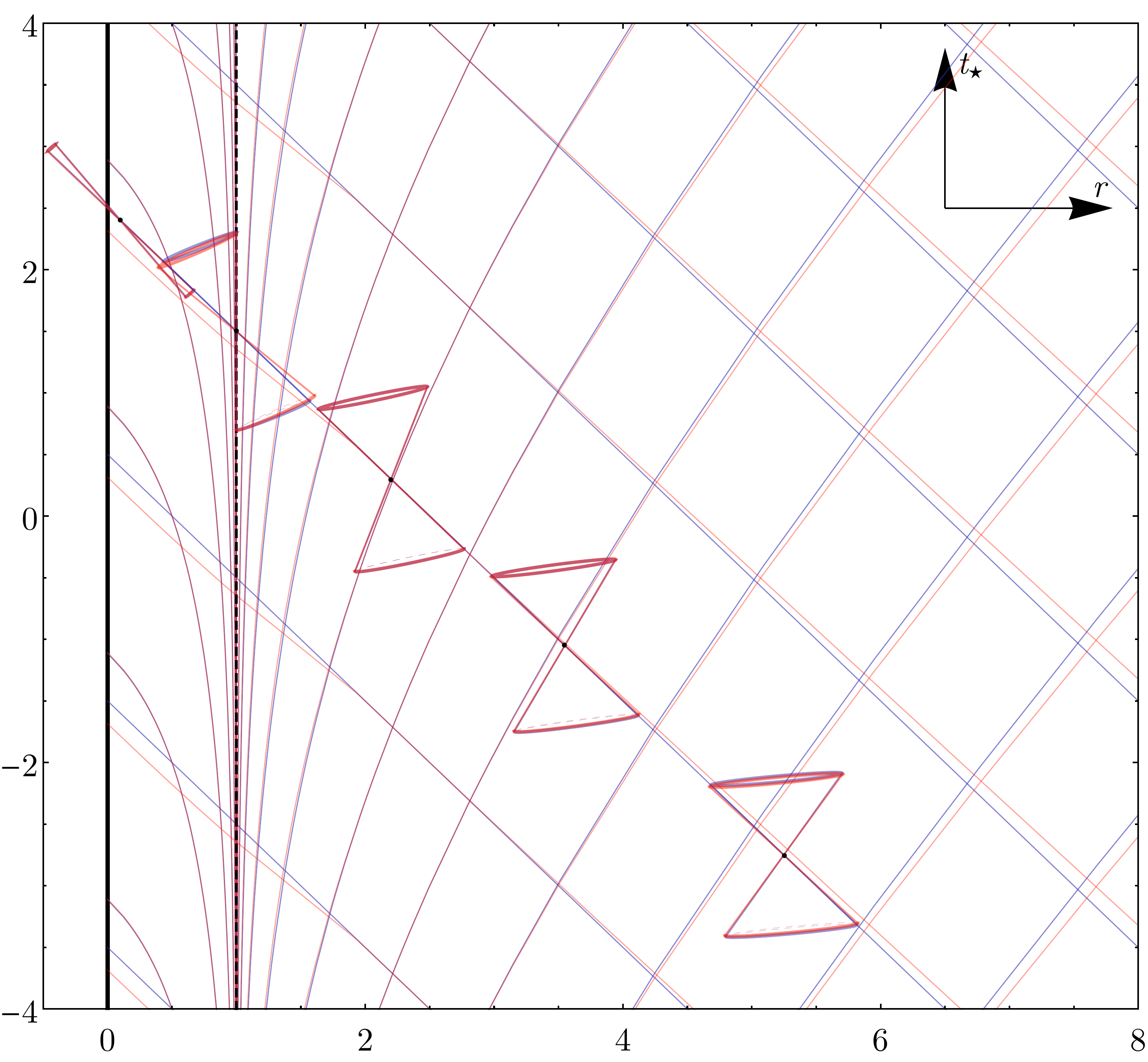}\label{fig:Spert-geo}} \\
	\centering
	\subfloat[Null cones and their horizontal (spacelike) sections in tangent space at the event horizon, $r_\eH=1$. The ellipses are the ellipsoids in 4D.]{\includegraphics[width=0.38\textwidth]{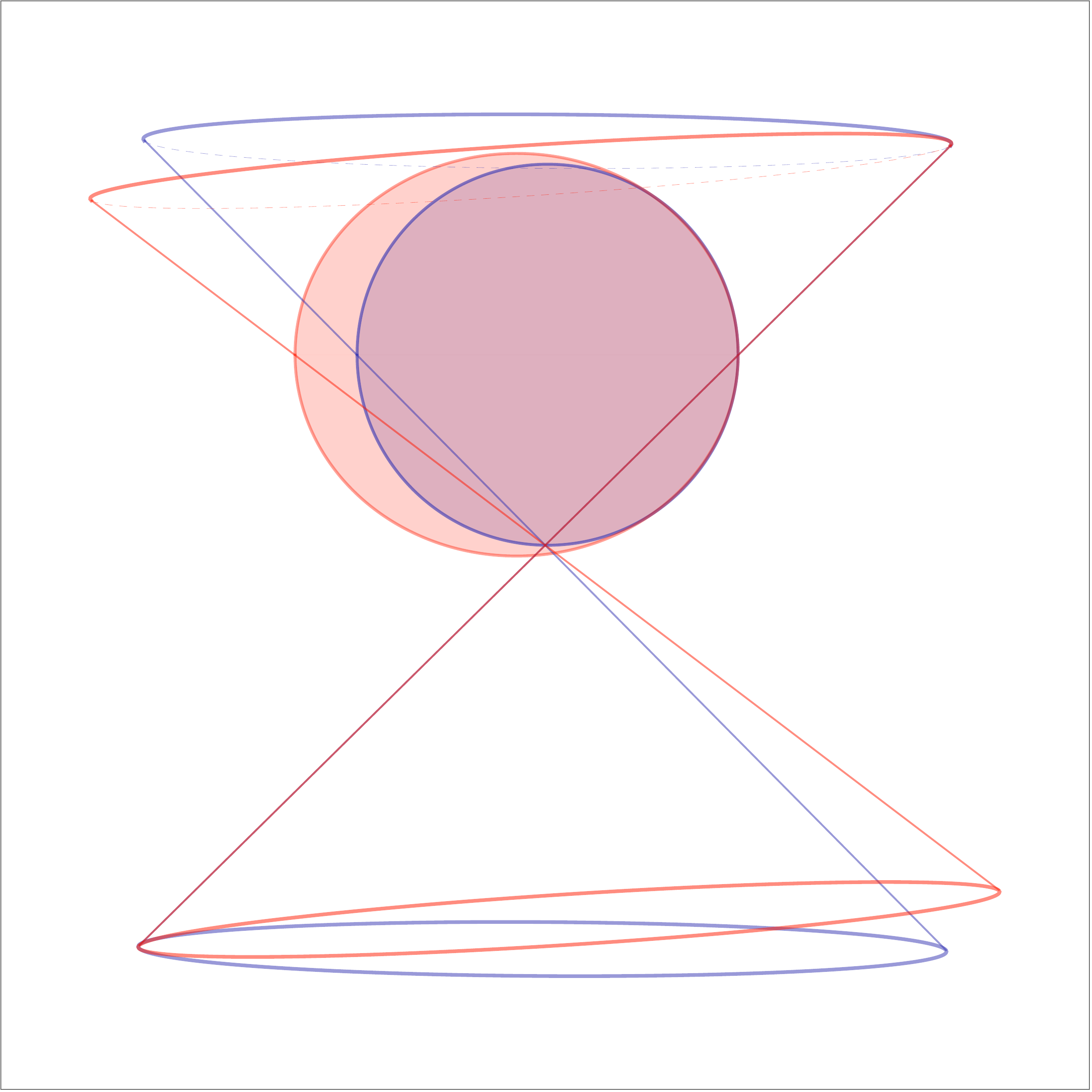}\label{fig:Spert-hor}}
	\hspace{0.5cm}
	\subfloat[Null cones and their horizontal (spacelike) sections in tangent space at $r=6$. The ellipses are the ellipsoids in 4D.]{\includegraphics[width=0.38\textwidth]{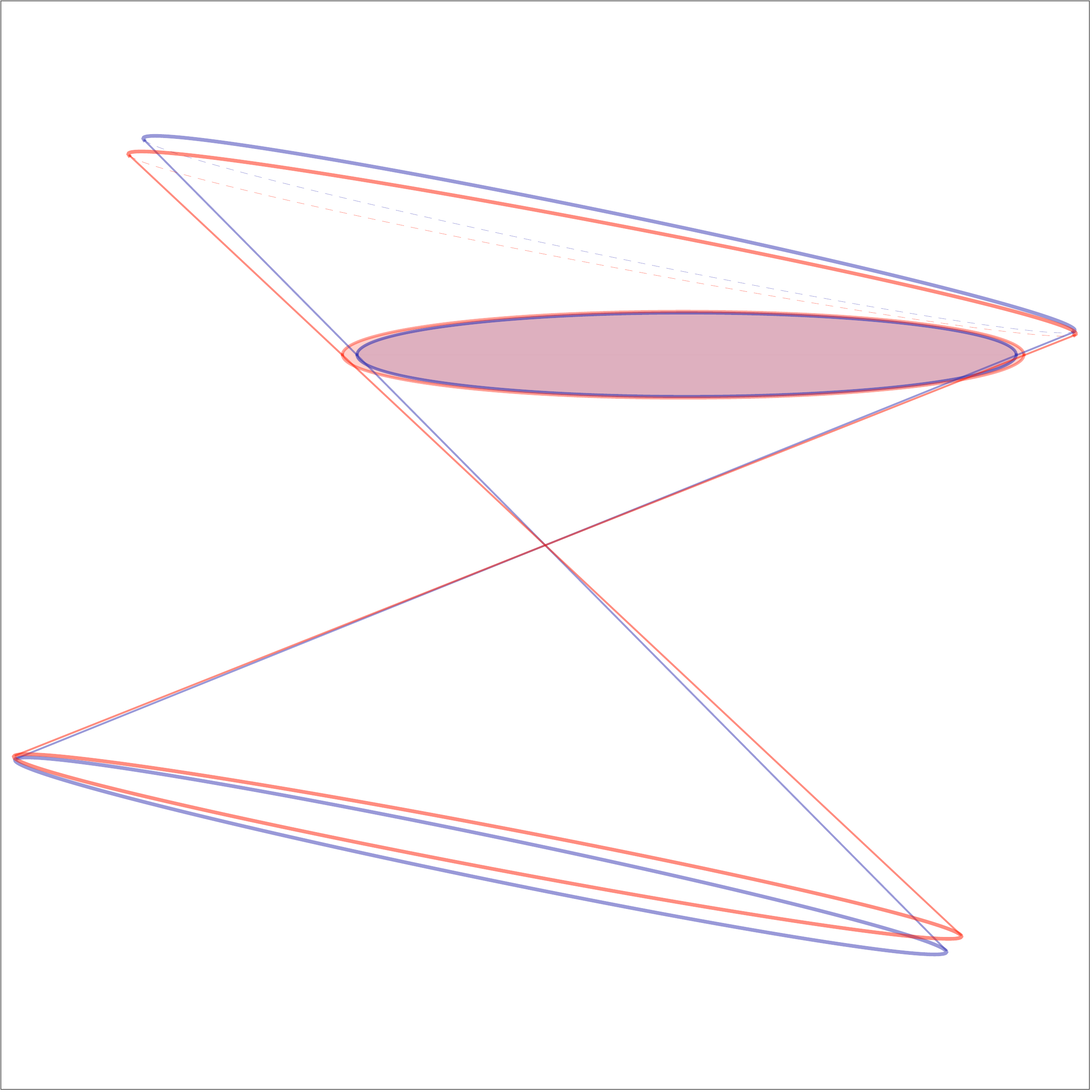}\label{fig:Spert-out}}
	\caption{Geodesics and null cones for the perturbation around the Schwarzshild solution with $R_\eH=1.03$.
	}
	\label{fig:Snullcones}
\end{figure}
\begin{figure}[t]
\centering
	\subfloat[$g$-geodesics and null cones (blue) and $f$-geodesics and null cones (red) for the improper bidiagonal solution around the SAdS one, with $R_\eH=3$ ($R_\eH^\mathrm{SAdS}\approx 2.633$). The common event horizon is at $r_\eH=1$ (dashed black line) and the curvature singularity is at $r=0$ (solid black line). Null cones follow one $g$-geodesic.]{\includegraphics[width=0.75\textwidth]{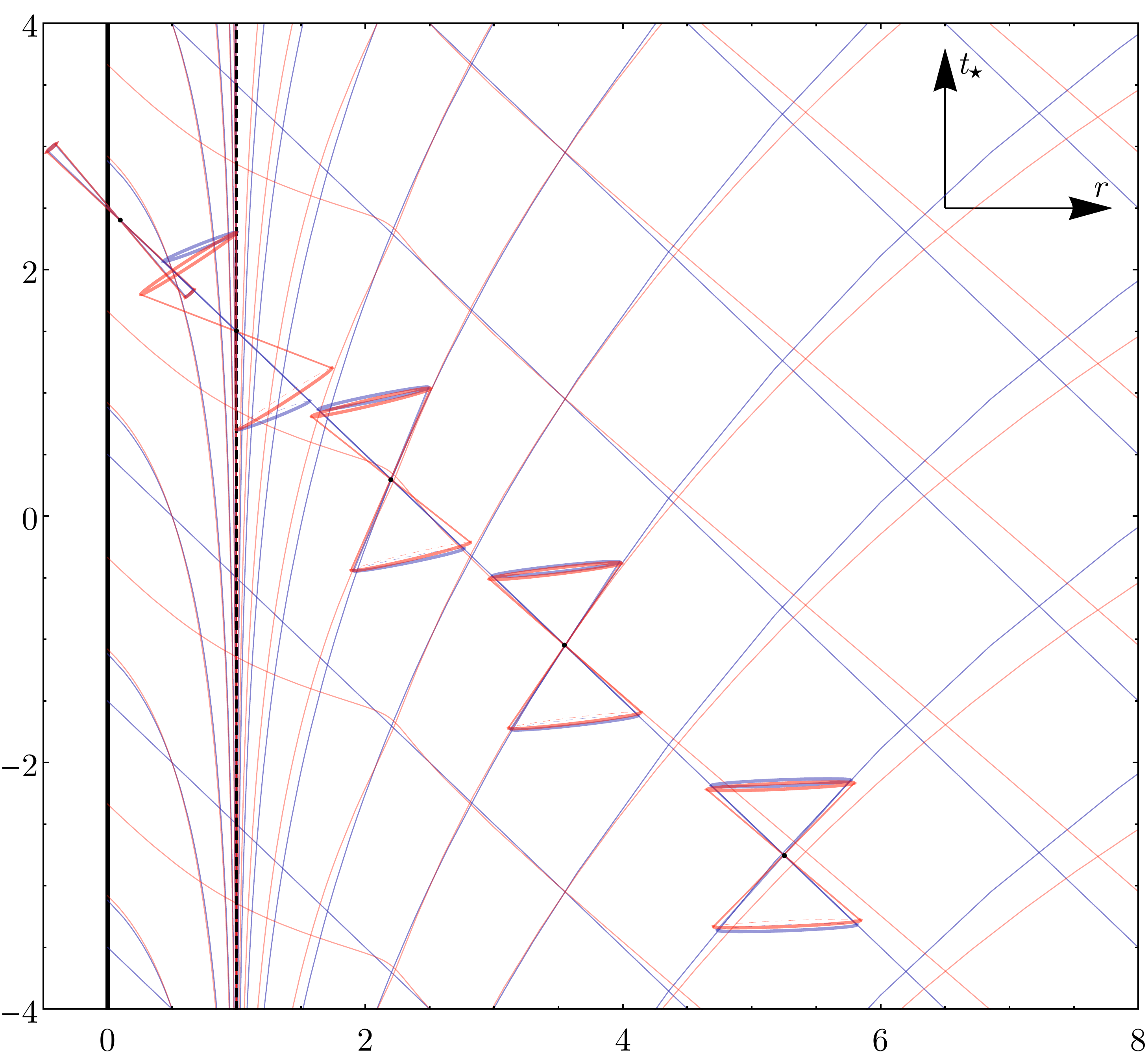}\label{fig:SAdSpert-geo}} \\
	\centering
	\subfloat[Null cones and their horizontal (spacelike) sections in tangent space at the event horizon, $r_\eH=1$. The ellipses are the ellipsoids in 4D.]{\includegraphics[width=0.38\textwidth]{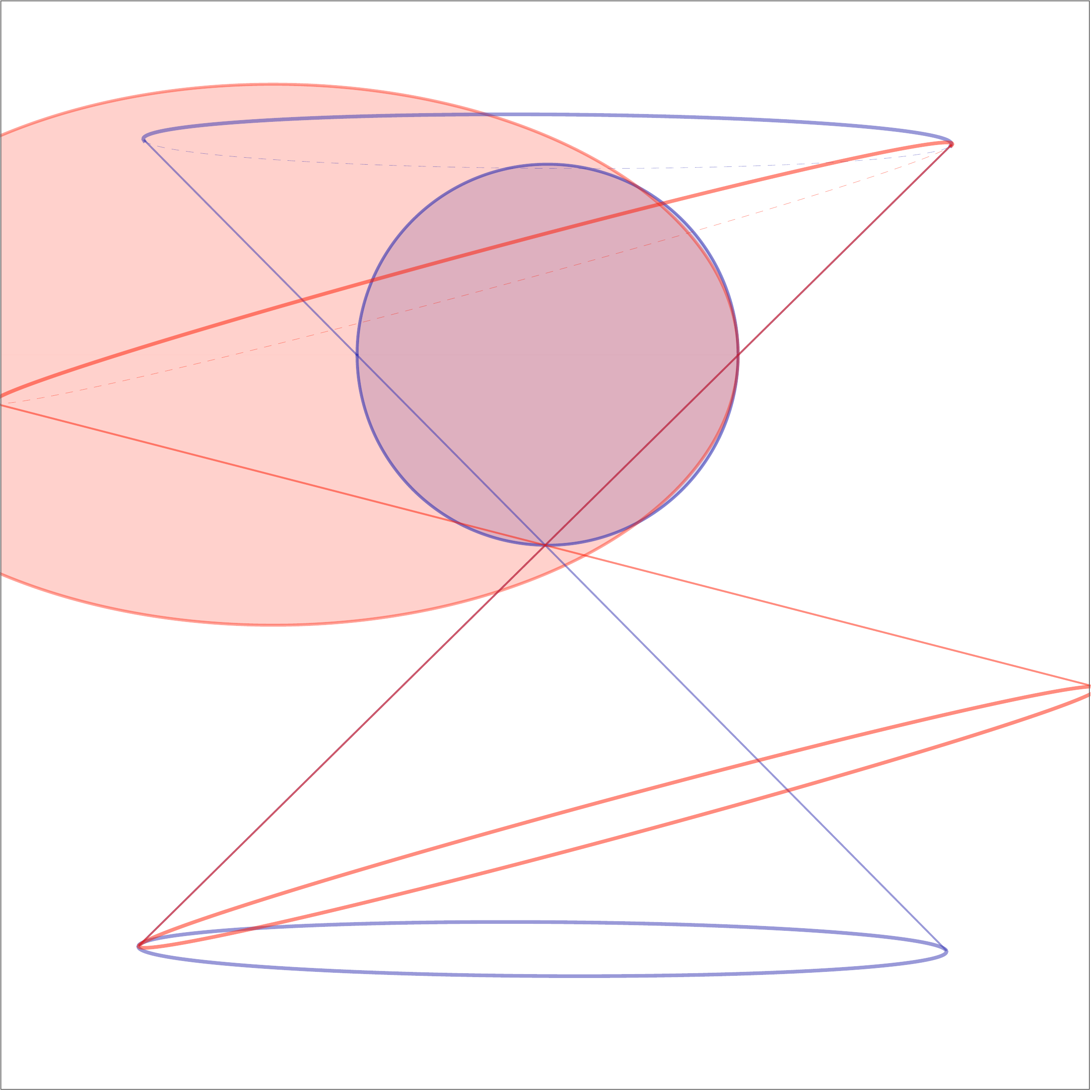}\label{fig:SAdSpert-hor}}
	\hspace{0.5cm}
	\subfloat[Null cones and their horizontal (spacelike) sections in tangent space at  $r=4$. The ellipses are the ellipsoids in 4D.]{\includegraphics[width=0.38\textwidth]{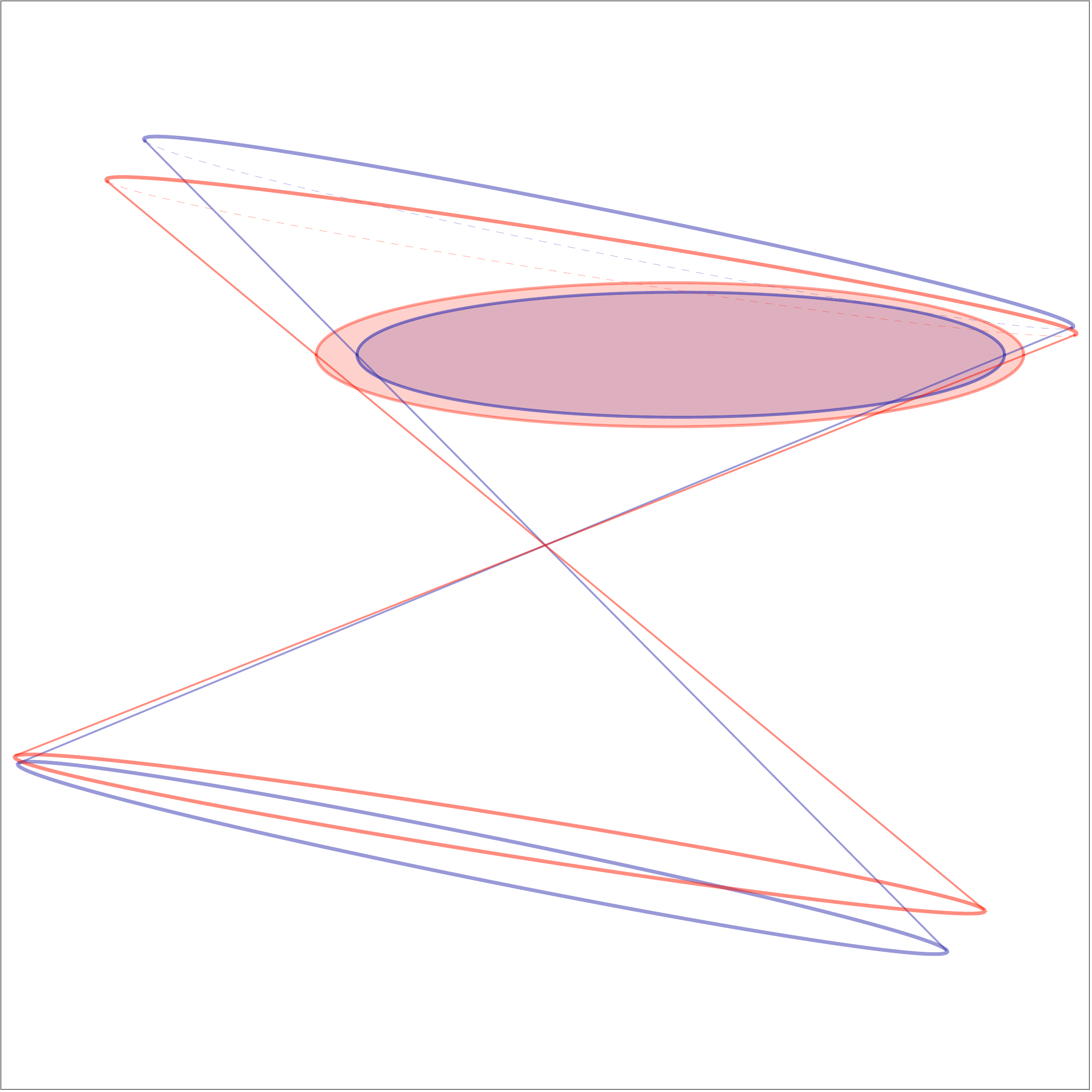}\label{fig:SAdSpert-out}}
	\caption{Geodesics and null cones for the perturbation around the SAdS solution with $R_\eH=3$.
	}
	\label{fig:SAdSnullcones}
\end{figure}
\begin{figure}[t]
\centering
	\subfloat[$g$-geodesics and null cones (blue) and $f$-geodesics and null cones (red) for the improper bidiagonal solution around the SdS one, with $R_\eH\approx -8.565$ ($R_\eH^\mathrm{SdS}\approx -8.557$). The common event horizon is at $r_\eH=1$ (dashed black line) and the cosmological horizon of SdS is at $r_\eC\approx 3.128$.]{\includegraphics[width=0.75\textwidth]{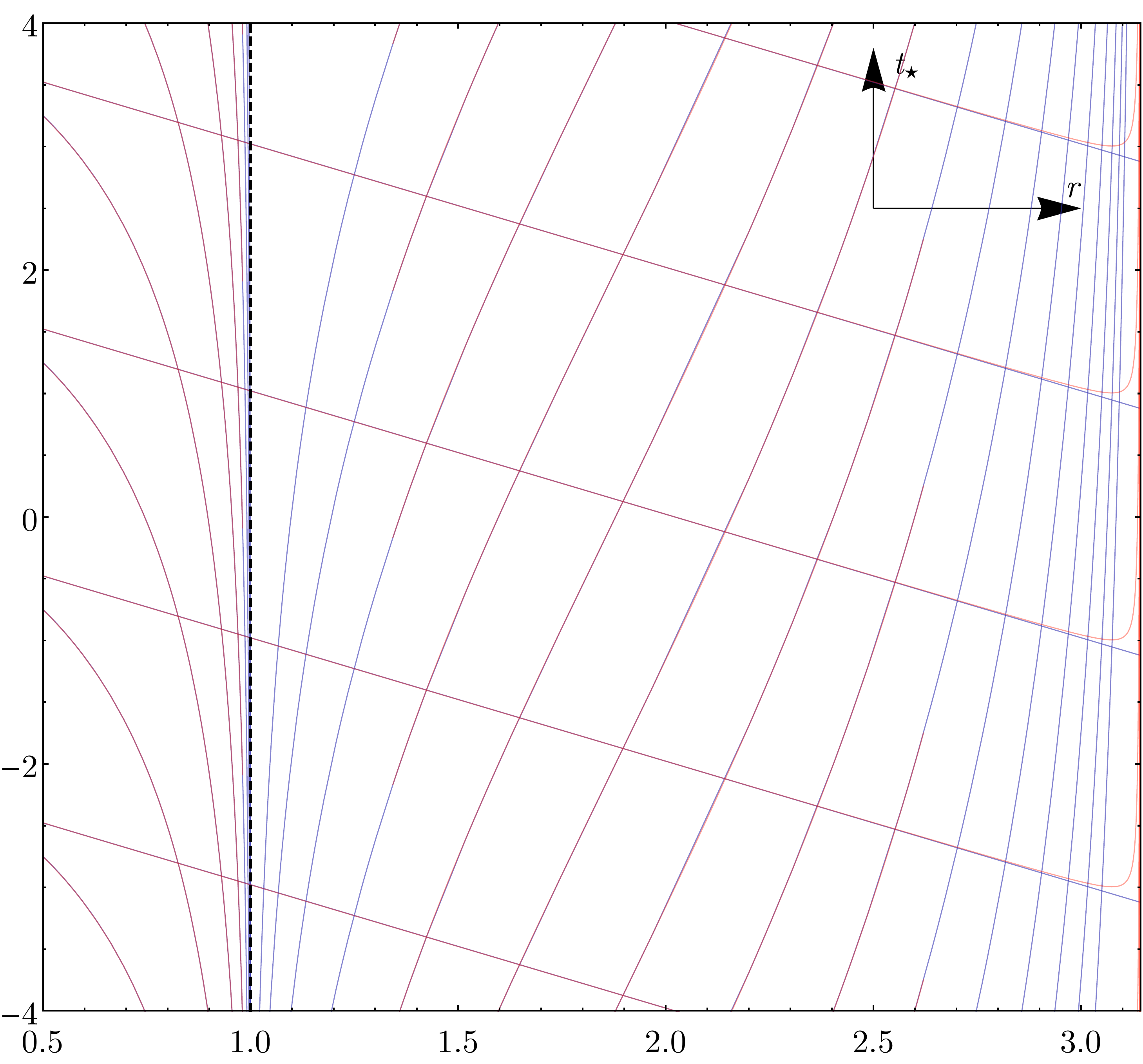}\label{fig:SdSpert-geo}} \\
	\centering
	\subfloat[Null cones and their horizontal (spacelike) sections in tangent space at the event horizon, $r_\eH=1$. The ellipses are the ellipsoids in 4D.]{\includegraphics[width=0.38\textwidth]{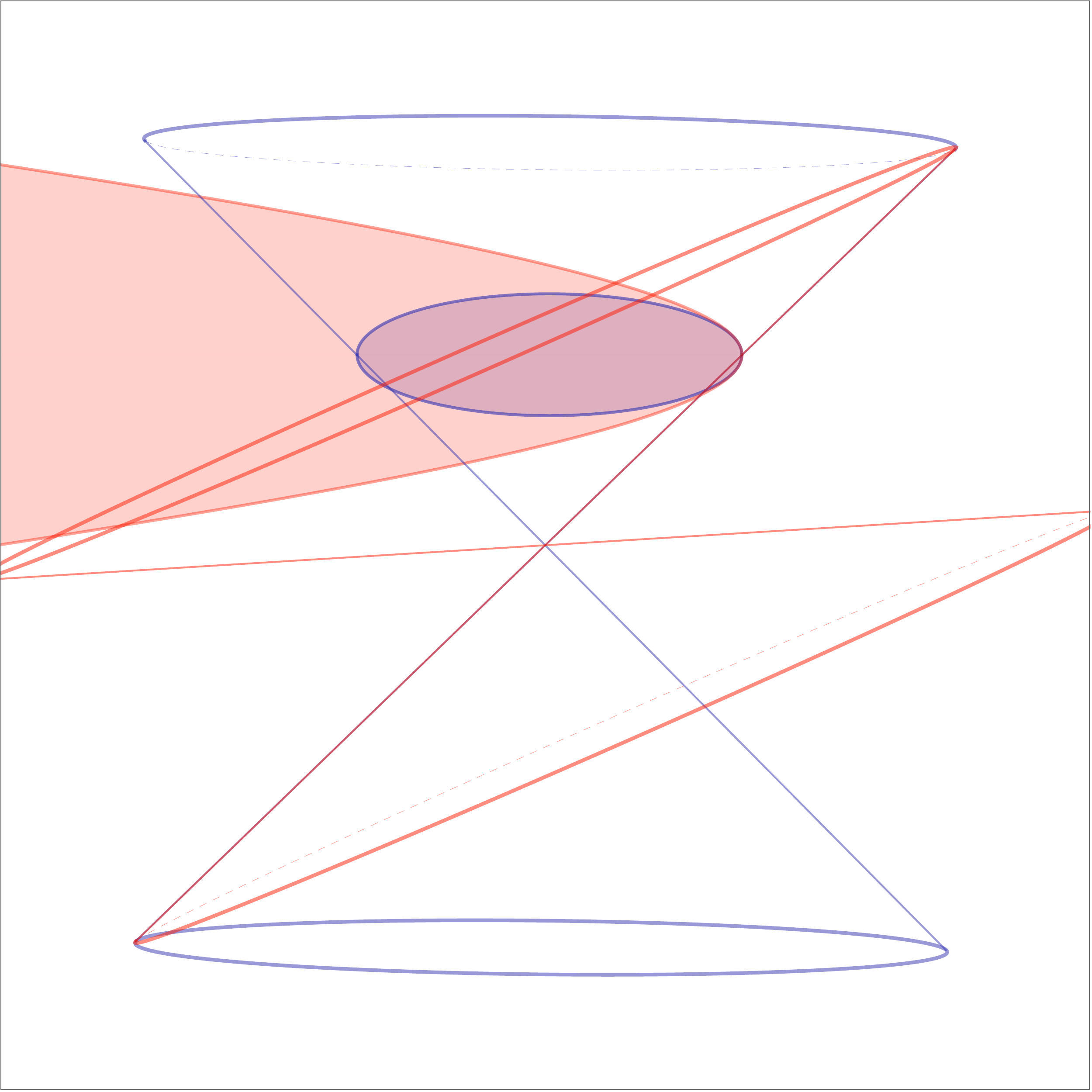}\label{fig:SdSpert-hor}}
	\hspace{0.5cm}
	\subfloat[Null cones and their horizontal (spacelike) sections in tangent space very close to the cosmological horizon, $r\lessapprox r_\eC\approx 3.128$. The ellipses are the ellipsoids in 4D.]{\includegraphics[width=0.38\textwidth]{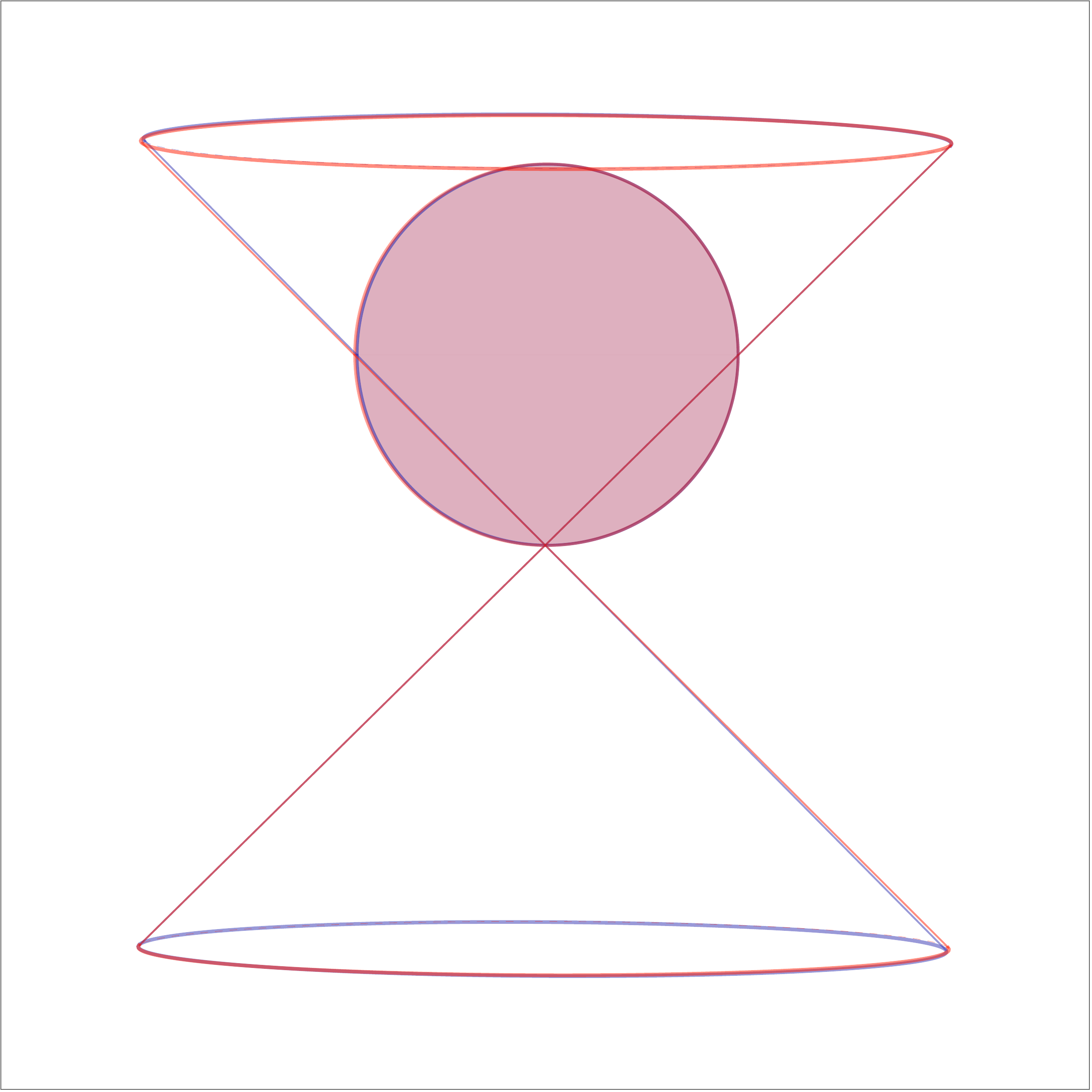}\label{fig:SdSpert-cosmhor}}
	\caption{Geodesics and null cones for the perturbation around the SdS solution with $R_\eH\approx -8.565$.
	}
	\label{fig:SdSnullcones}
\end{figure}

\section{Conclusions}
\label{sec:conclusions}

In this paper, we studied static and spherically symmetric bidiagonal black hole solutions in HR bimetric theory. 
Although the subject in part has been treated previously, this study is motivated by the fact that a nonlinear phase space stability analysis and a causal classification of these solutions was lacking. The former study shows that all the found non-GR solutions either diverge from the GR solutions at large radii, or have a singular square root $S=\sqrt{g^{-1}f}$ at finite $r$. These results are complementary to those in \cite{Volkov:2012wp} and extend them. There, the author considers the asymptotic properties of BH solutions in the leading order by studying the relative differences between GR solutions and their perturbations; here, we study the absolute differences between solutions and focus on their convergence properties in the phase space of our system of ODEs, i.e., we study the Lyapunov stability of the solutions, which was not analysed before.

We emphasize that the phase space analysis allowed us to look at \emph{all} solutions of the fields equations in a qualitative way, for a given set of the global parameters appearing in the action of the theory. Exploring the phase spaces for several different sets of global parameters gives equivalent results. Hence, we \emph{conjecture} that GR solution are Lyapunov unstable independently on the values of the global parameters of the theory. This conjecture is also motivated by analytical considerations about the behaviour of the fields equations at large radii. In particular, this conjecture implies that there cannot be asymptotically flat bidiagonal solutions other than Schwarzschild.

We pointed out that choosing a bidiagonal ansatz, is equivalent to imposing that the two metrics share their Killing horizons. We also presented an original proposition stating that, provided that one metric is regular, a necessary condition for the other metric not to have curvature singularities is that the matrix square root $S=\sqrt{g^{-1}f}$ is non-singular.

Thanks to an appropriate parametrisation of the metric functions, we have been able to detect pathological features of solutions.  
It also allowed us to find, for the first time, \emph{exact} initial conditions for the metric fields at the Killing horizon, using Eddington-Finkelstein coordinates. The initial conditions lie on a 2-dimensional surface in the phase space of our system of ODE. Every possible bidiagonal static and spherically symmetric black hole solution must have the metric fields intersecting this surface at the Killing horizon and, if present, at the cosmological horizon.

We found that the sole BH solutions asymptotically converging to flat, de Sitter or Anti-de Sitter spacetimes have proportional metrics, equivalent to the corresponding GR solutions. Other theoretically consistent solutions exist, but they all diverge from flat, de Sitter or Anti-de Sitter space at large radii. We show that the former correspond to {\em proper} bidiagonal solutions, and the latter to {\em improper} bidiagonal solutions. The proper bidiagonal solutions are bidiagonal everywhere, whereas the improper bidiagonal solutions are bidiagonal inside and outside the Killing horizon, but not \emph{at} the Killing horizon. Note, however, that all proper bidiagonal solutions do not necessarily correspond to GR solutions. 

Our results show that black holes having very small differences in the metric fields at their Killing horizons, will have completely different asymptotic structures.
In the light of the asymptotic exotic behaviour of the non-GR bidiagonal solutions, physical solutions would na\"ively correspond to the GR solutions. However, assuming that $S$ is non-singular, i.e., the solutions are theoretically consistent, we note that this behaviour start to become important only at radii similar to the Compton wavelength of the massive mode. This scale is many orders of magnitude outside the Killing horizon, typically even at cosmological length scales, and anyway far outside the region where the assumption of an isolated black hole breaks down.
Therefore, in a real physical setting, the divergences may be either non-existing or non-observable. 

Only studying static solutions, we can not answer the question of which solutions are realised in the process of gravitational collapse of matter. However, in \cite{Brito:2013wya,Babichev:2013una,Babichev:2014oua,Babichev:2015zub} it was shown that the bidiagonal solutions are dynamically unstable, and as such cannot represent the end point of gravitational collapse. In addition, \cite{Babichev:2014oua} showed that to have dynamically stable solutions, non-diagonal metric elements are needed, also at radii other then the horizon radius, $r_\eH$. Therefore, in order to investigate what is the end point of gravitational collapse in HR bimetric theory, a full dynamical treatment of the process probably need to be performed. Also, the Lyapunov stability of non-bidiagonal BHs and electrically charged BHs should be explored.

\begin{acknowledgments}

We want to thank Jonas Enander, Fawad Hassan, Bo Sundborg, Ingemar Bengtsson, Stefan Sj\"ors and Lars Bergstr\"om for helpful discussions. 
Our gratitude goes to Mikhail S.~Volkov and Richard Brito for reading the paper and providing valuable remarks.
We also thank Anders Lundkvist, Giacomo Monari, Julius Engels\"oy, Sebastian Baum and Luca Visinelli for useful suggestions.
Support for this study from the Swedish Research Council for EM is acknowledged.

\end{acknowledgments}

\appendix

\section{Old parametrisation for the metric fields}
\label{appendix-A}

Given the coordinate system $x^\mu =(t,r,\theta,\phi)$,
the most commonly used parametrisation of the metric fields for the bidiagonal ansatz is (see, for example, \cite{Volkov:2012wp,Brito:2013xaa,Babichev:2013pfa,Enander:2015kda}),
\begin{subequations}
\label{eq:oldmetrics}
\begin{align}
	\dd s^2_g &= -Q(r)^2  \,\dd t^2 + \dfrac{1}{N(r)^2}  \,\dd r^2 
	+ r^2\left(\dd \theta^2 + \sin(\theta)^2 \,\dd \phi^2\right), \\ 
	\dd s^2_f &= -a(r)^2  \,\dd t^2 + \dfrac{U'(r)^2}{Y(r)^2}  \,\dd r^2 
	+ U(r)^2\left(\dd \theta^2 + \sin(\theta)^2 \,\dd \phi^2\right),
\end{align}
\end{subequations}
where $Q(r),N(r),a(r),Y(r),U(r)\in \mathbb{R}$. The determinants of these metrics are,
\begin{equation}
	\det(g)=-\dfrac{Q(r)^2}{N(r)^2}r^4\sin(\theta)^2, \qquad \det(f)=-\dfrac{a(r)^2U'(r)^2}{Y(r)^2}U(r)^4\sin(\theta)^2.
\end{equation}
The square root matrix $S=\sqrt{g^{-1} f}$, after having chosen the principal branch, is,
\begin{equation}
	\tud{S}{\mu}{\nu} = \op{diag} \left( \, 
		\left| \dfrac{a(r)}{Q(r)} \right|,
		\left| \dfrac{N(r)U'(r)}{Y(r)} \right|,
		\left| \dfrac{U(r)}{r} \right|,
		\left| \dfrac{U(r)}{r} \right|
	\, \right).
\end{equation}
The absolute values can be removed assuming that the fields are always positive.
The determinant and the trace of the square root are,
\begin{equation}
	\det(S)=\dfrac{a(r)N(r)U(r)^2U'(r)}{Q(r)Y(r)r^2}, \qquad \mbox{Tr}(S)=\dfrac{a(r)}{Q(r)}+\dfrac{N(r)U'(r)}{Y(r)}+2\dfrac{U(r)}{r}.
\end{equation}
The Ricci scalars for the two metrics are,
\begin{align}
	\label{eq:oldscalars}
	R^g&=-\dfrac{2}{r^2Q}\left[ Q\left( N^2+2rN N'-1 \right)+Q'r\left( 2N^2+rNN' \right)+Q''N^2r^2 \right], \\
	R^f&=-\dfrac{2}{aU^2U'^3}\left[ aU'^2\left(Y^2U'+2UYY'-U' \right) \, + \right. \nonumber \\
		&\qquad \left. a'U\left( 2Y^2U'^2+YY'UU'- Y^2UU'' \right)+a''U^2Y^2U' \right],
\end{align}
where all fields are functions of the radial coordinate $r$.
We are now going to discuss why this parametrisation is not optimal when studying BH solutions.

First, this parametrisation is not defined inside the Killing horizon of the BH. Indeed, the norm of the translational Killing vector $\mathcal{K}^\mu = \delta^\mu_0$ is,
\begin{equation}\label{eq:kvf-old}
\mathcal{K}^2 ={\mathcal{K}}_{\mu}{\mathcal{K}} ^{\mu}=g_{\mu \nu}{\mathcal{K}}^{\mu}{\mathcal{K}} ^{\nu}=g_{00}=-Q(r)^2 ~~\Longrightarrow~~  \mathcal{K}^2 \le 0.
\end{equation}
In the coordinates adapted to the translational Killing vector, we cannot cross the Killing horizon since $\mathcal{K}$ becomes null at $r=r_\eH$. 
On the other hand, using the Eddington-Finkelstein coordinates is problematic since $\mathcal{K}^2 = -Q^2$ cannot be positive, so we cannot cross the Killing horizon. Yet, if we use the Eddington-Finkelstein coordinates, this parametrisation is defined in the interval $r\in [r_\eH,+\infty)$, which, in principle, allows specification of the initial values on the Killing horizon.
Nonetheless, the numerical solutions found in the literature (see, for example, \cite{Volkov:2012wp,Brito:2013xaa,Babichev:2013pfa,Enander:2015kda}) are obtained by imposing initial conditions close to the Killing horizon, after having expanded the equations up to some (usually the first) order around $r=r_\eH$. 
Namely, in order to define a Killing horizon at $r_\eH$, one must impose $Q(r_\eH)=a(r_\eH)=0$. In order for the determinants of the two metrics not to be zero at $r_\eH$, we must also have $N(r_\eH)=Y(r_\eH)=0$, with $\lim\limits_{r \rightarrow r_\eH}Q(r)/N(r)=\mathrm{const}$ and $\lim\limits_{r \rightarrow r_\eH}a(r)/Y(r)=\mathrm{const}$. These conditions can be difficult to control when dealing with numerical solutions. Also, pathologies in the solutions are not straightforward to analyse. If only one of the metric fields becomes zero at finite $r$, the singularity introduced can be of different kinds, depending on the behaviour of more fundamental quantities.

For example, suppose that, at some finite value $\bar{r}$, $Y(\bar{r})=0$, but the other metric fields and their derivatives are non-zero and finite. In this case $\det S,\det f$ and $f_{rr}$ diverge, whereas the $g$-sector is regular. Therefore we have a determinant singularity at $r=\bar{r}$, whose effects are studied in \cite{Gratia:2013gka,Gratia:2013uza}. Suppose, instead, that for some finite $r^*$, $a(r^*)= 0$, but the other metric functions and their derivatives are non-zero and finite. Then, $\det f=0$ and the Ricci scalar $R^f$ of $f_{\mu \nu}$ diverges, indicating a curvature singularity which is not hidden by a horizon, i.e. a naked singularity.

In principle, one could think that pathologies are automatically excluded by the equations of motion, i.e. an actual solution would not have these problems. This is not the case, for example, in numerical solutions. We therefore set out to find a more convenient parametrisation for the metric functions in \autoref{subsec:choice}.

We emphasise that detecting pathologies is equivalent to find the points where $\det(S)=0$, $\det(S^{-1})=0$, or where other elementary symmetric polynomials of $S$ are not finite. With our new parametrisation of the metric fields, this corresponds to points with a singular behaviour of $\Sigma(r)$, $\tau (r)$ and $R(r)$, separately. 

\begin{figure}[t]
\centering
	\includegraphics[width=125mm]{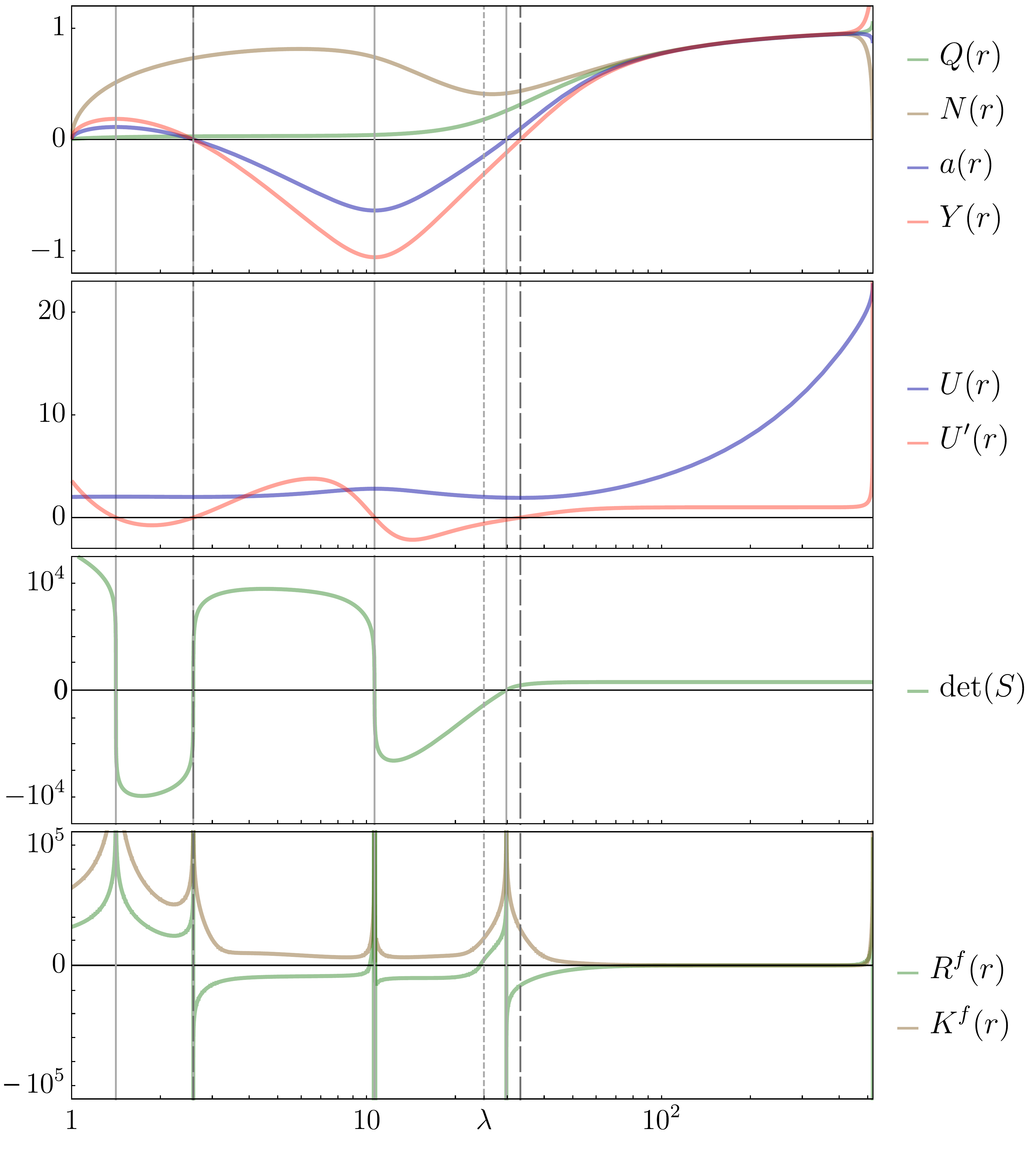}
	\caption{The relevant quantities for the solution obtained in \cite{Brito:2013xaa}, plotted in different panels. The first and second panels, showing the numerical solutions for the metric functions, together with the first derivative $U'(r)$, can be directly compared respectively with Figure 4 (second panel) and Figure 3 (second panel) of \cite[p.\,5]{Brito:2013xaa}. The third and fourth panels show, respectively, $\det S$ and the Ricci and Kretschmann scalars, $R^f$ and $K^f$, for the $f$ metric. $\lambda $ is the Compton wavelength of the massive mode. See text for explanations.}
	\label{fig:illSolution}
\end{figure}

We now consider an explicit example of a pathological behaviour, already understood as such in \cite{Brito:2013xaa}. The solution in question is plotted in \autoref{fig:illSolution} and is a good example of the importance of \autoref{th:finiteness}. Note that this solution is defined in the interval $(r_\eH=1,+\infty)$ and has $\beta_1=1,\beta _2=\beta _3=0$, $\beta _0 $ and $\beta _4$ determined by the asymptotically flatness conditions in \eqnref{eq:AFconditions} (so $c=1$) and $m_\mathrm{g}=0.04$ ($m_\mathrm{g}$ is denoted with $\mu $ in \cite{Brito:2013xaa}; the Compton wavelength of the massive graviton is then $\lambda =25$ in geometrical units with the Planck constant $h=1$). Hence $m_{\mathrm{g}}m_{\mathrm{S}}=m_{\mathrm{g}}r_\eH/2=0.02$, with $m_{\mathrm{S}}$ Schwarzschild mass of the BH. The second panel of Figure 1 in \cite[p.\,4]{Brito:2013xaa} tells us what is the value of $U(r_\eH)=u_\eH$ to choose as initial condition, given the value of $m_{\mathrm{g}}m_{\mathrm{S}}$. Note that, since $r_\eH=1$ both in \cite{Brito:2013xaa} and in our approach, $u_\eH$ corresponds to our $R_\eH$. For $m_{\mathrm{g}}m_{\mathrm{S}}=0.02$, $u_\eH\approx 50.2$.

In \autoref{fig:illSolution}, all relevant quantities for this solution are plotted: the metric functions and the first derivative $U'(r)$, which appears explicitly in $f_{\mu \nu}$, $\det S$ and the Ricci and Kretschmann scalars for $f_{\mu \nu}$. The curvature scalars for $g_{\mu \nu}$ diverge at $r=0$ only, where the curvature singularity arises, and are not plotted. The third panel of \autoref{fig:illSolution} shows $\det S$; the grey solid lines highlight the points at which $\det S=0$. These lines are plotted also in the other panels, in order to understand what gives rise to the singular behaviour. In the first panel, we see all metric functions, except $U(r)$ and $U'(r)$ which are plotted in the second panel. The solid grey lines correspond to the following 4 points: from left to right, (i) $U'(r)=0$, (ii) $a(r)=0$, (iii) $U'(r)=0$, (iv) $a(r)=0$. There are also large-dashed grey lines, highlighting the points at which $U'(r) = Y(r) = 0$, keeping $\det S$ finite. Note that the second solid grey line and the first dashed grey line are almost overlapping, yet the points they represent are different. Now, we can look at the fourth panel; the curvature scalars for $f$ are diverging when $\det S=0$ (whereas the $g$-sector is regular). This confirms the result of \autoref{th:finiteness}. Note that the change of sign of the fields is anyway not compatible with the selected principal branch of the square root matrix used in the equations of motion.

Having a singular square root $S$, this solution is not physically acceptable. In accordance with \autoref{th:finiteness}, it has four naked singularities at different finite $r$. In addition, it is not an asymptotically flat solution. Finally, $U(r)$ is not a monotonic function; which is required to have a spherically symmetric spacetime \cite{choquet2008general}. 
We also studied this solution in our new parametrisation, finding the equivalent results with respect to the curvature scalars and $\det S$. 

\section{Eddington-Finkelstein coordinates}
\label{appendix-B}

Here we introduce Eddington-Finkelstein coordinates to avoid the coordinate singularity at the Killing horizon.
In the coordinate chart $x^\mu=(t,\xi,\theta,\phi)$, the metrics $g$, $f$ and the square root matrix $S$ are given by,
\begin{align} \label{eq:Smetrics}
	g_{\mu \nu}&=
		\begin{pmatrix}
			-\ee^{q}F & 0 & 0 & 0 \\
			0 & F^{-1} & 0 & 0 \\
			0 & 0 & r^2 & 0 \\
			0 & 0 & 0 & r^2 \sin (\theta)^2 \\
		\end{pmatrix} \!, \qquad
	\begingroup\colSep{5pt}
	{S^{\mu}}_{ \nu}=
		\begin{pmatrix}
			\tau & 0 & 0 & 0 \\
			0 & \Sigma & 0 & 0 \\
			0 & 0 & R & 0 \\
			0 & 0 & 0 & R \\
		\end{pmatrix}
	\endgroup \!, \\
	f_{\mu \nu} &=
		\begin{pmatrix}
			-\ee^{q}\tau^2F & 0 & 0 & 0 \\
			0 & \Sigma^2F^{-1} & 0 & 0 \\
			0 & 0 & R^2r^2 & 0 \\
			0 & 0 & 0 &  R^2r^2\sin (\theta)^2 \\
		\end{pmatrix} \!.
\end{align}
The determinants of the metrics are,
\begin{equation}
	\det(g)=-\ee^{q}r^4\sin(\theta)^2 \qquad \det(f)=\tau^2\Sigma^2R^4 \; \det(g),
\end{equation}
which are always non-zero except at $r=0$, where the curvature singularity arises. This can be seen by looking at the Ricci and Kretschmann scalars for the $g$ metric,
\begin{subequations}
\label{eq:scalars}
\begin{align}
	R^g&=-\dfrac{2 \left(2 F'+F q'\right)}{r}+\dfrac{1}{2} \left[-2 F''-3 F' q'-F \left(2 q''+q'^2\right)\right]+\dfrac{2(1- F)}{r^2}, \\
	K^g&=\dfrac{4 F F' q'+4 F'^2+2 F^2 q'^2}{r^2}+\dfrac{1}{4} \left[2 F''+3 F' q'+F \left(2 q''+q'^2\right)\right]^2+\dfrac{4 (F-1)^2}{r^4}. 
\end{align}
\end{subequations}
The Ricci and Kretschmann scalars for the $f$ metric are more complicated and there is no need to write them, since \autoref{th:finiteness} implies that we can have induced curvature singularities in $f$ sectors whenever $\det (S)=0$ or $\det (S^{-1})=0$.

To avoid the coordinate singularity at the Killing horizon, we introduce the ingoing Eddington-Finkelstein coordinates $\bar{x}^\beta=\left( v,\xi,\theta,\phi \right)$ for $g$ defined by the following Jacobian,
\begingroup\colSep{5pt}
\begin{equation}
\label{eq:EFjacobian}
	{J^{\mu}}_{\nu}:= \dfrac{\partial x^{\mu}}{\partial \bar{x}^\nu}=
		\begin{pmatrix}
			1 &-\ee^{-q/2}F^{-1} & 0 & 0 \\
			0 & 1 & 0 & 0 \\
			0 & 0 & 1 & 0 \\
			0 & 0 & 0 & 1 \\
		\end{pmatrix} \!,
\end{equation}
\endgroup
where $x^\alpha=\left( t,\xi,\theta,\phi \right)$ are the old Schwarzschild coordinates. This defines the ingoing null coordinate $v$  for $g$ (also known as `advanced time'), 
\begin{equation}
\dd v=\dd t+\dfrac{\dd \xi}{\ee^{q(\xi)/2}F(\xi)}.
\end{equation} 
Applying the coordinate transformation \eqref{eq:EFjacobian} to the metrics \eqref{eq:Smetrics} yields \eqnref{eq:EFmetrics}. 

We could equally consider the outgoing null coordinate $u$; this would not affect any conclusions, however, since the equations are invariant under the transformation $u=-v$. Also, in principle, we could introduce the null coordinate adapted for $f$ rather than for $g$.

Note that, contrary to the GR case, we cannot explicitly integrate the differential $\dd v$ because we do not know analytically the functions $q(\xi)$ and $F(\xi)$. Nonetheless, we can still define the \emph{tortoise coordinate} in an implicit form,
\begin{equation}
	\dd \xi^*=\dfrac{\dd \xi}{\ee^{q(\xi)/2}F(\xi)} ~~ \Longrightarrow ~~ \xi^*=\int \dfrac{\dd \xi}{\ee^{q(\xi)/2}F(\xi)} + \mbox{constant},
\end{equation}
where $\dd \xi =\dd r$. Note that $\xi^*$ is a monotonic function of $\xi$ separately inside and outside the event horizon (or, for a SdS solution, between the horizons, inside the event horizon and outside the cosmological horizon),
\begin{equation}
	\dfrac{\dd \xi^*}{\dd \xi}=\dfrac{1}{e^{q(\xi)/2}F(\xi)}=
	\left\lbrace\begin{array}{l}
		>0, \quad\mbox{ if }  F(\xi)>0, \\
		<0, \quad\mbox{ if }  F(\xi)<0.
	\end{array}\right.
\end{equation}

\section{On bidiagonality condition in [44]}
\label{appendix-C}

In the following we make Proposition 1 in \cite{Deffayet:2011rh} more precise. Consider a coordinate patch where two static and spherically symmetric metrics take the form,
\begin{align}
	f &= - J(r) \, \dd t^2 + K(r) \, \dd r^2 + r^2 \left( \dd \theta^2 + \sin^2 \theta \, \dd \phi^2 \right), \\
	g &= - A(r) \, \dd t^2 + 2 B(r) \, \dd t \dd r + C(r) \, \dd r^2 + D(r) \left( \dd \theta^2 + \sin^2 \theta \, \dd \phi^2 \right).
\end{align} 
Suppose further that the Killing vector $\mathcal{K}=\partial_t$ is null with respect to $g$ at $r = r_\eH$.
Proposition 1 in \cite{Deffayet:2011rh} states that, if both metrics describe smooth geometries and are diagonal at $r=r_\eH$, the Killing vector $\mathcal{K}$ must also be null with respect to $f$ at $r = r_\eH$. 
The proof presented in \cite{Deffayet:2011rh} is elegant in its simplicity. It starts from the presuppositions $A(r_\eH) = 0$, $B(r)=0$, and that the trace $\Tr(g^{-1}f) = J/A + K/C + 2 r^2/D$ is finite by assumption. Then, the individual terms in the trace cannot cancel since they necessarily have the same sign as both diagonal metrics have Lorentzian signature when the translation Killing vector $\mathcal{K}$ is timelike for $r > r_\eH$.
Therefore $J(r_\eH) = 0$, otherwise $J/A$ would diverge at $r=r_\eH$.

Now, we further evolve the proof. Let us only consider the $(t,r)$-block of the coordinates, in which case $\det g = - AC$ and $\det f = - JK$, which are non-zero by the premise of having smooth geometries. From the assumptions of Proposition 1, besides the trace, the determinant $\det( g^{-1} f ) =(J/A)(K/C)$ is also regular.
Hence, because $J/A$ is finite, $K/C$ must also be finite, and they must be both non-zero.
To see what happens at $r=r_\eH$, we introduce the ingoing Eddington-Finkelstein coordinates $(v,r)$ adapted for $g$, by $\dd t^2 = \dd v^2 - 2\sqrt{C/A} \,\dd v \dd r + (C/A)\dd r^2$ at $r> r_\eH$.
Assuming Lorentzian signature of the metrics, we have (up to square root signs, not shown for readability),
\begin{equation}
	g = \begin{pmatrix} 
			-A & \sqrt{AC} \\ 
			\sqrt{AC} & 0 
		\end{pmatrix} \!, \quad
	f = \begin{pmatrix} 
			-J &~ \sqrt{\frac{J}{A}\frac{C}{K}JK} \\
			\sqrt{\frac{J}{A}\frac{C}{K}JK} &~ K - \frac{J}{A}C 
		\end{pmatrix} \!.
\end{equation}
Clearly, the metrics are regular at $r=r_\eH$ in this coordinate system.
Since both $J/A$ and $K/C$ are finite and non-zero, it is also $c \coloneqq (J/A)(C/K)$ and we can write,
\begin{equation}
	f = \begin{pmatrix} 
			-J &~ \sqrt{c JK} \\ 
			\sqrt{c JK} &~ K - \frac{J}{A}C
		\end{pmatrix} \!.
\end{equation}
At $r=r_\eH$ we have  $A(r_\eH)=J(r_\eH)=0$; thus,
\begin{equation}\label{eq:app-c-gf}
	g = \begin{pmatrix}
			0 & \sqrt{AC} \\ 
			\sqrt{AC} & 0
		\end{pmatrix}\!, \qquad
	f = \begin{pmatrix}
			0 &~ \sqrt{c JK} \\ 
			\sqrt{c JK} &~ K - \frac{J}{A}C 
		\end{pmatrix} \!.
\end{equation}
We can apply the theorem on canonical pair forms \cite{Uhlig:1973} on \eqref{eq:app-c-gf}. Thus, the metrics which describe smooth geometries cannot be \emph{both diagonal} at $r=r_\eH$ \emph{unless} $K-\frac{J}{A}C=0$ or equivalently $J/A=K/C$, that is, unless the $(t,r)$-blocks of $g$ and $f$ are conformal at $r_\eH$ (this corresponds to the proper bidiagonality condition in \autoref{prop:propdiag}).
In other words, strictly, Proposition 1 in \cite{Deffayet:2011rh} is stated for the proper bidiagonal metrics, but it is clear that the statement is also valid for the improper bidiagonal case for which $0\ne K - \frac{J}{A}C < \infty$ (crossing condition in \autoref{prop:crossing}), i.e., when the metrics cannot be simultaneously diagonalised at $r_\eH \in \lbrace r\,\vert\,A(r) = 0 \rbrace$ but they are bidiagonal elsewhere by assumption.

\section{Numerical details}
\label{appendix-D}
\begin{figure}[t]
	\centering
	\includegraphics[width=\textwidth]{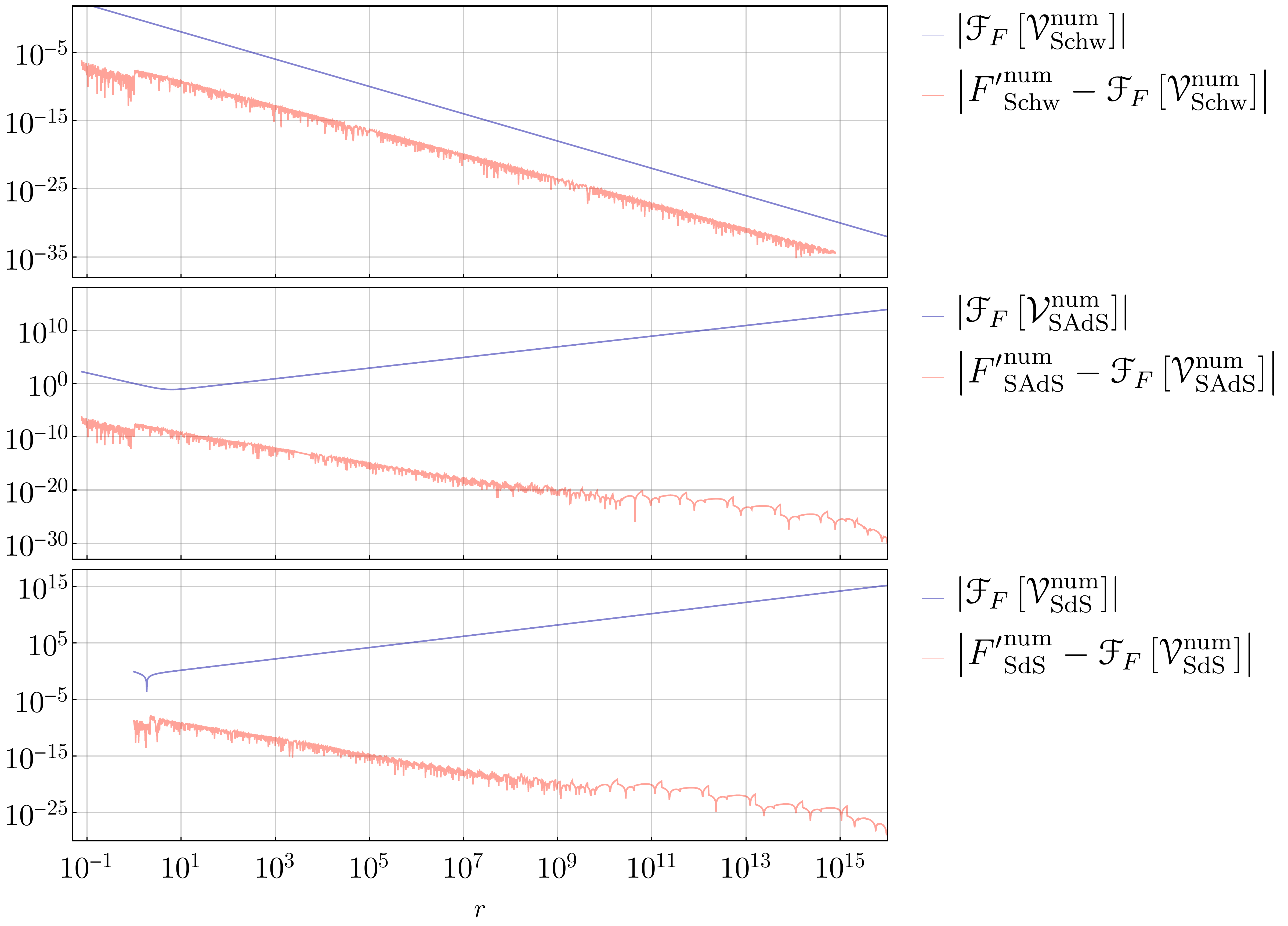}
	\caption{
		Residuals for the GR solutions over the range of integration (red), $r\in [5\times 10^{-2},10^{16}]$, together with $\mathscr{F}_F\left[ \mathscr{V}_{\mathrm{GR}}^{\mathrm{num}} \right]$ (blue). Residuals are always much smaller than $\mathscr{F}_F\left[ \mathscr{V}_{\mathrm{GR}}^{\mathrm{num}} \right]$. The curves are defined only over their integration domain. See text for explanations.
	}
	\label{fig:residualsGRsmall}
\end{figure}
\begin{figure}[t]
	\centering
	\includegraphics[width=1\textwidth]{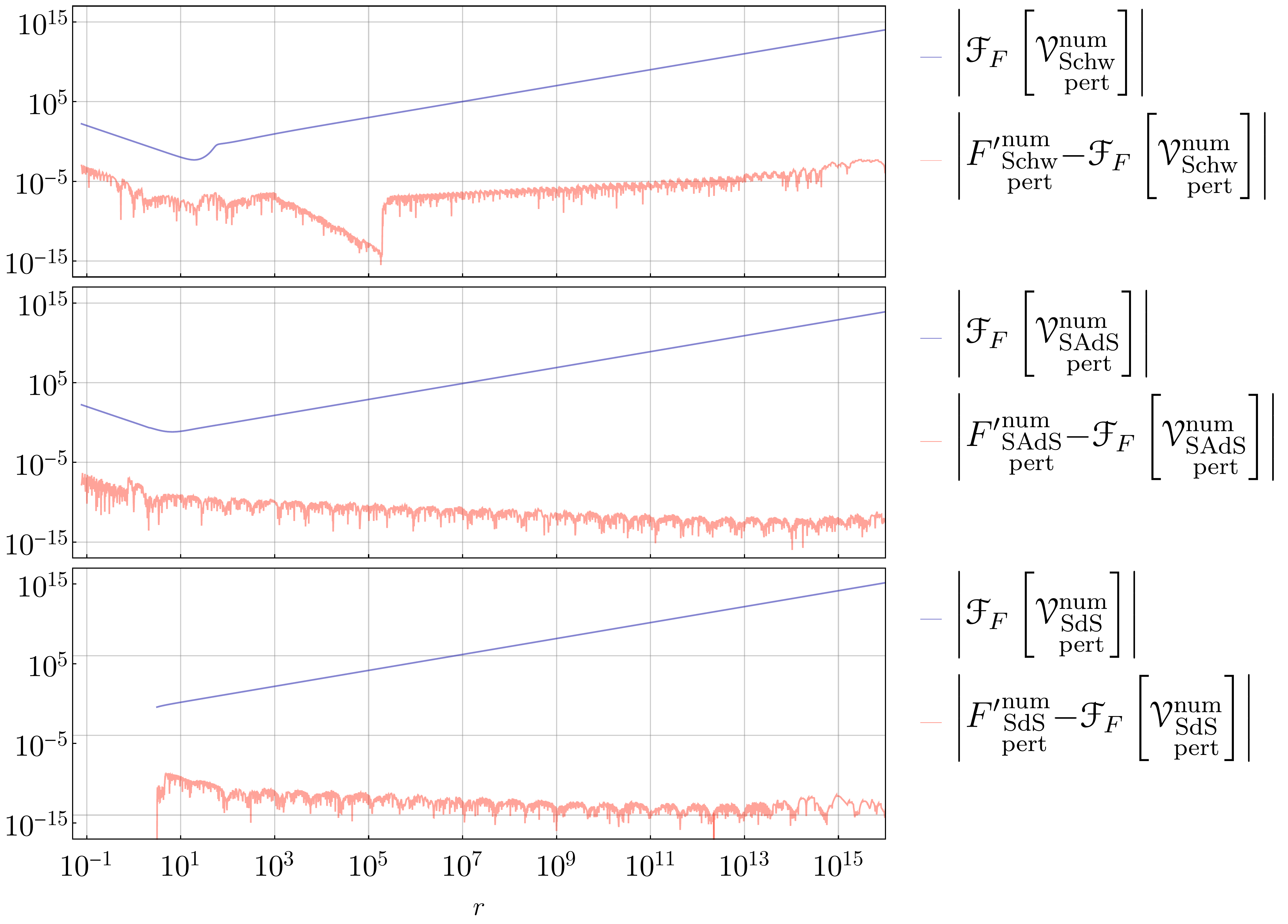}
	\caption{
		Residuals for perturbations around GR solutions over the range of integration (red), $r\in [5\times 10^{-2},10^{16}]$, together with $\mathscr{F}_F\left[ \mathscr{V}_{\mathrm{pert}}^{\mathrm{num}} \right]$ (blue). Residuals are always much smaller than $\mathscr{F}_F\left[ \mathscr{V}_{\mathrm{pert}}^{\mathrm{num}} \right]$. The curves are defined only over their integration domain. See text for explanations.
	}
	\label{fig:residualspertsmall}
\end{figure}
\begin{figure}[t]
	\centering
	\includegraphics[width=\textwidth]{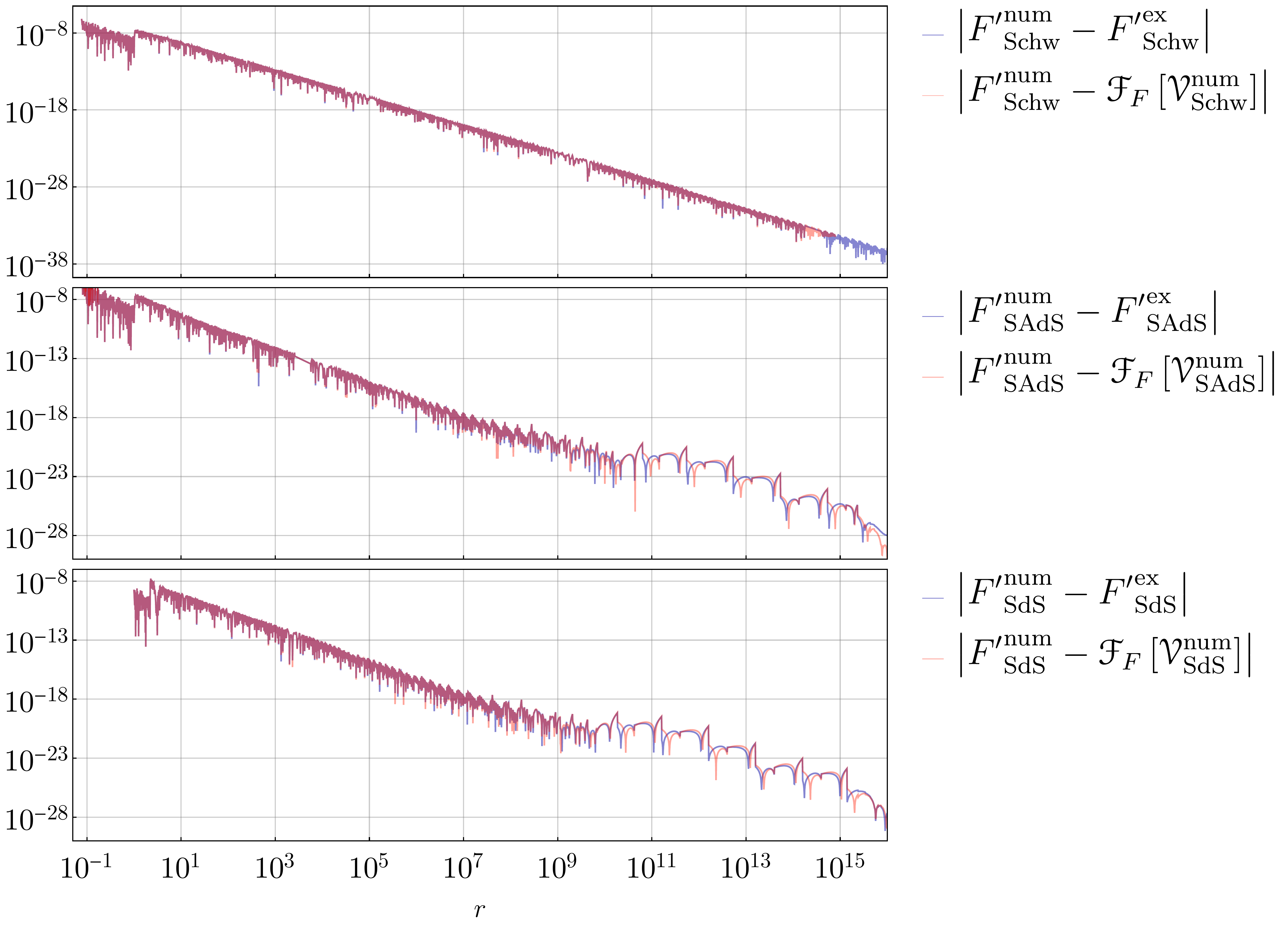}
	\caption{
		Residuals for the GR solutions over the range of integration, $r\in [5\times 10^{-2},10^{16}]$. They are calculated through \eqsref{eq:res1} (red) and \eqref{eq:res2} (blue). The two methods lead to the same result, up to higher order terms. The curves are defined only over their integration domain.
	}
	\label{fig:residualsGR}
\end{figure}
\begin{figure}[t]
	\centering
	\includegraphics[width=\textwidth]{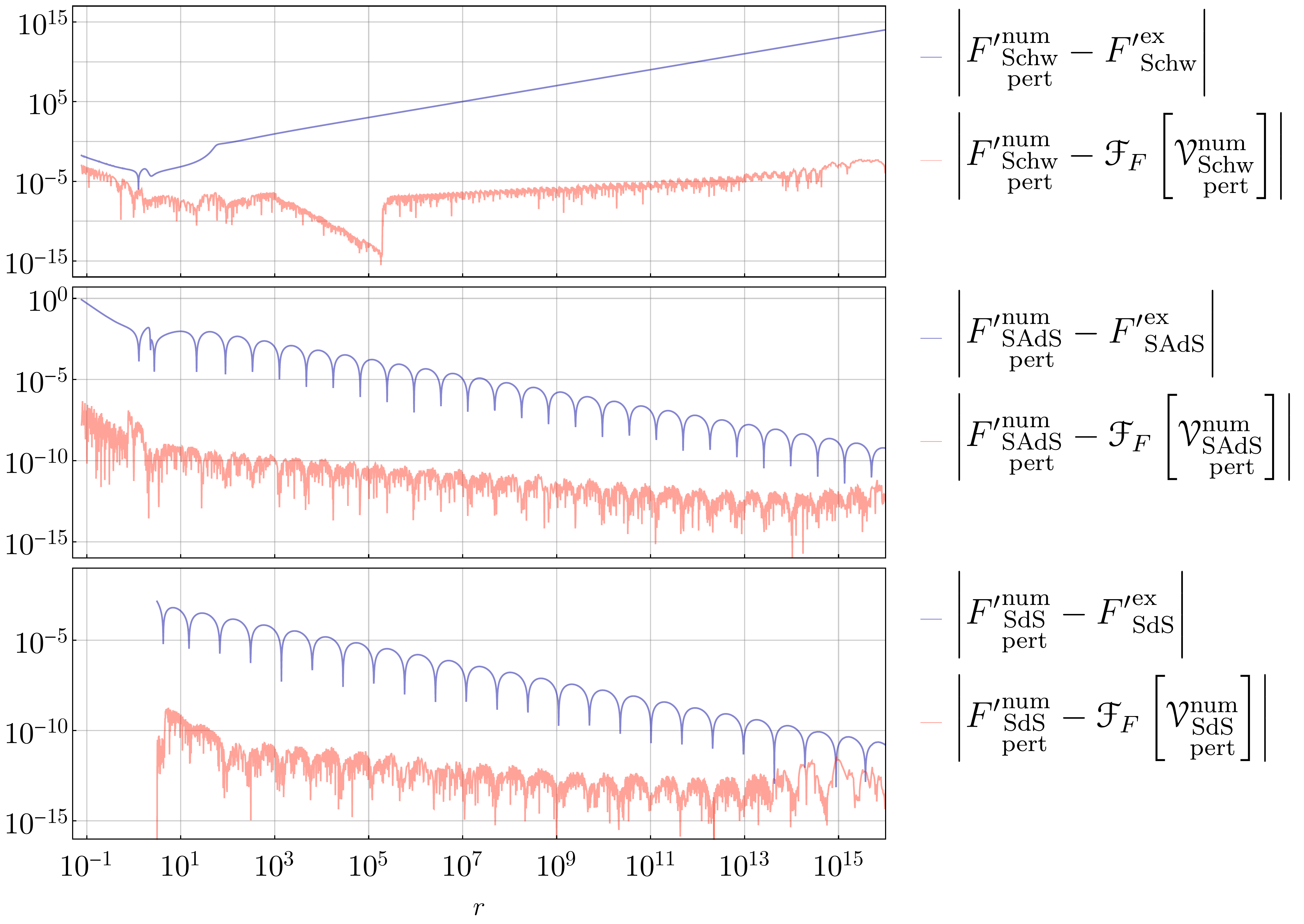}
	\caption{
		Residuals for perturbations around GR solutions over the range of integration $r\in [5\times 10^{-2},10^{16}]$. Residuals are calculated through \eqnref{eq:res1} (red). ${F'}_{\mathrm{pert}}^\mathrm{num}-{F'}_{\mathrm{GR}}^\mathrm{num}$ are calculated through \eqnref{eq:res2} (blue). The two methods lead to very different results, meaning that numerical solutions for perturbations are approximating exact solutions other than GR ones. The curves are defined only over their integration domain.
	}
	\label{fig:residualspert}
\end{figure}
In this appendix, we describe in more detail the numerical procedure used to obtain the solutions in \autoref{subsubsec:results}.  Moreover, we motivate the accuracy of the numerical solutions by studying their residuals, of which analysis is shown in Figures \ref{fig:residualsGRsmall}--\ref{fig:residualspert}.

For solving ODEs numerically, we used Wolfram Mathematica 11 \cite{mathematica}. In particular, \texttt{NDSolve\`{}LSODA} method was used, which is based on LSODA \cite{ref60,Brown:1989:RSM:2613951.2613987}. The numerical integration in Mathematica proceeds in two main steps. The first is the actual integration, which returns the values of the independent variable after each step and the values of the solution at each step, defining a grid of numerical points. The second step is to interpolate these points with a continuous function which is the final output of the numerical integration. This interpolation function should approximate the exact analytical solution.

In checking the validity of our numerical solutions, we first checked that our results are independent from the value of the \texttt{WorkingPrecision} parameter, which sets how many digits of precision are maintained in internal calculations. Ranging precision from 20 to 50, approximation and truncation errors did not affect the results. All the results presented in this paper were obtained with \texttt{WorkingPrecision} 50.

In addition, to verify that our numerical solutions are approximating the analytical solutions, we performed a backward error analysis, in agreement with the procedure described in \cite{Enright:1989:NEI:2613951.2614000,Shampine:2005:SOD:1052070.1052077}. Hence, our analysis was focused on the final interpolated function returned by Mathematica. We used the definition of residuals for a numerical solution as introduced in \cite{Enright:1989:NEI:2613951.2614000}. Suppose we have a first-order system of ODEs,
\begin{equation}
y'(x)=f(x,y(x)),
\end{equation}
with $y,f\in \mathbb{R}^n$ and $n\in \mathbb{N}$ number of ODEs. We solve for $y(x)$ numerically, obtaining its approximation $u(x)$. Then, the residuals of the numerical solution  $u(x)$ are defined as the defect,
\begin{equation}
\label{eq:res1}
r(x)\coloneqq u'(x)-f(x,u(x)).
\end{equation}
Clearly, if the numerical solution was equal to the exact solution, $u(x)=y(x)$, this quantity would be identically zero. We then ask whether the equation,
\begin{equation}
\label{eq:numeq}
u'(x)= f(x,u(x))+r(x),
\end{equation}
is close to the original one or not. This is equivalent to asking if residuals are small compared to $f(x,u(x))$. In \cite{Enright:1989:NEI:2613951.2614000,Shampine:2005:SOD:1052070.1052077}, it is pointed out that, if $u(x)$ is the approximation of $y(x)$ in a given interval of the independent variable $x$, then,
\begin{equation}
\label{eq:res2}
r(x)= u'(x)-y'(x)+\mbox{higher order terms}.
\end{equation}
We will use \eqsref{eq:res1}--\eqref{eq:res2} to show that the obtained solutions are accurate. In the following, we present the analysis only for the $F$ metric function, but the same conclusions hold for all the metric fields.

First, Figures \ref{fig:residualsGRsmall} and \ref{fig:residualspertsmall} show that residuals are extremely small compared to the vector field component $\mathscr{F}_F$ in \eqnref{eq:dynsystem} for all solutions presented in this paper. This means, according to \eqnref{eq:numeq}, that our numerical solutions satisfy equations which are very close to the original equations.

In order to study in more detail the properties of residuals, we consider GR solutions first. Knowing their analytic form, we can directly compute the residuals by using not only \eqnref{eq:res1}, but also \eqnref{eq:res2}. They are plotted in \autoref{fig:residualsGR}, where the red and blue curves are almost coinciding over the whole range of integration. These two curves are the residuals calculated using \eqsref{eq:res1} (red) and \eqref{eq:res2} (blue). Being equal, our numerical GR solutions are approximating the analytical GR solutions.

We can now apply the same analysis to perturbations around GR solutions. The purpose is to show that the results presented in Figures \ref{fig:Schwinstability}--\ref{fig:SdSinstability} are not affected by numerical errors. Suppose that perturbations around GR solutions tend to them for large radii. Then, according to \eqnref{eq:res2}, after some finite $r$ the difference $|{F'}_{\mathrm{pert}}^\mathrm{num}-{F'}_{\mathrm{GR}}^\mathrm{num}|$ must be due only to numerical residuals. Conversely, if this difference is not close to the residuals calculated through \eqnref{eq:res1}, the numerical solution is not approximating an exact GR solution. We can then trust Figures \ref{fig:Schwinstability}--\ref{fig:SdSinstability}, since the difference between the perturbations around GR solutions and GR solutions is not due to numerical errors.

In \autoref{fig:residualspert}, we see that the differences $|{F'}_{\mathrm{pert}}^\mathrm{num}-{F'}_{\mathrm{GR}}^\mathrm{num}|$ and $|{F'}_{\mathrm{pert}}^\mathrm{num}-\mathscr{F}_F[\mathscr{V}_{\mathrm{pert}}^\mathrm{num}]|$ for perturbations of GR solutions differ by many order of magnitudes over almost the whole range of integration. For $r>10^{14}$, the two quantities start to be comparable. This strongly indicates that perturbations around GR solutions are always different from the GR solutions, independent of numerical errors. In addition, if we calculate the residuals by using the numerical points calculated by Mathematica before doing the interpolation, they are always many order of magnitude smaller than $|{F'}_{\mathrm{pert}}^\mathrm{num}-{F'}_{\mathrm{GR}}^\mathrm{num}|$ (calculated numerically on the grid), even for $r>10^{14}$. This shows that the interpolation done by Mathematica is another source of error. Therefore, the fact that the red and blue curves in \autoref{fig:residualspert} are approaching each other for $r> 10^{14}$, is mostly due to an interpolation error and without physical significance. This is also suggested by the fact that residuals in \autoref{fig:residualspert} do not exhibit the clean oscillatory behaviour of $|{F'}_{\mathrm{pert}}^\mathrm{num}-{F'}_{\mathrm{GR}}^\mathrm{num}|$.

\bibliographystyle{apsrev4-1}
\providecommand{\href}[2]{#2}\begingroup\raggedright\endgroup

\end{document}